\pdfoutput=1

\documentclass[a4paper,twocolumn,11pt,accepted=2026-08-11]{quantumarticle}
\pdfoutput=1
\usepackage{units}
\usepackage{amsmath,braket}
\usepackage{amsthm}
\usepackage{amssymb}
\usepackage{graphicx}
\usepackage{color}
\usepackage{xcolor}
\usepackage{bbold}
\usepackage{apptools}
\usepackage{appendix}
\usepackage{hhline}

\definecolor{myurlcolor}{rgb}{0,0,0.7}
\definecolor{myrefcolor}{rgb}{0.1,0,0.9}

\usepackage[
	breaklinks,
	pdftex,
	colorlinks=true, 
	linkcolor=myrefcolor,
	citecolor=myurlcolor,
	urlcolor=myurlcolor
]{hyperref}
\usepackage[numbers,sort&compress]{natbib}

\newcommand{\SMLong}{Appendix}

\newcommand{\gEnt}[0]{fidelity of separability}

\newtheorem{theorem}{Theorem}

\AtAppendix{\counterwithin{theorem}{section}}
\AtAppendix{\counterwithin{lemma}{section}}
\AtAppendix{\counterwithin{fact}{section}}
\AtAppendix{\counterwithin{definition}{section}}

\usepackage[linesnumbered, ruled,vlined]{algorithm2e}

\graphicspath{{./figs/}}

\def\app#1#2{%
  \mathrel{%
    \setbox0=\hbox{$#1\sim$}%
    \setbox2=\hbox{%
      \rlap{\hbox{$#1\propto$}}%
      \lower1.1\ht0\box0%
    }%
    \raise0.25\ht2\box2%
  }%
}

\newtheorem{fact}{\protect\factname}
\ifx\proof\undefined
\newenvironment{proof}[1][\protect\proofname]{\par
	\normalfont\topsep6\p@\@plus6\p@\relax
	\trivlist
	\itemindent\parindent
	\item[\hskip\labelsep\scshape #1]\ignorespaces
}{%
	\endtrivlist\@endpefalse
}
\providecommand{\proofname}{Proof}
\fi
\makeatother

\newcommand{\tr}{\mathrm{tr}}

\newcommand\numberthis{\addtocounter{equation}{1}\tag{\theequation}}

\providecommand{\factname}{Fact}
\providecommand{\theoremname}{Theorem}
\providecommand{\claimname}{Claim}
\providecommand{\lemmaname}{Lemma}
\providecommand{\definitionname}{Definition}
\providecommand{\propositionname}{Proposition}
\providecommand{\corollaryname}{Corollary}
\providecommand{\conjecturename}{Conjecture}

\definecolor{THc}{rgb}{0.9,0.3,0.2}

\newcommand{\revA}[1]{{#1}}

\newtheorem{definition}{\protect\definitionname}

\newtheorem{proposition}{\protect\propositionname}

\usepackage{physics}

\newcommand{\subfigimg}[3][,]{%
	\setbox1=\hbox{\includegraphics[#1]{#3}}
	\leavevmode\rlap{\usebox1}
	\rlap{\hspace*{2pt}\raisebox{\dimexpr\ht1-0.5\baselineskip}{{\bfseries \large\textsf{#2}}}}
	\phantom{\usebox1}
}

\newcommand{\sectionMain}[1]{
\let\oldaddcontentsline\addcontentsline
\renewcommand{\addcontentsline}[3]{}
\section{#1}
\let\addcontentsline\oldaddcontentsline
}

\newcommand{\E}{\mathcal{E}}

\begin{document}
\title{Quantifying mixed-state entanglement via partial transpose and realignment moments}
\author{Poetri Sonya Tarabunga}
\email{poetri.tarabunga@tum.de}
\affiliation{Technical University of Munich, TUM School of Natural Sciences, Physics Department, 85748 Garching, Germany}
\affiliation{Munich Center for Quantum Science and Technology (MCQST), Schellingstr. 4, 80799 München, Germany}

\author{Tobias Haug}
\email{tobias.haug@u.nus.edu}
\affiliation{Quantum Research Center, Technology Innovation Institute, Abu Dhabi, UAE}

\begin{abstract}
    Entanglement plays a crucial role in quantum information science and many-body physics, yet quantifying it in mixed quantum many-body systems has remained a notoriously difficult problem. Here, we introduce  families of quantitative entanglement witnesses, constructed from partial transpose and realignment moments, which provide rigorous bounds on entanglement monotones as well as entanglement dimensionality. Our witnesses can be efficiently measured using SWAP tests or variants of Bell measurements, thus making them directly implementable on current hardware.
    Leveraging our witnesses, we present several novel results on entanglement properties of mixed states, both in quantum information and many-body physics. We develop efficient algorithms to test whether mixed states with bounded entropy have low or high entanglement, which previously was only possible for pure states. We also provide an efficient algorithm to test the Schmidt rank using only two-copy measurements, and the operator Schmidt rank using four-copy measurements.
    Further, our witnesses robustly  certify the quantum circuit depth in the presence of noise, as well as the Schmidt rank of mixed states. Finally, we show that the entanglement phase diagram of Haar random states, quantified by the partial transpose negativity, can be fully established solely by computing our witness,  a result that also applies to any state $4$-design. 
    Our witnesses can also be efficiently computed for matrix product states, thus enabling the characterization of entanglement in extensive many-body systems.
    Finally, we make progress on the entanglement required for quantum cryptography, establishing rigorous limits on pseudoentanglement and pseudorandom density matrices with bounded entropy.
    Our work opens new avenues for quantifying entanglement in large and noisy quantum systems.

\end{abstract}

\maketitle

 \let\oldaddcontentsline\addcontentsline
\renewcommand{\addcontentsline}[3]{}

\section{Introduction}

Entanglement is a key resource in quantum information science, and, at the same time, is of central importance in the characterization of quantum many-body systems~\cite{horodecki2009quantum,amico2008}. 
While the quantitative characterization of entanglement in pure states is well-established through measures such as the entanglement (R\'enyi) entropy~\cite{vidal2000entanglement}, the extension to the case of mixed states presents formidable challenges. 
The advent of modern noisy intermediate-scale quantum (NISQ) devices~\cite{preskill2018quantum,bharti2022noisy}, which are inherently prone to experimental noise, makes the characterization of mixed states particularly critical for the benchmarking and verification of these devices. In particular, characterizing the entanglement properties of noisy quantum systems is essential to gauge the performance of NISQ devices, especially since noise tends to hamper entanglement generation. This urgent need necessitates the development of efficient methods to quantify entanglement in complex noisy systems, a task that is notoriously more  difficult than its pure-state counterpart. 

One of the most common approaches for quantifying mixed-state entanglement is based on the celebrated positive partial transpose (PPT) criterion, which states that any separable state has a partial transpose with nonnegative eigenvalues~\cite{Peres1996,Horodecki1996}. This criterion leads to the PT negativity~\cite{vidal2002computable,plenio2005negativity} (also known as the logarithmic negativity), a widely used and computable entanglement monotone which essentially measures the negative eigenvalues of the partial transpose. While the PT negativity has found extensive applications in analytical studies~\cite{calabrese2012negativityqft,calabrese2014finitetemperature,lee2013entanglement,castelnovo2013negativity,hart2018entanglement,eisler2014entanglement,wen2016topological,ruggiero2016negativity}, it typically requires intricate calculations to obtain the full spectrum of the partial transpose. Furthermore, experimentally measuring the PT negativity generally necessitates full-state tomography, thereby severely limiting its applicability in real-world experiments.

To overcome this challenge, previous works have focused on the simpler task of witnessing~\cite{elben2020mixed,neven2021symmetryresolved,yu2021optimal,liu2022detecting}, which aims to determine whether a given state is entangled or not. 
It has been shown that this detection task can be performed by considering only low-order PT moments~\cite{elben2020mixed,neven2021symmetryresolved,yu2021optimal,liu2022detecting} and other positive maps~\cite{mallick2025higherdimensional,mukherjee2025efficientdetectiongenuinemultipartite,Zhang2022,Aggarwal2024,Wang2024}, which are easier to measure in experiments. While such advancements offer a more practical means for entanglement detection, a mere binary detection provides a very limited characterization of entanglement. As such, they have not been adopted in the study of entanglement in quantum many-body systems. Ultimately, the crucial task of quantitatively and efficiently assessing the degree of entanglement in large, mixed systems remains a major challenge.

While direct quantification using entanglement measures remains an unresolved problem, 
we show that quantitative witnessing, i.e. a method that provides not only detection but also quantitative information about entanglement~\cite{eisert2007quantitative}, can be performed efficiently. In this work, we introduce a family of quantitative entanglement witnesses for mixed states called the $p_\alpha$-negativity. The $p_\alpha$-negativity is monotonic,  and the PT negativity is recovered when $\alpha=1$. While the $p_\alpha$-negativity is not an entanglement monotone for $\alpha>1$, it crucially serves as a quantitative witness for mixed-state entanglement by providing a lower bound on the PT negativity. This property is paramount, as it provides not just a yes/no answer to the presence of entanglement but offers valuable information about its quantitative amount. We further demonstrate how the exploitation of state symmetries can significantly enhance the efficacy of these witnessing capabilities, offering a substantial improvement over previous symmetry-resolved approaches~\cite{neven2021symmetryresolved,rath2023entanglement}. Complementing the $p_\alpha$-negativity, we also introduce the $r_\alpha$-negativity, a related family of witnesses derived from the computable cross norm or realignment (CCNR) criterion~\cite{chen2003ccnr}. We show that both families provide lower bounds for the robustness of entanglement~\cite{vidal1999robustness}, an entanglement monotone, and the Schmidt number~\cite{terhal2000schmidt}, which serves as a measure of entanglement dimensionality. We focus on the case of $\alpha=4$ for both families, which we show to be efficiently measurable on quantum computers and numerically accessible via matrix product states. In this sense, our witnesses play a similar role to entanglement R\'enyi entropies, which have been instrumental for studying entanglement in pure states. Indeed, the latter are efficiently measurable for integer orders $\alpha \geq 2$ and provide lower bounds on the von Neumann entanglement entropy. Notably, both the $p_\alpha$-negativities and $r_\alpha$-negativities collapse to the entanglement R\'enyi entropies when the state is a pure state. \revA{We note that the $p_\alpha$-negativity and $r_\alpha$-negativity are examples of 
\emph{nonlinear} entanglement witnesses~\cite{guhne2006nonlinear,rico2024entanglement,rico2025state}, 
i.e. nonlinear functionals of the density matrix, in contrast to the more standard linear 
witnesses~\cite{guhne2009entanglement,chruciski2014entanglement,seong2026practical,szczepaniak2025entanglement,Bae2020}. For $\alpha=4$, both witnesses can be measured efficiently without full-state tomography, as is the case for conventional linear entanglement witnesses. }

We remark that various R\'enyi generalizations of the PT negativity have been previously introduced in the literature~\cite{calabrese2012negativityqft,lee2013entanglement}, and are often employed in numerical studies~\cite{chung2014entanglement,ding2024negativity,wukaihsin2020negativity,wangfohong2025negativity,wangfohong2025negativity,wybo2021dynamics}, when the complete PT spectrum is not available. However, it remained unclear how these quantities genuinely reflect the entanglement properties of the system, beyond simply serving as a “proxy” for the PT negativity, with a lack of rigorous justification. In fact, some of these quantities are not even well-defined in certain cases, such as for Werner states (see Appendix~\ref{sec:werner}). While a recent experiment~\cite{shaw2024benchmarking} introduced and measured a quantity that lower bounds the PT negativity, it relies on assumptions of closeness to, and knowledge on entanglement, of the target pure state, limiting its applicability beyond classically simulable regimes.

Our work directly addresses these shortcomings, since our $p_\alpha$-negativities have a clearly defined role as quantitative witnesses and are completely agnostic about the state in question. This allows us to derive several significant results on entanglement properties of mixed states, including concrete applications in entanglement testing, pseudoentanglement, circuit depth certification, and entanglement phase diagram in Haar random and stabilizer states.

Regarding entanglement testing, we introduce an algorithm to test whether a mixed state exhibits $O(\log n)$ or $\omega (\log n)$ entanglement, where $n$ is the number of qubits. The algorithm is efficient for weakly mixed states with $2$-R\'enyi entropy $S_2=O(\log n)$. Furthermore, we elucidate the entanglement properties of pseudorandom density matrix (PRDM)~\cite{bansal2024pseudorandomdensitymatrices}, defined as an ensemble of states that is computationally indistinguishable from Haar random states by any efficient quantum algorithm. In particular, we show that any PRDM with low entropy $S_2=O(\log n)$ must have PT negativity $\E=\omega(\log n)$.

\begin{figure}[htbp]
\centering
\includegraphics[width=1.05\linewidth]{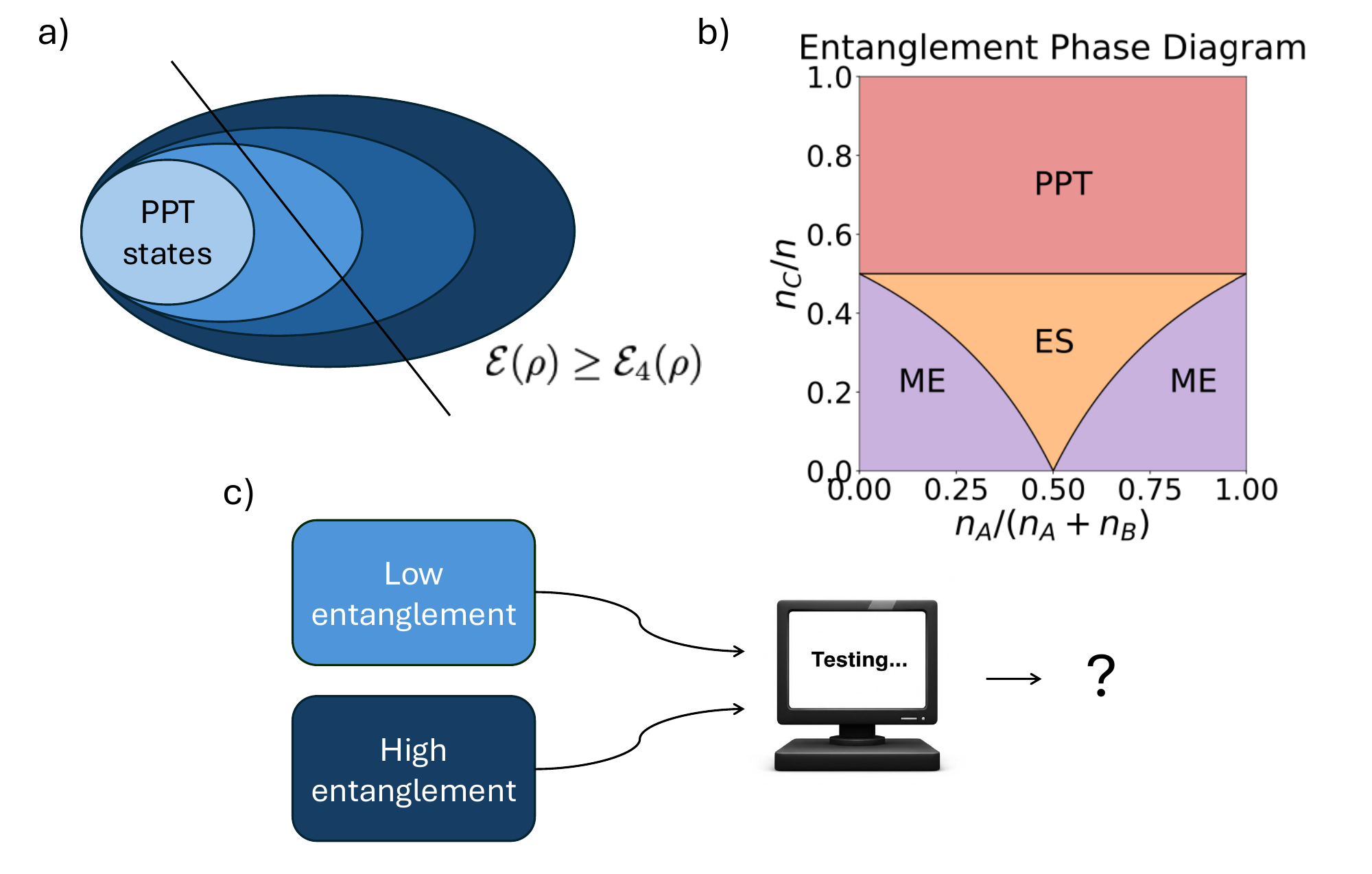}
\caption{a) Schematic representation of $n$-qubit state space in terms of the PT negativity $\E(\rho)$.  Darker regions indicate higher values of $\E(\rho)$. Measuring the $p_4$-negativity, $\E_4(\rho)$, provides a lower bound on $\E(\rho)$, thereby restricting the possible values of $\E(\rho)$ within the state space. b) Entanglement phase diagram of Haar random states and stabilizer states for partition sizes $n_A$, $n_B$, $n_C$. The entanglement is between subsystems $A$ and $B$, while subsystem $C$ is traced out. There are three different phases: positive partial transpose (PPT), entanglement saturation (ES) and maximally entangled (ME). (c)  Illustration of entanglement testing, where one needs to determine whether a given state has low or high entanglement.}
\label{fig:summary_results}
\end{figure}

These findings have direct implications for quantum cryptography and pseudoentanglement~\cite{ji2018pseudorandom}. We demonstrate a fundamental constraint on pseudoentanglement, which describes states with low entanglement $g(n)$ that are computationally indistinguishable from states with high entanglement $f(n)$ by any efficient quantum algorithm~\cite{bouland2022quantum,haug2023pseudorandom,bansal2024pseudorandomdensitymatrices}. We show that for mixed states with low entropy $S_2=O(\log n)$, the pseudoentanglement gap is fundamentally limited to $f(n)=\Theta(n)$ vs $g(n)=\omega(\log n)$, mimicking the constraints known for pure states. This result also establishes that the optimal pseudoentanglement gap is achieved only for highly mixed states with $S_2=\omega(\log n)$, where the gap is $f(n)=\Theta(n)$ vs $g(n)=0$~\cite{bansal2024pseudorandomdensitymatrices}. This implies that entropy is a necessary resource for securely hiding information about entanglement from eavesdroppers. Our work completes the understanding of the complexity of testing and pseudoentanglement as a function of the state's entropy, with a complete summary provided in Tab.~\ref{tab:testing} and Tab.~\ref{tab:pseudoresource}, respectively. 

Further, we provide efficient algorithms to robustly certify the circuit depth of weakly mixed states, which so far was only possible for perfectly pure states. We also show how to efficiently test the Schmidt rank as well as the operator Schmidt rank.

In the context of many-body physics, understanding the generic entanglement properties of quantum states is a fundamental task. Our $p_4$-negativity witness leads to a remarkable simplification in characterizing the entanglement phase diagram of Haar random states. Unlike previous cumbersome methods that necessitated calculating all PT moments and their subsequent analytical continuation to extract the PT negativity~\cite{shapourian2021entanglement}, our approach only requires computing  the $p_4$-negativity. As a corollary, we establish that any state $4$-design, a set of quantum states that mimics Haar random states up to fourth moments, exhibits the identical entanglement phase diagram as Haar random states. Surprisingly, we also discover that stabilizer states, a classically simulable class of quantum states~\cite{gottesman1997stabilizer,gottesman1998heisenberg,gottesman1998theory,aaronson2004improved}, share the same entanglement phase diagram as Haar random states, despite the fact that they do not form a $4$-design. Furthermore, we investigate the impact of noise on the entanglement of pure Haar random states, showing that the $p_4$-negativity effectively captures their entanglement properties even in the presence of noise. Notably, our findings reveal a surprising resilience: volume-law entanglement persists for sufficiently large subsystems, even when subjected to exponentially strong noise, highlighting the robustness of entanglement in generic quantum systems.

\begin{table}[htbp]\centering
\begin{tabular}{| l|l| }
    \hline    
     Entropy &  Copies  \\
    \hline\hline
    $S_2=0$  & $O(\text{poly}(n))$~\cite{ekert2002direct}\\\hline
 $S_2=O(\log n)$ & $O(\text{poly}(n))$ [this work] \\\hline
 $S_2=\omega(\log n)$ & $2^{\omega(\log n)}$~\cite{bansal2024pseudorandomdensitymatrices}\\ \hline
\end{tabular}
\caption{Number of copies of state $\rho$ needed to test whether mixed state $\rho$ has $O(\log n)$ or $\omega(\log n)$ entanglement, as characterized by PT negativity $\E(\rho)$ and \gEnt{} $D_\text{F}(\rho)$. Complexity depends on $2$-R\'enyi entropy $S_2=-\ln\text{tr}(\rho^2)$.
}
\label{tab:testing}
\end{table}

\begin{table}[htbp]\centering
\begin{tabular}{| l|l| }
    \hline    
    Entropy &  $f(n)$ vs $g(n)$  \\
    \hline\hline
    $S_2=0$  & $\Theta(n)$ vs $\omega(\log n)$~\cite{bouland2022quantum}\\\hline
 $S_2=O(\log n)$ & $\Theta(n)$ vs $\omega(\log n)$ [this work] \\\hline
 $S_2=\omega(\log n)$ & $\Theta(n)$ vs $0$~\cite{bansal2024pseudorandomdensitymatrices}\\ \hline
\end{tabular}
\caption{Pseudoentanglement gap $f(n)$ vs $g(n)$ between high-entanglement and low-entanglement state ensembles, depending on $2$-R\'enyi entropy $S_2=-\ln\text{tr}(\rho^2)$. Entanglement is characterized by PT negativity $\E(\rho)$ and \gEnt{} $D_\text{F}(\rho)$.
}
\label{tab:pseudoresource}
\end{table}

We summarize our main results in Fig.~\ref{fig:summary_results}. The rest of the paper is structured as follows. In Sec.~\ref{sec:prelim} we provide some preliminaries. In Sec.~\ref{sec:witness} we introduce the families of quantitative entanglement witnesses. In Sec.~\ref{sec:witness}, we show how our families of witness quantify entanglement dimensionality. In Sec.~\ref{sec:ccnr_osee} we discuss the relation between the $r_\alpha$-negativity and the operator space entanglement entropy. In Sec.~\ref{sec:ent_testing} we introduce an algorithm for entanglement testing in mixed states, as well as the Schmidt rank and operator Schmidt rank testing. In Sec.~\ref{sec:depth_certification} we propose a certification scheme for the minimal depth to prepare a given state. In Sec.~\ref{sec:prdm_pseudoentanglement} we discuss the entanglement properties of pseudorandom density matrices and pseudoentanglement. In Sec.~\ref{sec:entanglement_phase_diagram} we discuss the entanglement phase diagram of Haar random and stabilizer states. 
In Sec.~\ref{sec:noise_haar} we discuss the effects of noise on the entanglement in Haar random states.
In Sec.~\ref{sec:mps_mpo} we introduce algorithms to compute our witnesses using matrix product states and matrix product operators. In Sec.~\ref{sec:num_results} we provide numerical results on our witnesses in mixed many-body systems in various mixed-state settings. Finally, in Sec.~\ref{sec:discussion} we provide a discussion of our results. A number of technical points and proofs are reported in the appendices.
 
\section{Preliminaries} \label{sec:prelim}
We consider a system of $n$ qubits associated with the Hilbert space $\mathcal{H}_n=\bigotimes_{j=1}^n\mathcal{H}_j$, where $\mathcal{H}_j\simeq \mathbb{C}^2$. For a pure state of $n$ qubits $\ket{\psi}$, the bipartite entanglement between complementary subsystems $A$ and $B$ can be quantified by the entanglement R\'enyi entropy
\begin{equation}
    S_\alpha(\ket{\psi}) = \frac{1}{1-\alpha}\ln \text{tr}(\rho_A^\alpha),
\end{equation}
where $\rho_A=\text{tr}_B(\ket{\psi}\bra{\psi})$ is the reduced density matrix. A pure state has zero entanglement entropy iff it is an unentangled product state $\ket{\psi}=\ket{\psi_A}\otimes\ket{\psi_B}$. $S_\alpha(\ket{\psi})$ has the monotonicity property $S_\alpha(\ket{\psi})\geq S_\beta(\ket{\psi})$ for $\alpha\leq\beta$.

For a general mixed quantum state of $n$ qubits $\rho$, the bipartite entanglement between complementary subsystems $A$ and $B$ can be quantified from the partial transpose 
\begin{equation}
    \rho^{\Gamma} = \Gamma(\rho)\,,
\end{equation}
where $\Gamma$ denotes the partial transposition on subsystem $B$. More precisely, it is a linear map that acts as $\Gamma(X_A \otimes Y_B) = X_A \otimes Y_B^T$, where the superscript $T$ denotes matrix transpose. In other words, the matrix elements of the partial tranpose are given by
\begin{equation}
   \bra{a,b'} \rho^{\Gamma}  \ket{a',b} = \bra{a,b} \rho \ket{a',b'}, 
\end{equation}
where $a$ and $b$ denote computational basis states in subsystems $A$ and $B$, respectively. A necessary condition for a separable state is that their PT is always positive-semidefinite (Peres-Horodecki criterion)~\cite{Peres1996,Horodecki1996}. That is, a state with a non-positive PT (NPT state) is entangled. The converse is not true as there
exists PPT entangled states, also known as bound entangled states. The entanglement monotone related
to the PPT condition is the PT negativity (also known as logarithmic negativity), which is defined by~\cite{vidal2002computable}
\begin{equation} \label{eq:negativity}
    \E(\rho) = \ln \lVert \rho^\Gamma \rVert_1 = \ln \sum_i \lvert \lambda_i \rvert\,,
\end{equation}
where $\lambda_i$ are eigenvalues  of $\rho^{\Gamma}$. Note that the trace norm of an operator $X$ is given by $\lVert X \rVert_1=\sqrt{X^\dagger X}$. Hereafter, we will drop the dependence on $\rho$ when the state in question is clear from the context.

The computation of PT negativity generally requires full diagonalization of $\rho^{\Gamma}$, which is exponentially costly. To overcome this challenge, it is useful to analyze the more tractable (integer) moments of the partial transpose
\begin{equation}
    p_\alpha(\rho) = \text{tr}((\rho^{\Gamma})^\alpha)
\end{equation}
for integer $\alpha=1,2,\dots$. 
Note that $p_2=\text{tr}(\rho^2)$ is equivalent to the purity of $\rho$, with $2$-R\'enyi entropy $S_2=-\ln(p_2)$. 
These PT moments can be utilized to derive necessary and
sufficient conditions of entanglement in mixed states~\cite{elben2020mixed,neven2021symmetryresolved,yu2021optimal}. For example, the so-called $p_3$-PPT condition reads~\cite{elben2020mixed}
\begin{equation}
    \rho \in \text{SEP}\implies p_3 \geq p_2^2.
\end{equation}
Any state violating this condition is NPT and therefore
entangled. One can define a quantity related to the $p_3$-PPT condition~\cite{carrasco2024entanglementphasediagram}
\begin{equation}
    \Tilde{\E}_3(\rho) = \frac{1}{2}\ln\left(\frac{p_2^2}{p_3}\right),
\end{equation}
such that $\Tilde{\E}_3(\rho)>0$ implies that the state is entangled. Thus, $\Tilde{\E}_3(\rho)$ is an entanglement witness for mixed states. While the $p_3$-PPT condition is a practical tool to detect entanglement in mixed states, it is not known whether it is quantitatively (e.g. in terms of inequalities) related to the PT negativity. 

An alternative approach to detect entanglement in mixed states is via the CCNR criterion~\cite{chen2003ccnr}. We define the realignment matrix $R_\rho$ as
\begin{equation}
   \bra{a,a'} R_\rho \ket{b,b'} = \bra{a,b} \rho \ket{a',b'}, 
\end{equation}
which provides a necessary condition for separability that $\lVert R_\rho \rVert_1 \leq 1$. The CCNR negativity is defined by~\cite{yin2023universal,milekhin2024computablecrossnormtensor,yin2023mixedstate} 
\begin{equation}
    \mathcal{C}(\rho) = \ln \lVert R_\rho \rVert_1 = \ln \sum_i \lvert \lambda_i \rvert,
\end{equation}
where $\lambda_i$ are singular values of $R_\rho$. Unlike the PT negativity,  $\mathcal{C}(\rho)$ is not a proper entanglement measure, since it can become negative. Nevertheless, by the CCNR criterion, $\mathcal{C}(\rho)>0$ implies that the state is entangled, and thus it serves as an entanglement witness. We further denote the moments of the realignment matrix as
\begin{equation}
    r_\alpha(\rho) = \text{tr}((R_\rho^\dagger R_\rho)^{\alpha/2}) = \sum_i |\lambda_i|^{\alpha}.
\end{equation}

We will also consider two other entanglement monotones. First, the \gEnt{} is defined as~\cite{shapira2006groverian,streltsov2010linking,chen2014comparison}
\begin{equation} 
    D_\text{F}(\rho)=\min_{\sigma\in \text{SEP}}-\ln \mathcal{F}(\rho,\sigma)
\end{equation}
where we have the Uhlmann fidelity $\mathcal{F}(\rho,\sigma)=\text{tr}(\sqrt{\rho\sigma})^2$~\cite{jozsa1994fidelity}. 

Finally, the robustness of entanglement~\cite{vidal1999robustness} is given by
\begin{equation}\label{eq:robustness}
\mathcal{R}(\rho)=\min s \;\text{s.t} \; \rho = (1+s)\rho^+_s -s\rho_s^-, \quad\rho_s^+,\rho_s^- \in \text{SEP}\,.
\end{equation}

\section{Quantitative entanglement witnesses} \label{sec:witness}
In this section, we formally introduce the entanglement witnesses based on the PT moments and realignment moments. We also define the symmetry-resolved witnesses in the presence of a global symmetry.
\subsection{From the PT moments} \label{sec:pt_witness}
For an $n$-qubit state $\rho$, we define our entanglement witness, $p_\alpha$-negativity, as
\begin{align*}\label{eq:pn_neg}
    \E_\alpha(\rho) &= \frac{1}{2-\alpha}\ln\left(\Tilde{p}_\alpha\right)+ \frac{1-\alpha}{2-\alpha} \ln\left(p_2\right) \\
    &=\frac{1}{2-\alpha}[\ln(\Tilde{p}_\alpha)+(\alpha-1) S_2(\rho)]\,,\numberthis
\end{align*}
where we define $\Tilde{p}_\alpha(\rho) = \text{tr}((\lvert\rho^{\Gamma}\rvert)^\alpha)$. Note that $\Tilde{p}_\alpha=p_\alpha$ for even integer $\alpha$. $\E_\alpha(\rho)$ is an entanglement witness since $\E_\alpha(\rho)>0$ implies that the state is entangled. To see this, let $q_i=\lambda_i^2/p_2$, which forms a probability distribution since $q_i>0$ and $\sum_i q_i=1$. We will consider the (classical) $\alpha$-R\'enyi entropy
\begin{equation}
    H_\alpha(\{q\}) = \frac{1}{1-\alpha} \ln (\sum_{i}q_i^\alpha),
\end{equation}
which is known to satisfy the monotonicity $H_a\geq H_b$ for $a<b$. Since we can write $\E_\alpha(\rho)=\frac{1}{2}(H_{\alpha/2}(\{q\})-S_2(\rho))$, it immediately follows that the $p_\alpha$-negativity also satisfies the monotonicity $\E_\alpha \geq \E_\beta$ for $\alpha<\beta$. Importantly, this implies that
\begin{equation} \label{eq:bound_e}
    \E_\alpha(\rho) \leq \E(\rho)\,,
\end{equation}
for $\alpha\geq 1$, since $\E(\rho)\equiv\E_1(\rho)$.
Thus, $\E_\alpha(\rho)>0$ implies $\E(\rho)>0$, which in turn guarantees that the state is entangled. From the derivation, we  can also deduce that the inequality is saturated iff the distribution $q_i$ is flat (equivalently, when the absolute values of the spectrum $|\lambda_i|$ is flat). For pure states $\ket{\psi}$,  $\E_\alpha(\ket{\psi})=S_{\alpha/2}(\ket{\psi})$ for any $\alpha$, where $S_\alpha$ is the $\alpha$-R\'enyi entanglement entropy. As such, the $p_\alpha$-negativities can be viewed as a generalization of entanglement R\'enyi entropies to mixed states, inheriting some of their nice properties.

The $p_\alpha$-negativity has the following properties: i) Invariant under local unitaries $U_A \otimes U_B$, i.e. $\E_\alpha(U_A \otimes U_B \rho U_A^\dagger \otimes U_B^\dagger)=\E_\alpha(\rho)$,
ii) Additive, i.e. $\E_\alpha(\rho\otimes\sigma)=\E_\alpha(\rho)+\E_\alpha(\sigma)$.
iii) $-\frac{1}{2}S_2(\rho)\leq\E_\alpha(\rho)\leq \frac{1}{2}(n\ln2-S_2(\rho))$. The lower bound follows from $\E_\alpha(\rho)=\frac{1}{2}(H_{\alpha/2}(\{q\})-S_2(\rho))$ and $H_{\alpha/2}(\{q\})\geq0$, while the upper bound is shown in Appendix~\ref{sec:upper_bounds_neg}. 
Note that $\E_\alpha$ is not an entanglement monotone~\cite{vidal2000entanglement} because it can be non-positive for entangled states.

Importantly, the $p_\alpha$-negativity goes beyond simply determining the presence of entanglement. Indeed, as shown in Eq.~\eqref{eq:bound_e}, $\E_\alpha$ provides a lower bound on the PT negativity for $\alpha\geq1$, thus classifying it as a quantitative entanglement witness~\cite{eisert2007quantitative}. Namely, the value of $\E_\alpha$ offers quantitative information regarding the amount of entanglement present in the state. 
 Therefore, $\E_\alpha$  allows one not only to detect the entanglement but also to infer the quantitative value of the entanglement, which is useful to answer the
question of how useful a given state is, say, to perform certain quantum information tasks. Furthermore, we derive the inequality 
\begin{equation} \label{eq:neg_rob}
    \E(\rho) \leq \ln(2\mathcal{R}(\rho)+1)
\end{equation}
in Appendix~\ref{sec:neg_rob}. Therefore, our witnesses also provide a lower bound to the robustness of entanglement.

The $p_\alpha$-negativity for $\alpha<1$ is not an entanglement witness since it is always non-negative. Nevertheless, it remains useful to provide upper bounds for the PT negativity (see Appendix~\ref{sec:upper_bounds_neg}). Moreover, $\E_\alpha(\rho)=0$ for $\alpha<1$ implies $\E(\rho)=0$, thus it provides a sufficient (but not necessary) condition for separability.

We will be particularly interested in the $p_4$-negativity
\begin{equation}\label{eq:e4}
    \E_4(\rho) = \frac{1}{2}\ln\left(\frac{p_2^3}{p_4}\right)=\frac{1}{2}[-\ln(p_4)-3 S_2(\rho)]\,.
\end{equation}
For pure states, we have $\E_4=S_2(\rho_A)$, and the inequality in Eq.~\eqref{eq:bound_e} becomes $S_2(\rho_A)\leq S_{1/2}(\rho_A)$, which follows from the hierarchy of $\alpha$-R\'enyi entanglement entropy. We also note that $\Tilde{\E}_3=S_3(\rho_A)$ for pure states, which also provides a lower bound to $S_{1/2}(\rho_A)$. This lower bound is weaker than $S_2(\rho_A)$, since $S_2(\rho_A)\geq S_3(\rho_A)$. However, $\Tilde{\E}_3$ cannot provide a lower bound to $\E$ in the general case, since $p_3$ can be zero or negative~\cite{elben2020mixed}, in which case $\Tilde{\E}_3$ is not well-defined.

Notably, $\E_4$ can be measured efficiently without reconstructing $\rho$. Indeed, $p_4$ can be measured by a quantum circuit composed of $O(n)$ control-SWAP gates that acts on $4$ copies of $\rho$~\cite{carteret2005noiseless,gray2018machinelearning}. See Appendix~\ref{sec:efficientWitness} for details. Moreover, $p_2=\tr(\rho^2)$ can be measured efficiently via SWAP test or Bell measurements. The joint preparation of $4$ copies of $\rho$ may still be challenging in current experiments, and an alternative measurement scheme is through classical shadow~\cite{huang2020predicting} or randomized measurements~\cite{elben2019statistical,zhou2020singlecopies}. While these schemes generally require exponential number of samples, they have been applied to experimentally detect entanglement~\cite{elben2020mixed,neven2021symmetryresolved,brydges2019probing}.

We remark that the difference between $\Tilde{\E}_3$ and $\E_4$ is given by
\begin{equation}
    \E_4 - \Tilde{\E}_3 = \frac{1}{2}\ln \Tilde{r}_2,
\end{equation}
where
\begin{equation} \label{eq:r2}
    \Tilde{r}_2 = \frac{p_2 p_3}{p_4},
\end{equation}
which is studied in~\cite{carrasco2024entanglementphasediagram} as a tool to distinguish entanglement phase diagram in random states.

\subsection{From the realignment moments} \label{sec:ccnr_witness}
We also define entanglement witnesses based on the CCNR negativity $\mathcal{C}(\rho)$  using the realignment matrix moments. 
We define the $r_\alpha$-negativity as
\begin{equation}\label{eq:rn_neg}
\begin{split}
    \mathcal{C}_\alpha(\rho) &= \frac{1}{2-\alpha}\ln\left(r_\alpha\right)+ \frac{1-\alpha}{2-\alpha} \ln\left(r_2\right)\\ &=\frac{1}{2-\alpha}[\ln(r_\alpha)+(\alpha-1) S_2(\rho)]\,,
\end{split}
\end{equation}
Using similar argument for the $p_\alpha$-negativity, by replacing $p_\alpha$ with $r_\alpha$, one can show that $\mathcal{C}_\alpha$ satisfies the monotonicity $\mathcal{C}_\alpha \geq \mathcal{C}_\beta$ for $\alpha<\beta$, which implies
\begin{equation} \label{eq:bound_c}
    \mathcal{C}_\alpha(\rho) \leq \mathcal{C}(\rho),
\end{equation}
where $\mathcal{C}(\rho)\equiv\mathcal{C}_1(\rho)$. Although $\mathcal{C}(\rho)$ is not an entanglement monotone, it provides a lower bound to the robustness of entanglement via
\begin{equation}
    \mathcal{C}(\rho) \leq \ln(2\mathcal{R}(\rho)+1),
\end{equation}
as shown in Appendix~\ref{sec:neg_rob}. Therefore, $\mathcal{C}_\alpha(\rho)$ for $\alpha\geq1$ also serves as a quantitative entanglement witness, as it provides a lower bound to the robustness of entanglement.

In particular, the $r_4$-negativity reads
\begin{equation}\label{eq:r4}
    \mathcal{C}_4(\rho) = \frac{1}{2}\ln\left(\frac{r_2^3}{r_4}\right)=\frac{1}{2}[-\ln(r_4)-3 S_2(\rho)]\,,
\end{equation}
where $\mathcal{C}_4(\rho)>0$ implies that the state is entangled.  We can give bounds $-\frac{1}{2}S_2(\rho)\leq\mathcal{C}_4(\rho)\leq \min(n_A,n_B)\ln2-\frac{1}{2}S_2(\rho)$, where $n_{A(B)}$ is the number of qubits in $A(B)$. (see Appendix~\ref{sec:upper_bounds_neg} for the upper bound). The witness $\mathcal{C}_4$ can be measured by a scheme similar to that for $\E_4$~\cite{liu2022detecting}, without full-state tomography.  As detailed in Appendix~\ref{sec:CCNRmeas}, we introduce a new measurement scheme which employs a constant-depth circuit of Clifford gates, offering a significant advantage in terms of experimental feasibility. It can also be measured via randomized measurements as demonstrated in previous experiments~\cite{rath2023entanglement,brydges2019probing}.

\subsection{Symmetry-resolved entanglement witnesses} \label{sec:sr_witness}
We will now discuss how to exploit symmetry to enhance the bounds from the witnesses. We will first focus on the witnesses from the PT moments before commenting on the realignment moments.

In the presence of a global symmetry, the density matrix $\rho$ can be split into different charge sectors. For concreteness, we consider a global $U(1)$ symmetry where the state commutes with the total number operator $Q=\sum_i (I-Z_i)/2$. Such a state has a block-diagonal form $\rho=\bigoplus_q \rho_{(q)}$, where $\rho_{(q)}$ is a block corresponding to the charge $q$. In a bipartite system, we can write $Q=Q_A+Q_B$, where $Q_{A(B)}=\sum_{i\in A(B)} (I-Z_i)/2$.  

In this case, the partial transpose commutes with $Q'=Q_A-Q_B^T$ and it can be cast in the block diagonal form~\cite{cornfeld2018imbalance,neven2021symmetryresolved}: we denote as $\rho^\Gamma_{(q)}$ the resulting blocks, where $q$ indicates the charge. We denote the corresponding moments as $p_\alpha^{(q)}$, with
\begin{equation}
    p_\alpha^{(q)} = \text{tr}((\rho_{(q)}^\Gamma)^\alpha).
\end{equation}
We can exploit this to define $p_4$-negativity in a given block 
\begin{equation} \label{eq:e4_block}
    \E_4^{(q)}(\rho) = \frac{1}{2}\ln\left(\frac{(p_2^{(q)})^3}{p_4^{(q)}}\right)\,.
\end{equation}
By similar arguments as the non-symmetry-resolved case, one can show that
\begin{equation} \label{eq:bound_sr_eq}
    \E_4^{(q)}(\rho) \leq \E^{(q)}(\rho)\,,
\end{equation}
where 
\begin{equation} \label{eq:sr_negativity}
    \E^{(q)}(\rho) = \ln \lVert \rho^\Gamma_{(q)} \rVert_1^2 = \ln \sum_i \lvert \lambda_i^{(q)} \rvert\,,
\end{equation}
where $\lambda_i^{(q)}$ are the eigenvalues  of $\rho^{\Gamma}_{(q)}$. It is easy to see that $\E^{(q)}(\rho) \leq \E(\rho)$ for any given block, and combining with Eq.~\eqref{eq:bound_sr_eq}, we obtain
\begin{equation} \label{eq:bound_sr_e}
    \E_4^{(q)}(\rho) \leq \E(\rho)\,.
\end{equation}
Thus, $\E_4^{(q)}(\rho)$ plays a similar role as the symmetry-resolved $p_3$-PPT condition~\cite{neven2021symmetryresolved}, which is able to detect entanglement using symmetry-resolved moments for a fixed charge. Notably, one can enhance the detection capability and obtain a tighter lower bound for the negativity by combining all quantum number blocks. We define the symmetry-resolved (SR) $p_4$-negativity as
\begin{equation} \label{eq:sr_e4}
    \E_4^{\text{SR}}(\rho) = \ln\left(\sum_q \sqrt{\frac{(p_2^{(q)})^3}{p_4^{(q)}}}\right)\,.
\end{equation}
We have
\begin{equation} \label{eq:bound_sr_e2}
    \E_4^{\text{SR}}(\rho) \leq \E(\rho)\,.
\end{equation}
Thus, $\E_4^{\text{SR}}(\rho)$ is also a quantitative entanglement witness. 

Note that, we can also define
\begin{equation}\label{eq:e2}
    \Tilde{\E}_2(\rho) = \frac{1}{2}\ln\left(p_2\right)=-\frac{1}{2}S_2(\rho)\,,
\end{equation}
which satisfies
\begin{equation}
    \Tilde{\E}_2(\rho) \leq \E(\rho)\,.
\end{equation}
However, this inequality is not useful for entanglement detection since $\Tilde{\E}_2(\rho)\leq0$ for any state $\rho$. Nevertheless, the symmetry-resolved version
\begin{equation} \label{eq:sr_e2}
    \Tilde{\E}_2^{\text{SR}}(\rho) = \ln\left(\sum_q \sqrt{p_2}\right)\,,
\end{equation}
is not necessarily non-positive, so that it may already provide a nontrivial witness of entanglement since we have
\begin{equation}
    \Tilde{\E}_2^{\text{SR}}(\rho) \leq \E(\rho)\,.
\end{equation}
By Jensen's inequality, one can show the upper bound $\Tilde{\E}_2^{\text{SR}}(\rho)\leq\frac{1}{2}(\ln{N_q}-S_2(\rho))$, where $N_q$ is the number of quantum number. For the global $U(1)$ symmetry, this bound scales as $O(\ln n)$ since $N_q=n+1$. Thus, while $\Tilde{\E}_2^{\text{SR}}(\rho)$ is more tractable since it only involves the second PT moment, it has a drawback that it can only provide a limited lower bound to the negativity. In particular, it cannot be used to certify volume-law entanglement. 
Furthermore, $\Tilde{\E}_2^{\text{SR}}(\rho)$ fails to detect entanglement when the states are too mixed, with $S_2(\rho)\geq\ln(n+1)$.

Let us also mention that the SR witnesses can be used to detect the entanglement of arbitrary states, including those that do not have any symmetry~\cite{neven2021symmetryresolved}. Indeed, there exists a quantum channel $\Lambda$ that maps a state $\rho$ to $\sigma=\Lambda(\rho)$ that has a block-diagonal structure with respect to the desired symmetry. 
This channel is a local operation, which cannot increase the PT negativity. 
This implies
\begin{equation}
    \E_4^{\text{SR}}(\sigma) \leq \E(\sigma) \leq \E(\rho),
\end{equation}
which shows that $\E_4^{\text{SR}}(\Lambda(\rho))$ is also a quantitative entanglement witness.

Further, the SR witnesses can be defined in a similar way using the realignment matrix moments. It is known that the singular values of $R_\rho$ can be split into charge sectors $q$ which are eigenvalues of $Q_A\otimes \mathbb{1} - \mathbb{1} \otimes Q_A^T$~\cite{rath2023entanglement}. The symmetry-resolved moments are
\begin{equation}
    r_\alpha^{(q)} = \sum_i |\lambda_i^{(q)}|^\alpha.
\end{equation}
Then, we define
\begin{equation} \label{eq:sr_c4}
    \mathcal{C}_4^{\text{SR}}(\rho) = \ln\left(\sum_q \sqrt{\frac{(r_2^{(q)})^3}{r_4^{(q)}}}\right)\,.
\end{equation}
and
\begin{equation} \label{eq:sr_c2}
    \Tilde{\mathcal{C}}_2^{\text{SR}}(\rho) = \ln\left(\sum_q \sqrt{r_2}\right)\,,
\end{equation}
as our SR witnesses from the realignment matrix moments. Note that the CCNR negativity is not known to satisfy monotonicity under LOCC, and therefore the SR witnesses are only applicable for systems with symmetry, unlike the PT witnesses.

\section{Quantifying entanglement dimensionality} \label{sec:schmidt_number}
Beyond traditional entanglement measures, entanglement dimensionality has been identified as a crucial resource that quantifies the minimum dimension required to represent a given state~\cite{terhal2000schmidt}. This is formally measured by the Schmidt number, which we define as follows.

Consider the Schmidt decomposition for a pure state $\ket{\psi}$ in the subsystem $A$:
\begin{equation}
    \ket{\psi} = \sum_{i=1}^{s(\ket{\psi})} \sqrt{\lambda_i} \ket{i}_A \ket{i}_B,
\end{equation}
where $\lambda_i$ are known as the Schmidt coefficients 
The number of nonzero Schmidt coefficients is called the Schmidt rank  $s(\ket{\psi})$, which is an entanglement monotone for pure states~\cite{vidal2000entanglement}. For mixed states, we can extend this concept using a convex-roof construction to define the Schmidt number as~\cite{terhal2000schmidt,sanpera2001schmidt}
\begin{equation}
    s(\rho) = \min_{\rho=\sum_i p_i \ket{\psi_i}\bra{\psi_i}} \max_i s(\ket{\psi_i}) ,
\end{equation}
where the minimum is taken over all possible pure-state decompositions $\rho=\sum_i p_i \ket{\psi_i}\bra{\psi_i}$. The Schmidt number is also an entanglement monotone for general mixed states~\cite{terhal2000schmidt}.

Determining the exact Schmidt number for a given mixed state is a challenging task. Consequently, there has been significant interest in finding efficient methods to bound this quantity in experiments. Previous works have introduced techniques based on fidelity witnesses~\cite{terhal2000schmidt,sanpera2001schmidt}, covariance matrices~\cite{liu2024bounding}, and symmetric informationally complete measurement~\cite{morelli2023resource}. In the following, we demonstrate that our witnesses also provide a lower bound to the Schmidt number.

For pure states, the Schmidt rank is equivalent to the exponential of the entanglement R\'enyi entropy $e^{S_0(\ket{\psi})}$. Therefore, it can be lower bounded by the R\'enyi entropies as $\ln s(\ket{\psi})\geq S_\alpha(\ket{\psi})$ for any $\alpha>0$.

Following this idea, the $p_\alpha$-negativity also provides a lower bound to the Schmidt number for all mixed states. To show this, we first establish that the Schmidt number is bounded by the PT negativity as 
\begin{equation} \label{eq:neg_schmidt}
    \E(\rho) \leq \ln s(\rho).
\end{equation}
The proof is given in Appendix~\ref{eq:neg_schmidt}. As a consequence, the Schmidt number is also bounded by the $p_\alpha$ negativity due to Eq.~\eqref{eq:bound_e}. In particular, we have
\begin{equation}
    \E_4(\rho) \leq \ln s(\rho)
\end{equation}
and similarly
\begin{equation}
    \mathcal{C}_4(\rho) \leq \ln s(\rho)\,.
\end{equation}
This provides an efficient lower bound to the Schmidt number.

\section{Relation to operator space entanglement entropy} \label{sec:ccnr_osee}

Next, we discuss the connection between the $r_\alpha$-negativity and the operator space entanglement entropy (OSEE)~\cite{prosen2007operator}. For a mixed state $\rho$, the operator Schmidt decomposition with respect to the bipartition $A$ and $B$ is given by~\cite{zanardi2001entanglement}
\begin{equation}
\frac{\rho}{\sqrt{\text{tr}(\rho_A^2)}} = \sum_i \sqrt{\lambda_i^O} O_{A,i} \otimes O_{B,i}\,,    
\end{equation}
where $\lambda_i^O$ are the operator Schmidt coefficients and $\{O_{A,i}\}$ and $\{O_{B,i}\}$ are a set of orthonormal operators in the subsystem $A$ and $B$, respectively. The operator Schmidt decomposition can be obtained by vectorizing the (normalized) operator $\ket{\rho}=\frac{\rho}{\sqrt{\text{tr}(\rho_A^2)}}$ and performing the Schmidt decomposition for pure states on $\ket{\rho}$. The R\'enyi OSEE is given by~\cite{prosen2007operator}
\begin{equation}
    S^O_\alpha(\rho) = \frac{1}{1-\alpha}\ln(\sum_i (\lambda_i^O)^{\alpha} ). 
\end{equation}
A state $\rho$ has zero OSEE iff it can be written as a product state in the form of $\rho=\rho_A \otimes \rho_B$. The OSEE quantifies the representability of the state $\ket{\rho}$ as a matrix product state, since it plays a similar role as the entanglement R\'enyi entropy for pure states~\cite{schuch2008entropy}. Therefore, while OSEE is not a good measure of mixed-state entanglement (since separable mixed states can have non-zero OSEE), it has a direct relevance to the simulability of a mixed state using the well-established tools of matrix product states~\cite{Cirac2021mps}, originally developed for simulation of pure states. 

Furthermore, the OSEE can be directly related to the CCNR negativity~\cite{rath2023entanglement}. Indeed, one can see that the operator Schmidt coefficients $\lambda_i$ are simply the (rescaled) singular values of the realignment matrix $R_\rho$. Therefore, the R\'enyi OSEE are related to the realignment moments by
\begin{equation} \label{eq:relation_Ccnr_osee}
    S^O_\alpha(\rho) = \frac{1}{1-\alpha}\ln(r_{2\alpha}) + \frac{\alpha}{1-\alpha}S_2(\rho).
\end{equation}


\section{Entanglement testing} \label{sec:ent_testing}
While witnessing indicates the presence of entanglement, it is not a necessary criterion, as there are entangled states that have non-positive witness. 
An alternative way to study entanglement is via property testing~\cite{buhrman2008quantum,montanaro2013survey}: Here, one asks whether a given state has a property, or is it far from it?  Notably, as our witnesses are quantitative, they can be utilized for this task.

We now give an efficient testing algorithm to determine whether a state $\rho$ has high or low entanglement, where we assume that $\rho$ has low entropy $S_2(\rho)=O(\log n)$. 
This has been an open problem, as only for for $S_2=0$ efficient tests for entanglement have been known~\cite{ekert2002direct}, while for $S_2=\omega(\log n)$, testing entanglement has been shown to be inherently inefficient~\cite{bansal2024pseudorandomdensitymatrices}. Here, we design an efficient algorithm to test for entanglement which reliably distinguishes states depending on their entanglement:
\begin{theorem}[Efficient testing of entanglement]\label{thm:testing}
Let $\rho$ be an $n$-qubit state with $S_2(\rho)=O(\log n)$ where for a bipartition it is promised that 
\begin{align*}
\mathrm{either}\quad (a)& \,\,\E(\rho)=O(\log n) \,\,\mathrm{and}\,\,D_\mathrm{F}(\rho)=O(\log n) \,,\\
\mathrm{or}\quad (b)& \,\, \E(\rho)=\omega(\log n)\,\,\mathrm{and}\,\, D_\mathrm{F}(\rho)=\omega(\log n)\,.
\end{align*} 
Then, there exists an efficient quantum algorithm to distinguish case ($a$) and ($b$) using $\mathrm{poly}(n)$ copies of $\rho$ with high probability. 
\end{theorem}
We provide a short sketch of the proof idea in the following:
For case ($a$) one has $\E_4(\rho)=O(\log n)$, which follows from the bound on $\E$ of Eq.~\eqref{eq:bound_e}.
Case ($b$) implies that $S_2(\rho_A)=\omega (\log n)$, which can be shown using a bound on the \gEnt{}
\begin{equation} \label{eq:bound_df}
    S_\alpha(\rho_{A(B)}) \geq D_\text{F}(\rho),
\end{equation}
which is proven in Appendix~\ref{sec:proof_df}. 
Then, efficiency of testing follows from the efficient algorithms to estimate case ($a$) via $\E_4(\rho)$ (see Appendix~\ref{sec:efficientWitness} or Ref.~\cite{carteret2005noiseless,gray2018machinelearning}) and case ($b$) via $S_2(\rho_A)$ using SWAP tests. In particular, $\E_4(\rho)$ can be measured using SWAP tests involving $4$ copies of $\rho$, while $S_2$ can be measured using $2$ copies. 
Note that for $\E_4(\rho)$ to be efficiently estimated to additive precision, we require the condition $S_2(\rho)=O(\log n)$.


Next, we also give an efficient testing algorithm for the operator Schmidt rank, defined as the number of non-zero operator Schmidt coefficients. It is also equivalent to the (exponential of) OSEE $S_\alpha^O$ for $\alpha=0$. As mentioned in the previous section, while the OSEE does not directly quantify mixed-state entanglement, it determines whether the state can be efficiently represented as a matrix product operator. 

In the case of pure states, testing whether the Schmidt rank of a state \revA{across a given bipartition} is at most $r$ can be done using a technique known as weak Schur sampling~\cite{childs2007weakschur} using $O(r^2)$ copies. This algorithm requires entangled measurements on many copies of the state. For easier extension to mixed states, we first give a simplified testing algorithm for the Schmidt rank. We define the set of states of at most Schmidt rank $r$ as $M_r=\{\ket{\eta} : \ket{\eta} = \sum_{i=1}^r \sqrt{\lambda_i} \ket{\phi_i}_A \otimes \ket{\phi_i}_{B}\}$, and the maximum overlap of $\ket{\psi}$ with the states in $M_r$ as $\mathcal{F}_r(\ket{\psi})=\max_{\ket{\eta}\in M_r}\vert \braket{\eta}{\psi}\vert^2$. Now, we give a testing algorithm which only utilizes two-copy measurements:
\begin{theorem}[Efficient testing of Schmidt rank] \label{thm:schmidt_rank_test}
Let $\ket{\psi}$ be an $n$-qubit state where it is promised that 
\begin{align*}
\mathrm{either}\quad (a)& \,\,\mathcal{F}_r(\ket{\psi})\geq \epsilon_1 \,,\\
\mathrm{or}\quad (b)& \,, \mathcal{F}_r(\ket{\psi})\leq \epsilon_2\,,
\end{align*} 
\revA{across a given bipartition}, where it is assumed that $\epsilon_1^2/r> \epsilon_2+2\epsilon$.
Assuming $\epsilon=1/\mathrm{poly}(n)$, there exists an efficient quantum algorithm to distinguish case ($a$) and ($b$) using $O(r^2)$ two-copy measurements of $\ket{\psi}$ with high probability. 
\end{theorem}

The proof is given in Appendix~\ref{sec:schmidtrank}, where the key idea is to make use of the inequality 
\begin{equation} 
    \text{tr}(\rho_A^2) \leq  \mathcal{F}_r(\ket{\psi}) \leq \sqrt{r \text{tr}(\rho_A^2) }.
\end{equation}
Therefore, we can distinguish the two cases simply by measuring the purity $\text{tr}(\rho_A^2)$, which can be done by two-copy measurements of the state.

Now, we consider the maximum fidelity $\mathcal{F}_r^O(\rho)=\max_{\ket{O}\in M_r^O}\vert \braket{O}{\rho}\vert^2$~\cite{wang2008alternative}, where 
\begin{equation}
    M_r^O=\{\ket{O} : \frac{O}{\sqrt{\text{tr}(O_A^2)}} = \sum_i \sqrt{\lambda_i^O} O_{A,i} \otimes O_{B,i}\}
\end{equation}
is the set of operators whose operator Schmidt rank is at most $r$. By similar idea as the Schmidt rank testing, we give a testing algorithm for the operator Schmidt rank \revA{across a given bipartition}:
\begin{theorem}[Efficient testing of operator Schmidt rank]
Let $\rho$ be an $n$-qubit state where it is promised that 
\begin{align*}
\mathrm{either}\quad (a)& \,\,\mathcal{F}_r^O(\rho)\geq \epsilon_1 \,,\\
\mathrm{or}\quad (b)& \,\, \mathcal{F}_r^O(\rho)\leq \epsilon_2\,,
\end{align*} 
\revA{across a given bipartition}, where it is assumed that $\operatorname{tr}(\rho_A^2)=1/\mathrm{poly}(n)$ and  $(\epsilon_1^2/r- \epsilon_2)\operatorname{tr}(\rho_A^2)^2>1/\mathrm{poly}(n)$. Then, there exists an efficient quantum algorithm to distinguish case ($a$) and ($b$) using $O(r^2)$ copies of $\rho$ with high probability. 
\end{theorem}

The proof, as given in Appendix~\ref{sec:operator_schmidt_rank}, is along the same line as the Schmidt rank testing, replacing the purity by $r_4/r_2^2$. Note that $r_4$ can be efficiently measured by the algorithm detailed in Appendix~\ref{sec:CCNRmeas}, while $r_2$ is simply the purity.


\section{Circuit depth certification} \label{sec:depth_certification}
A key task is to certify the depth of the circuit that was used to prepare a given quantum state $\rho$. For noiseless circuits, i.e. $S_2(\rho)=0$, the circuit depth can be certified efficiently using Bell measurements~\cite{hangleiter2023bell}.
We now prove that the same task is also efficient for noisy quantum circuits, assuming the final state has low entropy, i.e. $S_2(\rho)=O(\log n)$.

We apply a noisy circuit of $d$ layers  $\rho=\Lambda_d\circ\Lambda_{d-1}\dots\circ \Lambda_1(\ket{0}\bra{0})$, where $\Lambda_i$ consists of noisy two-qubit gates implemented in a fixed architecture. 
The entanglement usually grows with each layer, and can be used to bound the total depth. One chooses an bipartition of equal size, where the cut is placed such that the two-qubit gates that act across the cut are minimal. Then, the maximal possible entanglement is only a function of the $\vert \partial A\vert$ two-qubit gates that act across the cut, while all other gates cannot increase entanglement. 

We now bound the maximal entanglement generated in such circuits: 
We show that, the overall entanglement after $d$ layers is bounded as
\begin{equation} \label{eq:depth_bound}
    \mathcal{E}(\rho)\leq  d\vert \partial A\vert\ln 2\,.
\end{equation}
The proof is given in Appendix~\ref{sec:depth_bound}.

With this bound, we can efficiently certify the minimal circuit depth $d$ that generated a given state $\rho$ by measuring $\mathcal{E}_4(\rho)$. In particular, we have
\begin{equation}
    d \geq \frac{1}{\vert \partial A\vert \ln2}\mathcal{E}_4(\rho)\,.
\end{equation}
where we used the inequality Eq.~\eqref{eq:bound_e}. By measuring $\mathcal{E}_4(\rho)$, one can certify $d$, where the estimator of $\mathcal{E}_4(\rho)$ gives a non-trivial lower bound on $d$ as long as $S_2(\rho)=O(\log n)$. We derive the error bounds on the estimator for certification in Appendix~\ref{sec:depthcert}.
For example, when we have $\mathcal{E}_4=\omega(\log n)$, our estimator correctly determines $d=\omega(\log n)$ using a polynomial number of measurements.

\section{Pseudorandom density matrices and pseudoentanglement} \label{sec:prdm_pseudoentanglement}
Our quantitative witnesses have direct implications on recent developments in quantum cryptography.
In particular, pseudorandom states (PRS) have been proposed as efficiently preparable pure state ensembles that are indistinguishable from Haar random states by any efficient quantum algorithm~\cite{ji2018pseudorandom}. Such states have been generalized to mixed states, known as pseudorandom density matrices (PRDMs). They are efficiently preparable states which are indistinguishable from random mixed states with entropy $S_2$ for any efficient algorithm~\cite{bansal2024pseudorandomdensitymatrices}.  Their formal definition is given in Appendix~\ref{sec:prdm}.

For pure states, it has been shown that PRS (which are equivalent to PRDMs with $S_2=0$)  must have $\omega(\log n)$ entanglement entropy~\cite{ji2018pseudorandom,bouland2022quantum}. Further, it has been shown that PRDMs with $S_2=\omega(\log n)$ can possess any amount of entanglement. 
However, the entanglement property of PRDMs for bounded entropy $S_2=O(\log n)$ has been unclear. 
We now show that for $S_2=O(\log n)$, the entanglement of PRDMs as measured by PT negativity is bounded as follows:
\begin{proposition}[Entanglement of PRDM]
    Any ensemble of PRDMs  with entropy $S_2=O(\log n)$ must have PT negativity $\E(\rho)=\omega(\log n)$ with high probability. 
\end{proposition}
This is proven in Appendix~\ref{sec:prdm}. The key idea is to assume by contradiction that there exists a PRDM with $\E(\rho)=O(\log n)$. This would imply that $p_4=1/\mathrm{poly}(n)$, which can be efficiently distinguished from Haar random states via the efficient measurement of the PT moment $p_4$ in Appendix~\ref{sec:efficientWitness}.

A related notion is pseudoentanglement~\cite{bouland2022quantum,bansal2024pseudorandomdensitymatrices}, which are (efficiently preparable) ensemble of states with low entanglement $g(n)$ that are indistinguishable from
an ensemble with high entanglement $f(n)$. The formal definition is given in Appendix~\ref{sec:prdm}. 
Our results allow to provide new bounds on pseudoentanglement of mixed states. 
A fundamental question is the difference in entanglement between the two ensembles, which is dubbed the pseudoentanglement gap.
Previous works have shown that the pseudoentanglement gap is bounded as $f(n)=\Theta(n)$ vs $g(n)=\omega(\log n)$ for pure states, i.e. $S_2(\rho)=0$~\cite{bouland2022quantum}, while for highly mixed states with $S_2(\rho)=\omega(\log n)$, one finds the maximal possible gap of $f(n)=\Theta(n)$ vs $g(n)=0$~\cite{bansal2024pseudorandomdensitymatrices}. However, the gap for weakly mixed states $S_2(\rho)=O(\log n)$ has been an open problem.

Here, we now compute the entanglement in terms of PT negativity and \gEnt{}. 
Then, we can bound the largest possible pseudoentanglement gap as follows:
\begin{proposition}[Pseudoentanglement of mixed states]
Pseudoentangled state ensembles with entropy $S_2=O(\log n)$ can have a pseudoentanglement gap of at most $f(n)=\Theta(n)$ vs $g(n)=\omega(\log n)$, where entanglement is measured by both PT negativity $\E(\rho)$ and \gEnt{} $D_\mathrm{F}(\rho)$. 
\end{proposition}
\begin{proof}
    This follows directly from contradiction: Assume there exists pseudoentangled state ensembles with low-entanglement $g(n)=\E(\rho)=D_\mathrm{F}(\rho)=O(\log n)$, and high-entanglement $f(n)=\E(\rho)=D_\mathrm{F}(\rho)=\Theta(\log n)$. Then, via Thm.~\ref{thm:testing} and assuming $S_2(\rho)=O(\log n)$, those ensembles can be efficiently distinguished, thus they cannot be pseudoentangled. 
\end{proof}
We note that pseudoentanglement is connected to the existence of a fundamental primitive of quantum cryptography~\cite{brakerski2022computational}: EFI pairs are pairs of mixed states that are statistically far, yet indistinguishable for any efficient quantum algorithm. 
EFI pairs require pseudoentanglement to exist~\cite{grilo2025quantum}, thus our bounds place potential restrictions on EFI pairs as well.


\section{Entanglement phase diagram} \label{sec:entanglement_phase_diagram}
In this section, we discuss the entanglement phase diagram in two different sets of states: Haar random states and stabilizer states.
\subsection{Haar random states} \label{sec:haar}
We will consider a tripartite system in a pure state $\ket{\psi}\in\mathcal{H}_A\otimes\mathcal{H}_B\otimes\mathcal{H}_C$, where $\ket{\psi}$ is Haar random. Here, the region $X$ consists of $n_X$ qubits with Hilbert space dimension $L_X=2^{n_X}$, where $X\in\{A,B,C\}$. We also denote $n_{AB}=n_A+n_B$ and $L_{AB}=L_A L_B=2^{n_{AB}}$. \revA{We will analyze the Haar-random induced mixed states $\rho=\text{tr}_C(\ket{\psi})$ which are generated by tracing out subsystem $C$ from tripartite pure Haar random states.  For $n_\text{C}=0$, such random states are Haar random states, while for $n_\text{C}\rightarrow\infty$, they asymptotically become maximally mixed states.}

Their entanglement properties have been studied in~\cite{shapourian2021entanglement} using the PT negativity ${\mathcal E}(\rho)$ (see also~\cite{fukuda2013partial,bhosale2012entanglement,aubrun2012partial,aubrun2013entanglement,aubrun2012phase}).

We now denote the average over the Haar measure $\mu_n$ of $n$-qubit states as ${\mathbb E}_{\text{H}}[.]\equiv \mathbb{E}_{\rho \in \mu_n}[.]$. 
The scaling behavior of the average negativity, ${\mathbb E}_{\text{H}}[{\mathcal E}(\rho)]$, determines an entanglement phase diagram for Haar random states as a function of the partition sizes, which consists of three different ``entanglement phases''.
This includes: 
(Phase~I)~For $n_C>n_{AB}$, ${\mathbb E}_{\text{H}}[{\mathcal E}(\rho)]$ vanishes and thus, on average, $\rho$ is PPT. This phase is called the PPT phase. (Phase II)~For $n_C<n_{AB}$ and $n_B>n/2$, ${\mathbb E}_{\text{H}}[{\mathcal E}(\rho)]\sim n_A$ and thus the subsystem $A$ is maximally entangled with the subsystem $B$, yet it is not entangled  with the subsystem $C$. In other words, the entanglement in this phase is fully bipartite. This phase is called the maximally entangled (ME) phase. By symmetry, the ME phase also appears for $n_A>n/2$. (Phase III) For $n_C<n_{AB}$ and $n_A,n_B<n/2$, the states are tripartite entangled since all three subsystems $A,B,$ and $C$ are mutually entangled, and ${\mathbb E}_{\text{H}}[{\mathcal E}(\rho)]\sim (n_{AB}-n_C)/2$. This phase is called the Entanglement Saturation (ES) phase. The phase diagram is sketched in Fig.~\ref{fig:summary_results}b. Here, we show that the PT negativity in all three phases is exactly reproduced (in the leading order) by the $p_4$-negativity.

The Haar average of the PT moments for integer $\alpha$ is given by~\cite{shapourian2021entanglement}
\begin{equation}
    \mathbb{E}_{\text{H}}[p_\alpha]\simeq\frac{1}{(L_A L_BL_C)^\alpha} \sum_{\tau\in \mu_\alpha}L_C^{c(\tau)} L_{A}^{c(\sigma_+\circ \tau)} L_{B}^{c(\sigma_-\circ \tau)},
    \label{eq:ExactExpression}
\end{equation}
where $\sigma_{\pm}$ are two special permutations defined as $\sigma_\pm(k)=(k\pm 1) \mod n$, i.e.  cyclic (and anti-cyclic) permutations. Here, for any permutation $\tau\in \mu_\alpha$, $c(\tau)$ is the number of cycles in $\tau$, including the cycles of length one. By Eq.~\eqref{eq:ExactExpression}, one can obtain the thermodynamic limit of $\mathbb{E}_{\text{H}}[p_\alpha]$ for various partition sizes and compute the PT negativity by analytic continuation as
\begin{equation}
\mathbb{E}_{\text{H}}[\E(\rho)]\simeq\lim_{\alpha\to1/2}\ln\mathbb{E}_{\text{H}}[p_{2\alpha}],   
\end{equation}
as done in Ref.~\cite{shapourian2021entanglement}. We will now see how the negativity can be simply obtained by utilizing the $p_4$-negativity, without requiring any knowledge on the higher PT moments.

For the analysis, we will utilize the upper bounds for the PT negativity: 
\begin{equation}
    \E(\rho) \leq \min\{n_A \ln{2},n_B\ln{2},\frac{1}{2}(n_{AB}\ln{2}-S_2(\rho_{AB})\}.
\end{equation}
 These bounds are shown in Appendix~\ref{sec:upper_bounds_neg} through the $p_\alpha$-negativity for $\alpha<1$. In Haar random states, $\mathbb{E}_{\text{H}}[S_2(\rho_{AB})]\geq -\ln \mathbb{E}_{\text{H}}[\Tr(\rho_{AB}^2)]\simeq \min(n_{AB}\ln{2},n_C\ln{2})$, where the first inequality is because the function $f(x)=-\ln x$ is convex. Thus, we have an upper bound for $\mathbb{E}_{\text{H}}[\E(\rho)]$ as
\begin{align}\label{eq:upp_bound_haar}
    \mathbb{E}_{\text{H}}[\E(\rho)] \leq &\min\{n_A \ln{2},n_B\ln{2},\\
    &\frac{1}{2}(n_{AB}\ln{2}-\min(n_{AB}\ln{2},n_C\ln{2}))\}.
\end{align} 
Next, we compute $\mathbb{E}_{\text{H}}[\E_4(\rho)]$, which gives a lower bound for $\mathbb{E}_{\text{H}}[\E(\rho)]$.

Using Eq.~\eqref{eq:upp_bound_haar}, we can already deduce that, for $n_C>n_{AB}$, $\mathbb{E}_{\text{H}}[\E(\rho)]\lesssim0$ from the third upper bound. This immediately implies that this phase is a PPT phase, with $\mathbb{E}_{\text{H}}[\E(\rho)]\simeq0$. Although $\E_4$ is not necessary to obtain the negativity in this phase, we compute it here for completeness. In this phase, ${\mathbb E}_{\text{H}}[p_\alpha]\simeq L_{AB}^{1-\alpha}$ in the thermodynamic limit~\cite{shapourian2021entanglement}. In particular, ${\mathbb E}_{\text{H}}[p_4]\simeq L_{AB}^{-3}$, and thus ${\mathbb E}_{\text{H}}[-\ln{p_4}]\gtrsim3n_{AB}\ln{2}$. We also have $-\ln{p_2}=S_2(\rho_{AB})\leq n_{AB}\ln{2}$. It follows that ${\mathbb E}_{\text{H}}[\E_4(\rho)]={\mathbb E}_{\text{H}}[-\frac{1}{2}\ln{p_4}+\frac{3}{2}\ln{p_2}]\gtrsim 0$. Combining ${\mathbb E}_{\text{H}}[\E_4(\rho)]\leq{\mathbb E}_{\text{H}}[\E(\rho)]$ and the upper bound obtained above, we obtain ${\mathbb E}_{\text{H}}[\E_4]\simeq0$. Thus, ${\mathbb E}_{\text{H}}[\E_4]$ has the same leading order as ${\mathbb E}_{\text{H}}[\E]$.

Next, for $n_C<n_{AB}$ and both $n_A<n/2$ and $n_B<n/2$, one gets in the thermodynamic limit
\[
{\mathbb E}_{\text{H}}[p_\alpha]\simeq\begin{cases}
\displaystyle{\frac{C_k L_{AB}}{(L_{AB} L_C)^k}},&\alpha=2k\\[4mm]
\displaystyle{\frac{(2k+1)C_k}{(L_{AB} L_C)^k}},&\alpha=2k+1,
\end{cases}
\]
where $C_k=\binom{2k}{k}/(k+1)$ is the $k$th Catalan number. In particular, ${\mathbb E}_{\text{H}}[p_4]\simeq L_{AB}^{-1} L_{C}^{-2}$, and thus ${\mathbb E}_{\text{H}}[-\ln{p_4}]\gtrsim n_{AB}\ln{2}+2n_C\ln{2}$. Since $-\ln{p_2}\leq n_C\ln{2}$, we have ${\mathbb E}_{\text{H}}[\E(\rho)]\geq{\mathbb E}_{\text{H}}[\E_4(\rho)]\gtrsim \frac{1}{2}(n_{AB}-n_C)\ln{2}$. Combining with the third upper bound in Eq.~\eqref{eq:upp_bound_haar}, we obtain 
\begin{equation}
    {\mathbb E}_{\text{H}}[\E]\simeq{\mathbb E}_{\text{H}}[\E_4]\simeq  \frac{1}{2}(n_{AB}-n_C)\ln{2}.
\end{equation}

Finally, when $n_{AB}>n_C$ and $n_A>n/2$, we obtain in the thermodynamic limit
\[
{\mathbb E}_{\text{H}}[p_\alpha]\simeq\begin{cases}
L_C^{1-\alpha}L_{B}^{2-\alpha},&\alpha=2k\\[4mm]
(L_C L_{B})^{1-\alpha}&\alpha=2k+1.
\end{cases}
\]
In particular, ${\mathbb E}_{\text{H}}[p_4]\simeq L_{B}^{-2} L_{C}^{-3}$, and thus ${\mathbb E}_{\text{H}}[-\ln{p_4}]\gtrsim 2n_B\ln{2}+3n_C\ln{2}$. We again obtain ${\mathbb E}_{\text{H}}[\E(\rho)]\geq{\mathbb E}_{\text{H}}[\E_4(\rho)]\gtrsim n_B\ln{2}$. The case in which $n_{AB}>n_C$ and $n_B>N/2$ is obtained by replacing $L_{B}$ with $L_{A}$ in the latter formula. Combining with the first and second upper bound in Eq.~\eqref{eq:upp_bound_haar}, we obtain 
\begin{equation}
    {\mathbb E}_{\text{H}}[\E]\simeq{\mathbb E}_{\text{H}}[\E_4]\simeq  \min(n_A,n_B)\ln{2}.  
\end{equation}

To sum up, we find that the leading order of ${\mathbb E}_{\text{H}}[\E_4]$ in the thermodynamic limit is identical to ${\mathbb E}_{\text{H}}[\E]$ in all the three entanglement phases of Haar random states. Specifically, this is established by showing that ${\mathbb E}_{\text{H}}[\E]$ is lower and upper bounded by the same leading-order scaling, thus fixing the scaling of ${\mathbb E}_{\text{H}}[\E]$ even without considering higher-order moments. In contrast, ${\mathbb E}_{\text{H}}[\E]$ was previously obtained by computing all the PT moments, followed by analytic continuation~\cite{shapourian2021entanglement}. This has a remarkable implication: the PT negativity in Haar-random states can be obtained in the leading order by utilizing low-order PT moments, which can be estimated experimentally and computed numerically, thereby obviating the need for the generally more intricate task of computing all PT moments. Additionally, we stress that our analysis above does not make any assumption of self-averaging, i.e. ${\mathbb E}[\ln{X}]\simeq\ln{\mathbb E}[X]$. This ensures the generality of the argument, which implies that these results hold for any state design which is at least a 4-design, where the resulting entanglement phase diagram reproduces that of the Haar random induced mixed states.

In addition, we note that the results above similarly hold for the $p_\alpha$-negativity for $\alpha\geq0$, as detailed in Appendix~\ref{sec:pneg_haar}. This behavior resembles the scaling of entanglement R\'enyi entropies in Haar random states, which saturate their maximal values (in the leading order) for any R\'enyi index $\alpha$~\cite{Page1993,nadal2010phase}.

\subsection{Stabilizer states} \label{sec:stab}
 Pure stabilizer states $\ket{\psi_S}$ are states generated by applying Clifford unitaries $U_\text{C}$ onto the $\ket{0}$ state. Here, Clifford unitaries $U_\text{C}$ are unitaries generated from Hadamard, S-gate and CNOT gates~\cite{gottesman1998heisenberg}. Clifford unitaries map (under conjugation) any $n$-qubit Pauli operator $P$ to some $n$-qubit Pauli operator, $P'$, i.e. $P'=U_\text{C}P U_\text{C}^\dag$. It is known by the Gottesman-Knill theorem that stabilizer states can be simulated efficiently on classical computers~\cite{gottesman1997stabilizer,gottesman1998heisenberg,gottesman1998theory,aaronson2004improved}.

As shown in Ref.~\cite{bravyi2006ghz},
any three-partite stabilizer state $\ket{\psi_S}$ can be decomposed into GHZ states, Bell states, and product states, distributed among the three parties, $A$, $B$, and $C$. That is, $\ket{\psi}$ can be written as
\begin{eqnarray}
\label{eq:GHZ}
\ket{\psi_S} &=& U_A U_B U_C\ket{0}^{\otimes s_{A}}_A
\ket{0}^{\otimes s_{B}}_B
\ket{0}^{\otimes s_{C}}_C
\ket{\mathrm{GHZ}}^{\otimes g_{ABC}}_{ABC}
\nonumber \\
&&
\ket{\mathrm{EPR}}^{\otimes e_{AB}}_{AB}
\ket{\mathrm{EPR}}^{\otimes e_{AC}}_{AC}
\ket{\mathrm{EPR}}^{\otimes e_{BC}}_{BC},
\end{eqnarray}
with $U_A,U_B,U_C$ unitary Clifford operators on $A,B,C$, respectively, $\ket{\mathrm{EPR}}_{AB}$ denotes a two-qubit EPR pair with one qubit in $A$ and the other in $B$ (similarly for $\ket{\mathrm{EPR}}_{AC}$ and $\ket{\mathrm{EPR}}_{BC}$), and $\ket{\mathrm{GHZ}}_{ABC}$ is a three-qubit GHZ state with one qubit in each of $A,B,$ and $C$. Given the decomposition above, one can directly see that the entanglement negativity is given by $\E(\rho)=e_{AB}$~\cite{sang2021entanglement}. It can also be seen that the negativity spectrum of stabilizer states  is constrained to two values $\lambda_i=\pm \sqrt{p_3}$, i.e all eigenvalues $\lambda_i$ of $\rho^\Gamma$ are either $\sqrt{p_3}$ or $-\sqrt{p_3}$. One can then obtain the PT moments of stabilizer states as~\cite{carrasco2024entanglementphasediagram}
\begin{eqnarray}
p_2 &=& \left(\frac{1}{2}\right)^{g_{ABC}+e_{AC}+e_{BC}}
\nonumber \\
p_3 &=& \left(\frac{1}{4}\right)^{e_{AB}} \left(\frac{1}{4}\right)^{g_{ABC}+e_{AC}+e_{BC}}
\nonumber \\
p_4 &=& \left(\frac{1}{4}\right)^{e_{AB}} \left(\frac{1}{8}\right)^{g_{ABC}+e_{AC}+e_{BC}}.
\label{eq:pnstab}
\end{eqnarray}
The PT negativity for stabilizer states is then given by a simple function of the PT moments as~\cite{carrasco2024entanglementphasediagram} 
\begin{equation}\label{eq:neg_stab}
\E(\rho)=\Tilde{\E}_3(\rho) \quad \text{ for all stabilizer states.}
\end{equation}
This shows that stabilizer states are PPT iff they satisfy the  $p_3$-PPT condition~\cite{elben2020mixed}. Analogously, we find that 
\begin{equation}\label{eq:neg4_stab}
\E(\rho)=\E_4(\rho) \quad \text{ for all stabilizer states.}
\end{equation}
Namely, stabilizer states saturate the bound in Eq.~\eqref{eq:bound_e}. This also follows from the flatness of the PT spectrum discussed above.

We now consider a tripartite system in a pure state $\ket{\psi}\in\mathcal{H}_A\otimes\mathcal{H}_B\otimes\mathcal{H}_C$, where $\ket{\psi}$ is a random stabilizer state. The entanglement phase diagram as a function of the partition size can be obtained using $\E(\rho)=\Tilde{\E}_3(\rho)$. Since stabilizer states form a $3$-design~\cite{kueng2015qubitstabilizer,zhu2017multiqubit}, we have $\mathbb{E}_{\text{S}}[p_{2,3}]=\mathbb{E}_{\text{H}}[p_{2,3}]$, where ${\mathbb E}_{\text{S}}[.]$ denotes the average over stabilizer states. Therefore, we can again compute ${\mathbb E}_{\text{S}}[\E]$ by using the upper bounds in Eq.~\eqref{eq:upp_bound_haar} and obtaining a lower bound to ${\mathbb E}_{\text{S}}[\Tilde{\E}_3]$ using a similar approach as in the Haar random case. Note that the upper bounds in Eq.~\eqref{eq:upp_bound_haar} are obtained from the second moments, and thus they also apply to stabilizer states. One can then verify that 
\begin{equation}
    {\mathbb E}_{\text{S}}[\E]={\mathbb E}_{\text{S}}[\Tilde{\E}_3]\simeq  0,
\end{equation}
for $n_C>n_{AB}$,
\begin{equation}
    {\mathbb E}_{\text{S}}[\E]={\mathbb E}_{\text{S}}[\Tilde{\E}_3]\simeq \frac{1}{2}  (n_{AB}-n_C)\ln 2,
\end{equation}
for $n_C<n_{AB}$ and both $n_A<n/2$ and $n_B<n/2$, and 
\begin{equation}
    {\mathbb E}_{\text{S}}[\E]={\mathbb E}_{\text{S}}[\Tilde{\E}_3]\simeq  \min(n_A,n_B)\ln{2},
\end{equation}
for $n_C<n_{AB}$ and either $n_A>n/2$ or $n_B>n/2$. In particular, the entanglement phase diagram reproduces that of Haar random states, despite the fact that stabilizer states do not form a $4$-design~\cite{zhu2016clifford}. It was shown, for example, that the quantity $\Tilde{r}_2$, which is a function of the PT moments up to the fourth order (see Eq.~\eqref{eq:r2}), takes different values across the phase diagram for Haar random and stabilizer states~\cite{carrasco2024entanglementphasediagram}, thus also highlighting the limitation of $\Tilde{r}_2$ in distinguishing entanglement phases. We further note that, since the relation $\E(\rho)=\Tilde{\E}_3(\rho)$ is specific to stabilizer states, the results presented above do not necessarily extend to other state $3$-designs. 

Further, the realignment moments of stabilizer states are given by~\cite{yin2023mixedstate}
\begin{eqnarray}
r_\alpha &=& \left(\frac{1}{2}\right)^{S_{AB}+(\frac{\alpha}{2}-1)(S_{A}+S_{B})},
\label{eq:rnstab}
\end{eqnarray}
where $S_A$ denotes the entanglement of the subsystem $A$. The CCNR negativity for stabilizer states is then given by
\begin{equation}
    \mathcal{C}(\rho) = \frac{S_A+S_B}{2} - S_{AB},
\end{equation}
which also yields
\begin{equation}\label{eq:cneg4_stab}
\mathcal{C}(\rho)=\mathcal{C}_4(\rho)\quad \text{ for all stabilizer states,}
\end{equation}
thus saturating the bound in Eq.~\eqref{eq:bound_c}. This implies that the singular values of the realignment matrix are flat.


\section{Entanglement and noise} \label{sec:noise_haar}
In this section, we analyze the average bipartite entanglement of Haar-random pure states subjected to global depolarizing noise. These systems approximate deep random quantum circuits with local noise~\cite{dalzell2021randomquantumcircuits}. We are interested in determining under which condition entanglement survives in the presence of noise. We consider Haar states subject to global depolarizing noise
\begin{equation}
    \rho'
 = (1-p)\ket{\psi}\bra{\psi} + p I/2^n,   
\end{equation}
with noise probability $p$ and $\ket{\psi}$ being drawn from the Haar measure. One can show that the R\'enyi-2 entropy of the depolarized state is given by ${\mathbb E}_{\text{H}}[S_2[\rho']]= -\ln\left[(1-p)^2 + p(2-p)/2^n \right]$. Now, let $p_\alpha'$ be the PT moment of the depolarized Haar-random states and $p_\alpha$ those of a Haar-random state. Then, the mean values ${\mathbb E}_{\text{H}}[p_\alpha']$ can be expressed in terms of the mean values ${\mathbb E}_{\text{H}}[p_\alpha]$ as
\begin{equation} \label{eq:pn_depo}
{\mathbb E}_{\text{H}}[p_\alpha']=\sum_{k=0}^\alpha\binom{\alpha}{k}(1-\epsilon)^k{\mathbb E}_{\text{H}}[p_k](\epsilon/2^n)^{\alpha-k}\,,
\end{equation}
with ${\mathbb E}[p_0]=p_0=2^n$.

For any fixed $p<1$ and $n \to \infty$, we find that the leading term ${\mathbb E}_{\text{H}}[\E_4]\simeq \min(n_A,n_B)\ln{2}$ for extensive subsystems $n_A$ and $n_B$. Therefore, the averaged bipartite entanglement entropy of the depolarized state, quantified by the PT negativity, is unchanged from that of the pure state, which follows the Page curve~\cite{Page1993}. As a consequence, $\E_4$ is remarkably effective in detecting mixed-state entanglement in Haar states even in the presence of strong noise.

We further consider exponential noise $p=1-2^{-\beta n}$ with fixed $\beta$, and determine for which $\beta$ the witness $\E_4$ can successfully detect entanglement. For $\beta\leq1/2$, we have ${\mathbb E}_{\text{H}}[S_2(\rho')]\simeq2\beta n\ln 2$. Moreover, using Eq.~\eqref{eq:pn_depo}, we obtain ${\mathbb E}_{\text{H}}[p_4]\simeq2^{-4\beta n-2n_B}$, yielding
\begin{equation} \label{eq:e4_depo}
    {\mathbb E}_{\text{H}}[\E_4]\gtrsim  \min(n_A,n_B)\ln{2} - \beta n\ln{2}.
\end{equation}
Thus, we detect entanglement via $\E_4>0$ for extensive subsystems whenever $\beta < \min(n_A,n_B)/n$. Beyond detection, the value of $\E_4$ in Eq.~\eqref{eq:e4_depo} shows rigorously that the negativity obeys a volume law in the entangled regions. Notably, the $p_3$-PPT condition is unable to show the volume law. 
In fact, the PT negativity can also be obtained by considering the distribution of the Schmidt coefficients of a Haar random state, yielding~\cite{shaw2024benchmarking}
\begin{equation} \label{eq:e_depo}
    {\mathbb E}_{\text{H}}[\E]\simeq  \min(n_A,n_B)\ln{2} - \beta n\ln{2}.
\end{equation}
Thus, $\E_4$ matches $\E$ in depolarized Haar random states, i.e. ${\mathbb E}_{\text{H}}[\E_4]\simeq{\mathbb E}_{\text{H}}[\E]$. This highlights the effectiveness of $\E_4$ to quantitatively detect entanglement, even under exponentially strong noise.

We provide further discussion on the effects of global depolarizing noise on entanglement in Appendix~\ref{sec:depo_noise}.

\section{Matrix product states and matrix product operators algorithms} \label{sec:mps_mpo}

\begin{figure*}[htbp]
\centering
\subfigimg[width=0.40\linewidth]{a}{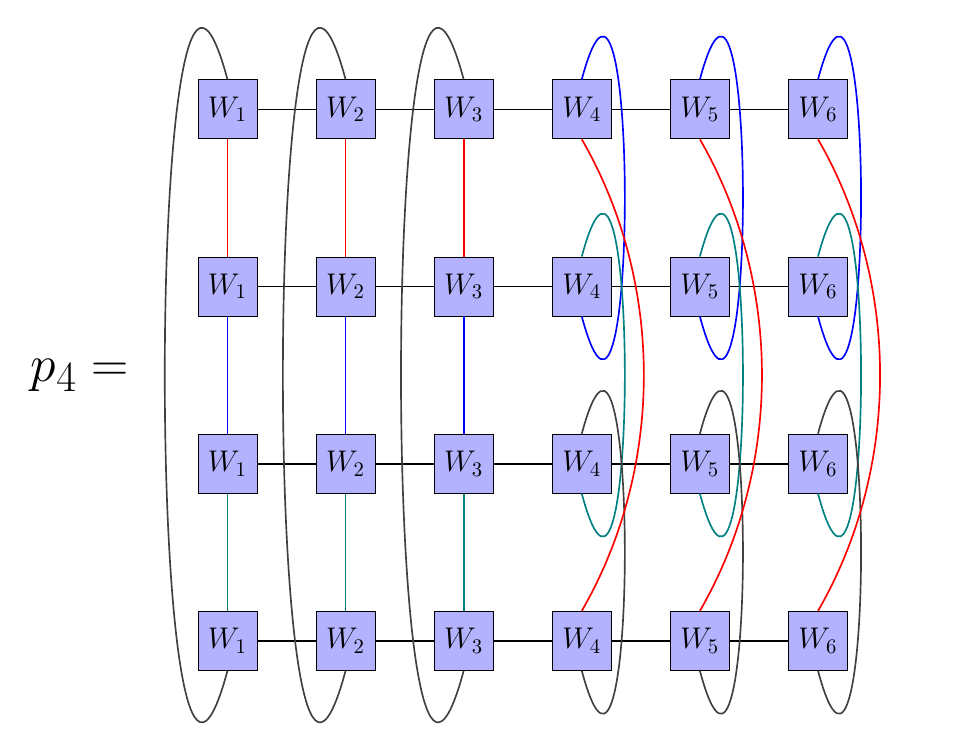}
\subfigimg[width=0.40\linewidth]{b}{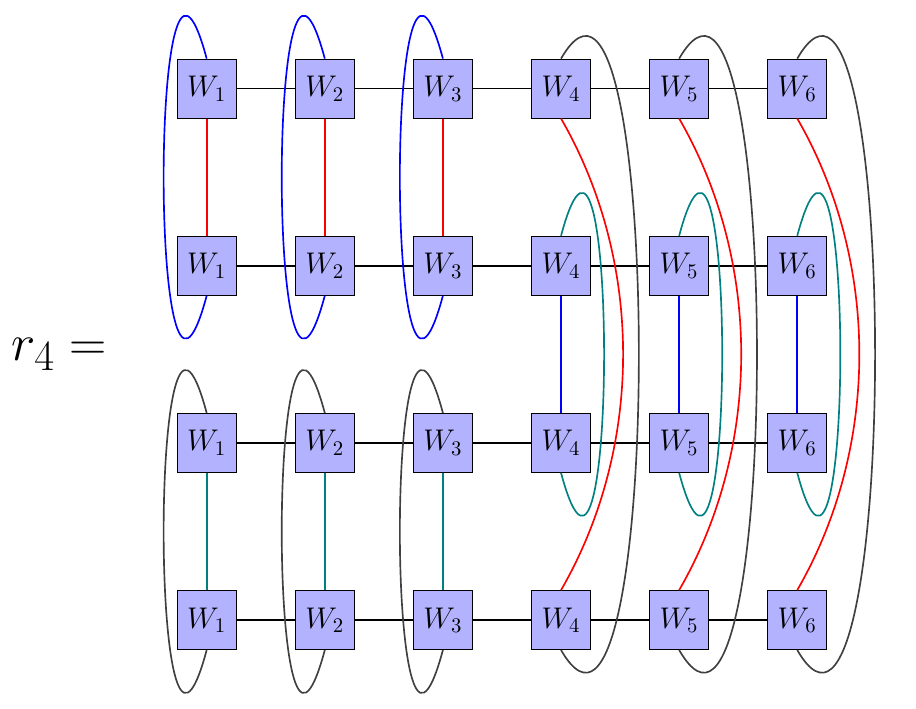}
\caption{ MPO contractions to compute (a) $p_4$ and (b) $r_4$. The system size is $n=6$ and we take $A=\{1,2,3\}$ and $B=\{4,5,6\}$.}
\label{fig:mpo_p4_r4}
\end{figure*}

Here, we introduce algorithms to compute partial transpose and realignment moments for matrix product states (MPS) and matrix product operators (MPO). 

Let us consider a system of $n$ qubits in a pure state $|\psi \rangle$ given by an MPS of bond dimension $\chi$:
\begin{equation} \label{eq:mps}
|\psi \rangle=\sum_{s_1,s_2,\cdots,s_n} A^{s_1}_1 A^{s_2}_2 \cdots A^{s_n}_n |s_1,s_2,\cdots s_n \rangle
\end{equation}
with $A_i^{s_i}$ being $\chi \times \chi$ matrices, except at the left (right)
boundary where $A_1^{s_1}$ ($A_n^{s_n}$) is a $1 \times \chi$ ($\chi \times 1$) row (column) vector. Here $s_i \in \left \lbrace 0, 1 \right \rbrace$  is a local computational basis. We then study the entanglement within $\rho=\text{tr}_C(\ket{\psi}\bra{\psi})$, where $C$ is a subsystem of $\ket{\psi}$. The case of translation-invariant MPS is discussed in Appendix~\ref{sec:ti_mps}.

We also consider a density matrix $\rho$ of $n$ qubits represented in the following MPO form
\begin{equation}
    \rho =\sum_{\boldsymbol{s},\boldsymbol{s'}} W^{s_1,s'_1}_1 W^{s_2,s'_2}_2 \cdots W^{s_n,s'_n}_n | s_1, \cdots, s_n \rangle \langle s'_1, \cdots, s'_n |
\end{equation}
with $W^{s_i,s'_i}_i$ being $\chi \times \chi$ matrices, except at the left (right)
boundary where $W^{s_1,s'_1}$ (or $W^{s_n,s'_n}$) is a $1 \times \chi$ ($\chi \times 1$) row (column) vector. 

 \subsection{Partial transpose moments}
 The PT moments $p_\alpha$ can be computed efficiently for integer $\alpha\geq 1$ for MPS. To do so, we write $p_\alpha$ as 
 \begin{equation}
\begin{split}
     p_\alpha = \sum_{    a_{(\gamma)},a'_{(\gamma)},
    b_{(\gamma)},b'_{(\gamma)}
    } \prod_{\gamma=1}^\alpha \left[ \delta_{a'_{(\gamma)},a_{(\sigma_+(\gamma))}} \delta_{b'_{(\gamma)},b_{(\sigma_-(\gamma))}} \vphantom{\langle a_{(\gamma)},b_{(\gamma)}\vert \rho\vert a'_{(\gamma)},b'_{(\gamma)}\rangle} 
    \right.
    \\
    \left.
    \langle a_{(\gamma)},b_{(\gamma)}\vert \rho\vert a'_{(\gamma)},b'_{(\gamma)}\rangle \right]\,.
    \end{split}
 \end{equation}
 Using this identity, we can compute $p_\alpha$ for MPS by forming a transfer matrix at each site, defined as
\begin{equation}
    L_i = \sum_{\{s_{(\gamma)},s'_{(\gamma)}\}}  \prod_{\gamma=1}^\alpha \delta_{s'_{(\gamma)},s_{(\pi(\gamma))}}\bigotimes_{\gamma=1}^\alpha A_i^{s_{(\gamma)}} \otimes \overline{A}_i^{s'_{(\gamma)}}\,,
\end{equation}
where $\overline{A}$ indicates the complex conjugate of $A$, $\pi()$ is a permutation of indices, with $\pi=\sigma_+$ if $i\in A$, $\pi=\sigma_-$ if $i\in B$, and $\pi=e$ if $i\in C$.
The terms $\delta_{s'_{(\gamma)},s_{(\pi(\gamma))}}$ imply contractions of physical indices $s'_{(\gamma)}$ and $s_{(\pi(\gamma))}$. 
We now define a tensor $L$ initialized to $L=L_1$ and iteratively compute $L\to L L_i$, for $i=2,\cdots,n$. At each step (before the final site), $L$ is a tensor with $2\alpha$ indices, each of size $\chi$. To minimize the contraction cost, each of the tensors $A_i$ is contracted to $L$ one by one, yielding the overall cost of $O\left(n\alpha\chi^{2\alpha+1}\right)$. The tensor $L$ after the final iteration is a scalar, which is equivalent to $p_\alpha$. 

The algorithm can be straightforwardly adapted to MPO, and the MPO contraction is sketched in Fig.~\ref{fig:mpo_p4_r4}(a).  We define the transfer matrix
\begin{equation}
    L_i = \sum_{\{s_{(\gamma)},s'_{(\gamma)}\}}  \prod_{\gamma=1}^\alpha \delta_{s'_{(\gamma)},s_{(\pi(\gamma))}}\bigotimes_{\gamma=1}^\alpha W_i^{s_{(\gamma)},s'_{(\gamma)}},   
\end{equation}
and the algorithm proceeds as above: We define $L=L_1$ and iteratively compute $L\to L L_i$, for $i=2,\cdots,n$, as illustrated in Fig.~\ref{fig:tm_update}. The final contracted tensor is equivalent to $p_\alpha$. We note that cost to contract $L$ for MPOs is given by $O\left(n\alpha\chi^{\alpha+1}\right)$. 

\begin{figure}[htbp]
\centering
\subfigimg[width=0.49\linewidth]{a}{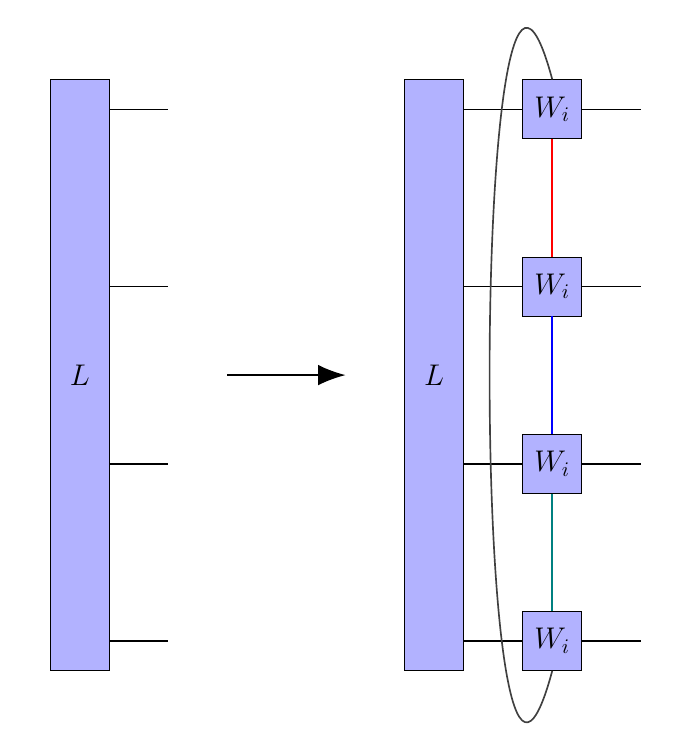}
\subfigimg[width=0.49\linewidth]{b}{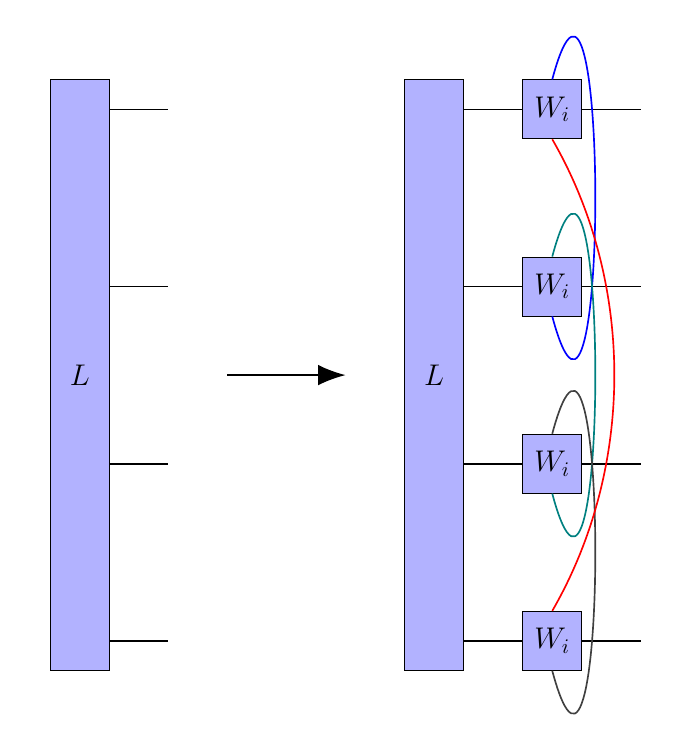}
\caption{Update of the tensor $L\to LL_i$ in the calculation of $p_4$ if (a) $i\in A$ or (b) $i\in B$.}
\label{fig:tm_update}
\end{figure}

\subsection{Realignment moments}
For the realignment moment, we first consider an MPO with the subsystem $A=\{1,2,\cdots,n_A\}$ and $B=A^\text{c}$. The MPO contraction for computing $r_4$ is sketched in Fig.~\ref{fig:mpo_p4_r4}(b). Viewing the MPO as an MPS with two physical indices, one can see that the contraction is equivalent to the purity of the MPS. Similarly, higher order $r_\alpha$ is equivalent to the higher order $\alpha-$R\'enyi entanglement entropy of the corresponding MPS. The latter can be obtained from the Schmidt coefficients of the MPS, which can be calculated using standard MPS routines with cost $O(n\chi^3)$ for any $\alpha$. Indeed, as has been discussed in Sec.~\ref{sec:ccnr_osee}, the Schmidt coefficients of the (normalized) MPO correspond to the singular values of the realignment matrix $R_\rho$. In this way, the CCNR negativity $\mathcal{C}$ can be exactly computed efficiently. 

Given MPS $\ket{\psi}$, we are interested in computing entanglement of the mixed state $\rho=\text{tr}_C(\ket{\psi})$, with subsystem $C$ connected to the right boundary of MPS.
Since the density matrix $\rho$ can be represented by an MPO with bond dimension $\chi^2$, where the local tensors are $U_i^{s_i,s_i'}=A_i^{s_i} \otimes \overline{A}_i^{s'_i}$, the CCNR negativity can be naively obtained with cost $O(n\chi^6)$. We now propose an algorithm with a reduced cost.

Let us compute entanglement for $\rho$ with bipartition $A=\{1,2,\cdots,n_A\}, B=\{n_A+1,n_A+2,\cdots,n_A+n_B\}$. We note that the algorithm can be adapted for the case where $A$ and $B$ are contiguous segments in the left and right parts of the MPS, respectively (see Appendix~\ref{sec:mps_pt_left_right}). First, we place the orthogonality center of the MPS on the site $n_A$. Next, we define the transfer matrix
\begin{equation}
    R_i = \sum_{\{s,s'\}}   A_i^{s} \otimes \overline{A}_i^{s'} \otimes A_i^{s'} \otimes \overline{A}_i^{s},
\end{equation}
for $i\in B$. We also define $R=\mathbb{1}_\chi \otimes \mathbb{1}_\chi$ and compute $R \to R_i R$ iteratively for $i=n_A+n_B,\cdots,n_A+1$, with cost $O(n\chi^5)$. The final tensor $R$ is a tensor with four indices $R_{\alpha,\alpha',\beta,\beta'}$. The tensor $L$ can be defined in a similar way from the left part of the MPS. Because the orthogonality center is placed at $n_A$, $L$ is simplified to $L_{\alpha,\alpha',\beta,\beta'}=\Lambda_{\alpha,\beta'} \Lambda_{\beta,\alpha'}$, where $\Lambda=\sum_s A_{n_A}^{s} \otimes \overline{A}_{n_A}^{s}$. Next, we compute
\begin{equation}
    \Gamma_{\alpha,\alpha',\gamma,\gamma'}=\sum_{\beta,\beta'}R_{\alpha,\alpha',\beta,\beta'}  \Lambda_{\beta,\gamma'} \Lambda_{\gamma,\beta'}.
\end{equation}
Writing $\Gamma$ as a $\chi^2 \times \chi^2$ matrix by combining indices $\Gamma_{(\alpha,\alpha'),(\gamma,\gamma')}$, one can see that
\begin{equation}
    r_\alpha = \text{tr}(\Gamma^{\alpha/2}),
\end{equation}
for even integer $\alpha$. Thus, the eigenvalues of $\Gamma$ are the squared singular values of $R_\rho$, and the CCNR negativity can again be exactly computed. The  diagonalization of $\Gamma$ can be done with cost $O(\chi^6)$, which results in the overall cost of $O(\chi^6 + n \chi^5)$. Note that, for $\alpha=4$, $r_4 =\text{tr}(\Gamma^{2})$ can be computed more efficiently by matrix multiplication, and thus the cost for computing $r_4$ is $O(n\chi^5)$.

\section{Mixed-state entanglement in many-body systems} \label{sec:num_results}
While entanglement of pure many-body systems has been routinely studied~\cite{amico2008}, characterizing it in mixed states has remained a significant challenge due to the inherent difficulty in computing mixed-state entanglement measures in many-body systems. Mixed states are ubiquitous, arising naturally from finite temperature, noise, or partial tracing. With our witnesses, we now study the entanglement in mixed many-body systems, demonstrating their potential for characterizing entanglement in various physical phenomena. In the following, we focus on the witnesses from the PT moments.

First, we consider the transverse-field Ising model
\begin{equation}\label{eq:ising}
    H_\text{TFIM}=-\sum_{k=1}^{n-1}\sigma^x_k\sigma^x_{k+1}-h\sum_{k=1}^n \sigma_k^z
\end{equation}
where $h$ is the transverse field and we impose open boundary condition. In this model, there is a transition from the ferromagnetic
phase to the paramagnetic phase at $h_c=1$. The critical point is governed by the Ising CFT with central charge $c=1/2$~\cite{DiFrancesco}. The TFIM possesses $\mathbb{Z}_2$ symmetry, generated by $\prod_i \sigma^z_i$.

As initial warm-up to understand our witnesses,  we begin with the pure ground state of $n=12$ sites, and we set $n_{A}=n_{B}= 5$ sites in the middle as the subsystem, where $A$ and $B$ are adjacent, contiguous subsystems. We show $\E,\Tilde{E}_3,\E_4$ in Fig.~\ref{fig:neg_L12_ising}(a), where the inequalities $\E_4,\E_4^{\text{SR}},\Tilde{\E}_2^{\text{SR}}\leq\E$ are confirmed. The entanglement is detected by all the witnesses $\Tilde{E}_3, \E_4,\E_4^{\text{SR}},\Tilde{\E}_2^{\text{SR}}$ in the interval $h \in[0,2]$, where the state is indeed entangled, as confirmed by positive $\E$. Next, we consider disconnected intervals $A$ and $B$ consisting of $n_{A}=n_{B}= 5$ sites on the left and right end of the chain, respectively. As shown in  Fig.~\ref{fig:neg_L12_ising}(b), the non-SR witnesses $\Tilde{\E}_3, \E_4$ fail to detect entanglement in a large interval, as it becomes negative for $h\gtrsim0.9$. The SR witnesses $\E_4^{\text{SR}},\Tilde{\E}_2^{\text{SR}}$, defined with respect to the $\mathbb{Z}_2$ symmetry, display significant improvement over the non-SR ones, only becoming slightly negative at $h\gtrsim1.5$, where the PT negativity is indeed close to zero.

Now, we subject  the ground state of the TFIM to noise, turning it into a mixed state, where each qubit is affected by single-qubit depolarisation noise $\Gamma(\rho)=(1-p)\rho+p \text{tr}_1(\rho) I_1/2$ with probability $p$. This channel preserves the $\mathbb{Z}_2$ symmetry of the systems, so that the SR witnesses can still be utilized.  Fig.~\ref{fig:neg_depo}(a) shows $\E_4$ for $n=6$ and $n_A=n_B=3$. We see that, for small enough $p$, $\E_4$ is effective in detecting the entanglement. Further, we compare various entanglement witnesses with $n=6,n_A=3$, and $p=0.05$ in Fig.~\ref{fig:neg_depo}(b). We again find that the SR witnesses are more reliable in detecting entanglement. Further, we show the witnesses as a function of $p$ in Fig.~\ref{fig:neg_thermal}(a). We find that the PT negativity vanishes above a value $p_c$, a phenomena reminiscent of sudden death that is known to occur at finite-temperature systems~\cite{sherman2016}. Indeed, we confirm that, at finite temperature, the PT negativity vanishes for $T\gtrsim1.5$ at $h=1.5$, as shown in Fig.~\ref{fig:neg_thermal}(b). Here, $n=10$ and $n_A=n_B=5$. Once again, we see the advantage of SR witnesses in witnessing entanglement in this case, becoming negative at higher temperature than the non-SR counterpart. At finite temperature, the PT negativity has been shown to follow an area law scaling~\cite{sherman2016}. Using the bounds shown in Sec. \ref{sec:pt_witness} and \ref{sec:sr_witness}, it immediately follows that our quantitative witnesses are bounded by the area law scaling.  

\begin{figure}[htbp]
\centering
\subfigimg[width=0.49\linewidth]{a}{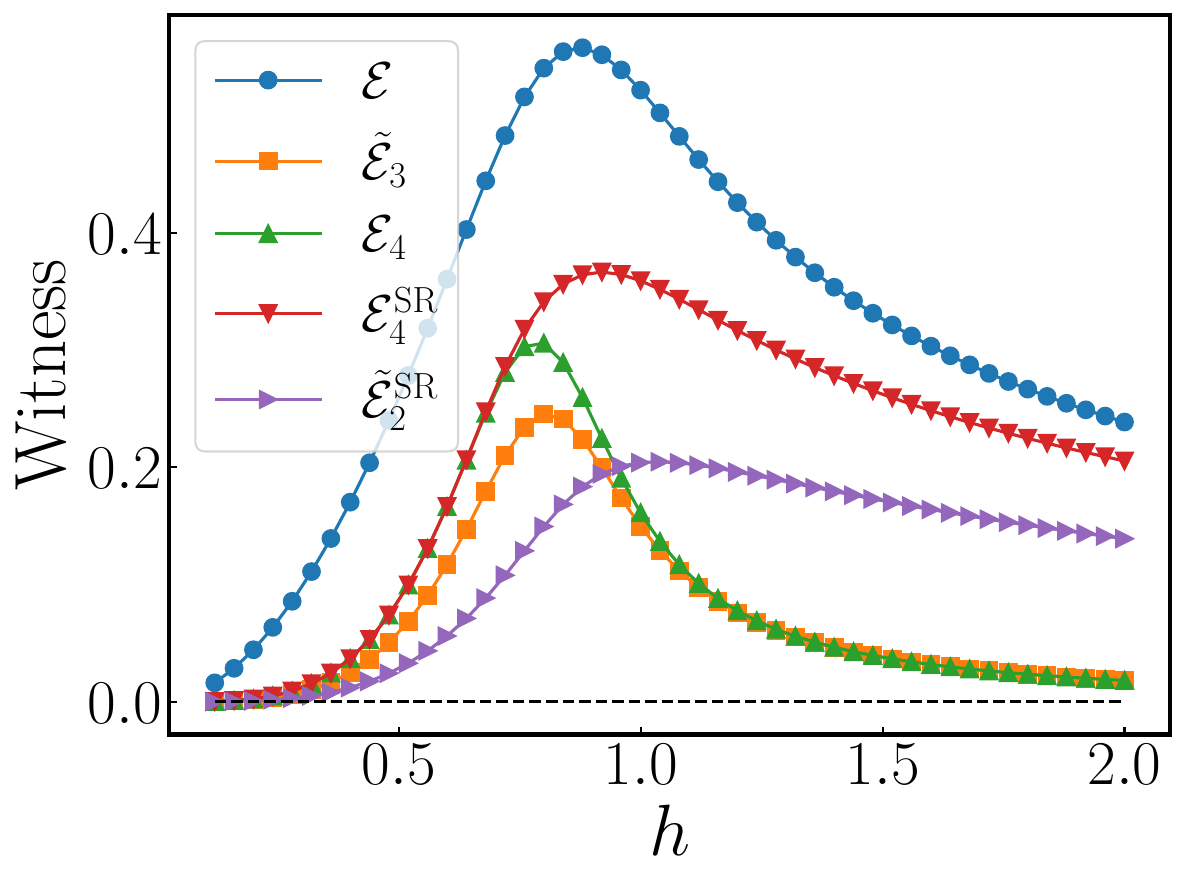}
\subfigimg[width=0.49\linewidth]{b}{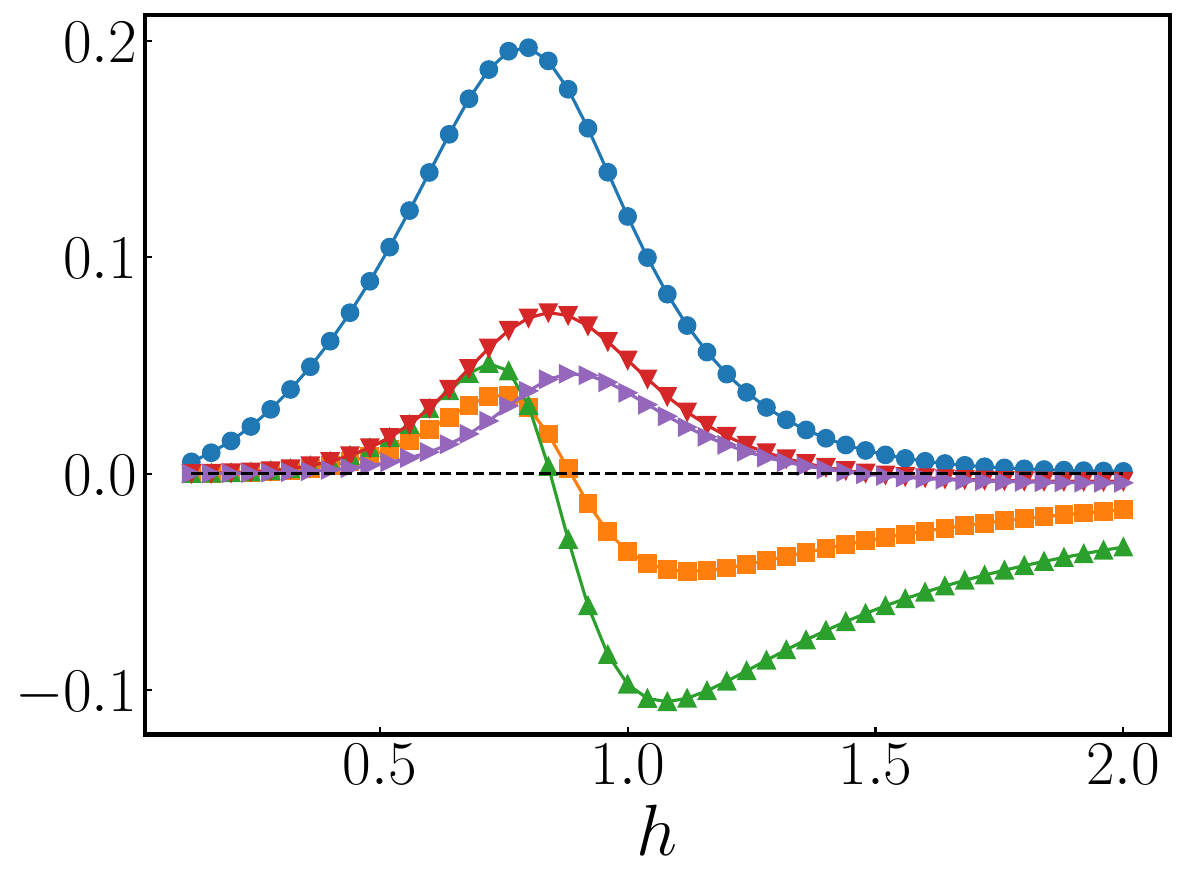}
\caption{Entanglement witnesses for the groundstate of the TFIM~\eqref{eq:ising} of $n=12$ sites and subsystem size $n_A=n_B=5$ for (a) connected and (b) disconnected partitions as a function of the transverse field $h$. }
\label{fig:neg_L12_ising}
\end{figure}

\begin{figure}[htbp]
\centering
\subfigimg[width=0.49\linewidth]{a}{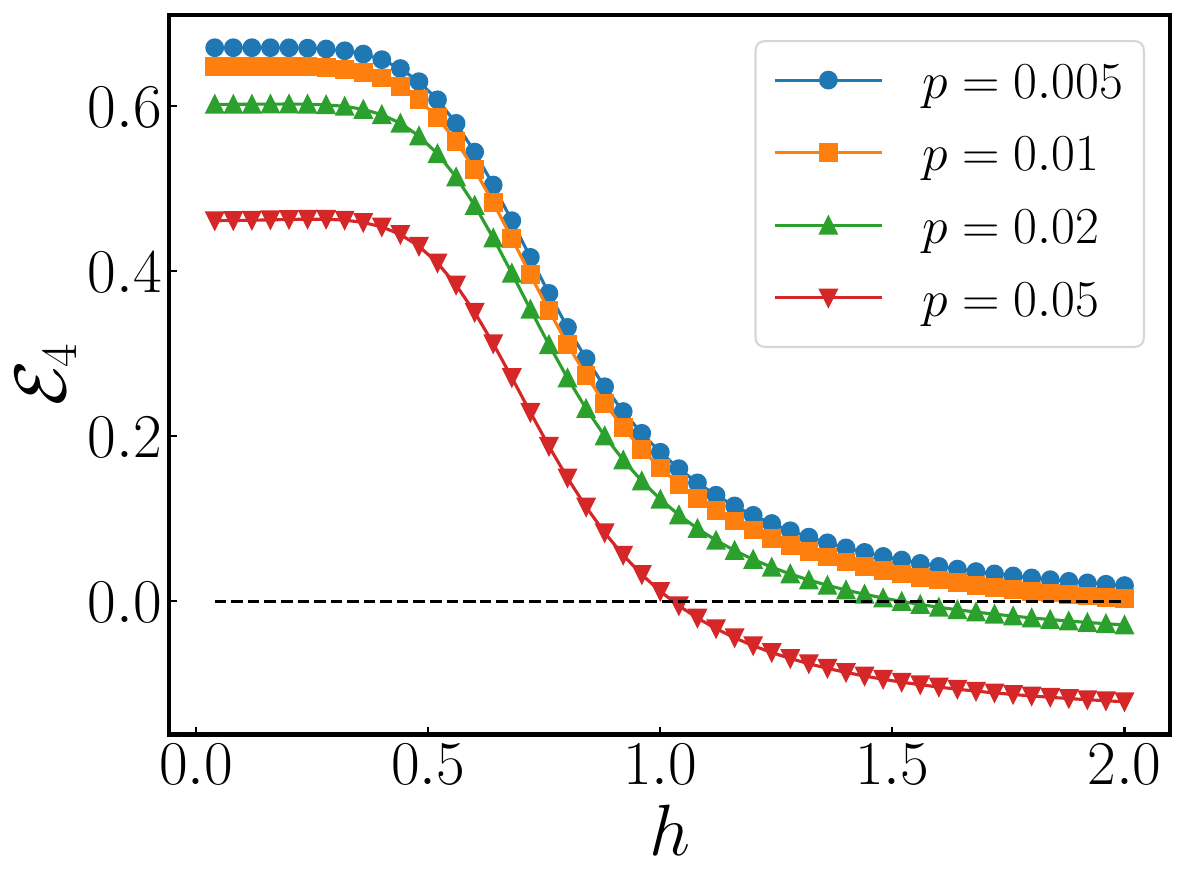}
\subfigimg[width=0.49\linewidth]{b}{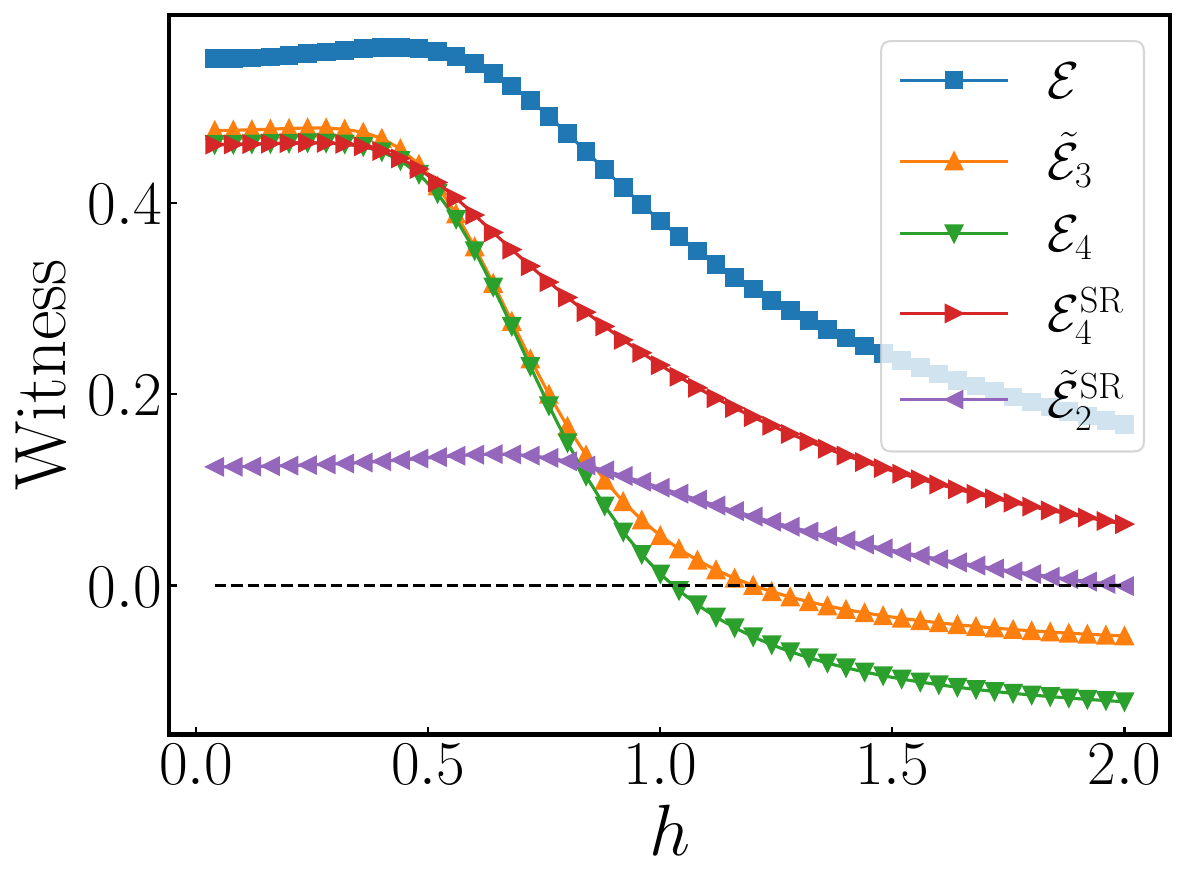}
\caption{ Entanglement witnesses for the groundstate of the TFIM~\eqref{eq:ising} of $n=6$ sites under depolarization noise. The partition size is $n_A=3$. In (a), we show $\E_4$ for different depolarisation probability $p$. For $p\in\{0.005,0.01\}$, $\E_4$ detects the entanglement at all interval of $h$ considered. In (b), we show different entanglement witnesses for $p=0.05$.}
\label{fig:neg_depo}
\end{figure}

\begin{figure}[htbp]
\centering
\subfigimg[width=0.49\linewidth]{a}{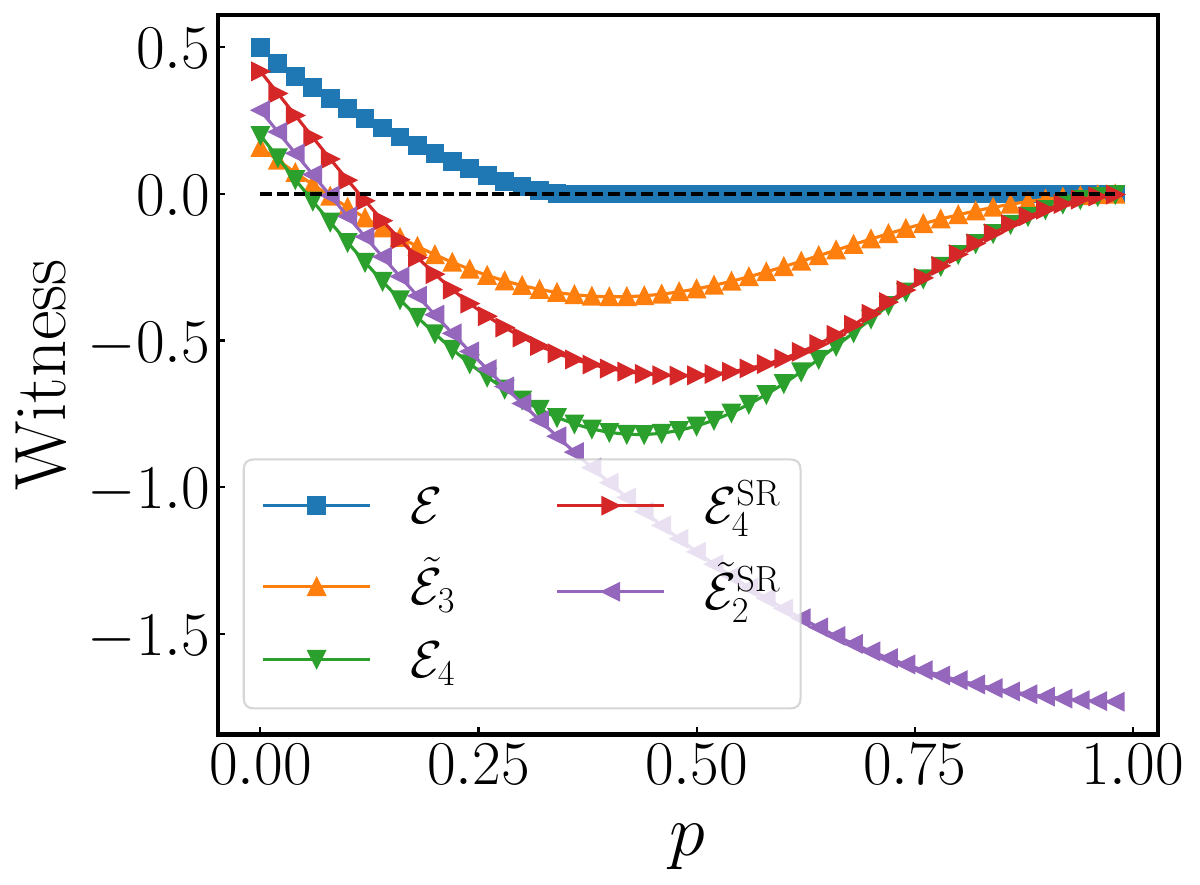}
\subfigimg[width=0.49\linewidth]{b}{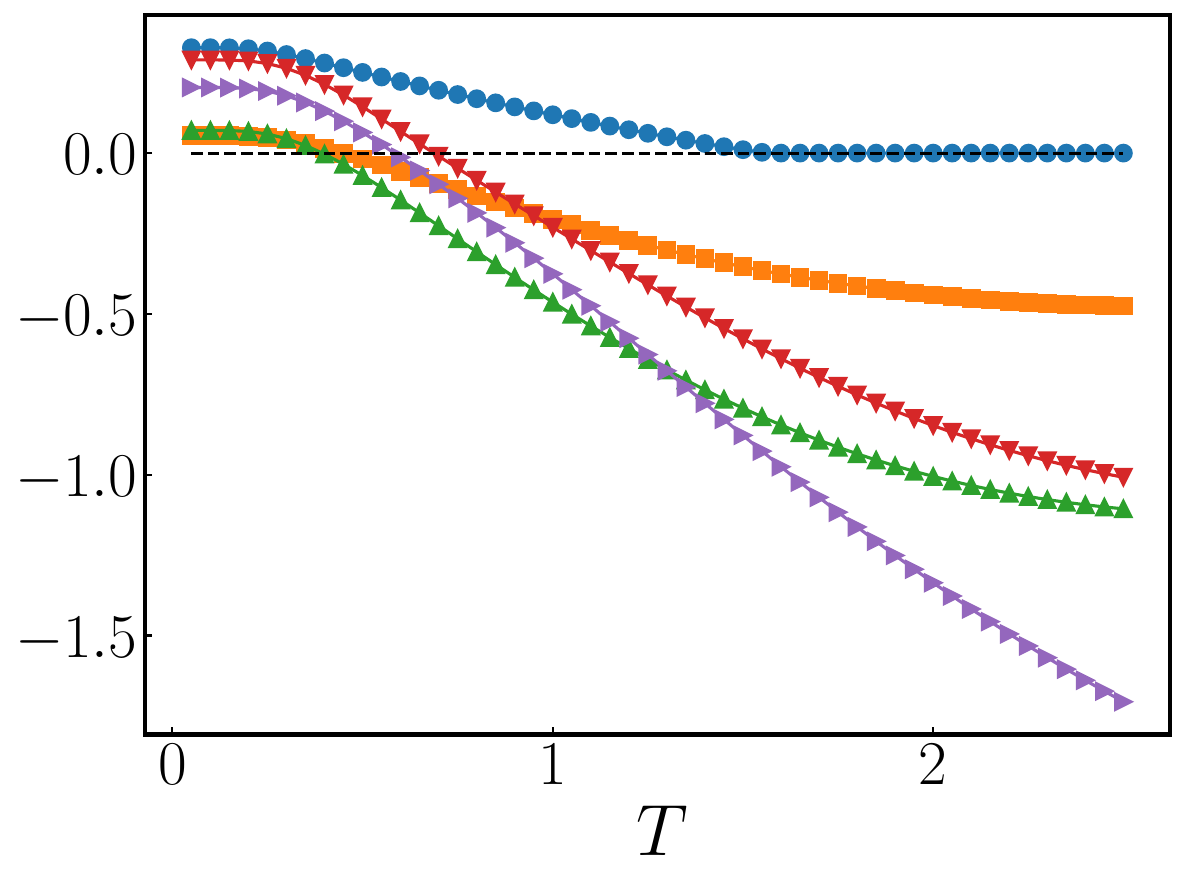}
\caption{ Entanglement witnesses for the TFIM~\eqref{eq:ising} (a) under depolarization noise with $n=6$ and $h=1$, and (b) at finite temperature $T$ with $n=10$ sites and $h=1.5$. The partition size is $n_A=n/2$. 
}
\label{fig:neg_thermal}
\end{figure}

We further illustrate the entanglement witnesses for the ground-state of the XXZ model, which Hamiltonian reads 
\begin{equation} \label{eq:xxz}
H_\text{XXZ}=-\sum_{k=1}^{n-1} (\sigma^x_k\sigma^x_{k+1}+\sigma^y_k\sigma^y_{k+1}+\Delta\sigma^z_k\sigma^z_{k+1})
\end{equation}
with the anisotropy $\Delta$. The model has a $U(1)$ symmetry related to magnetization conservation. We compute the ground state within the half-filling excitation symmetry sector $N_\text{p}=\sum_{k=1}^n\sigma^z_k=0$. The system hosts an antiferromagnetic phase when $\Delta<-1$, a critical phase for $\Delta \in [-1,1]$, and a ferromagnetic phase for $\Delta>1$. We will analyze the entanglement in the range $\Delta \in [-4,1]$.

The results for the tripartite case are shown in Fig.~\ref{fig:neg_L12_xxz}, with the same partitions as in the TFIM. In this case, all the witnesses considered can detect the entanglement of the ground states for all $\Delta$, both for connected and disconnected partitions.

\begin{figure}[htbp]
\centering
\subfigimg[width=0.49\linewidth]{a}{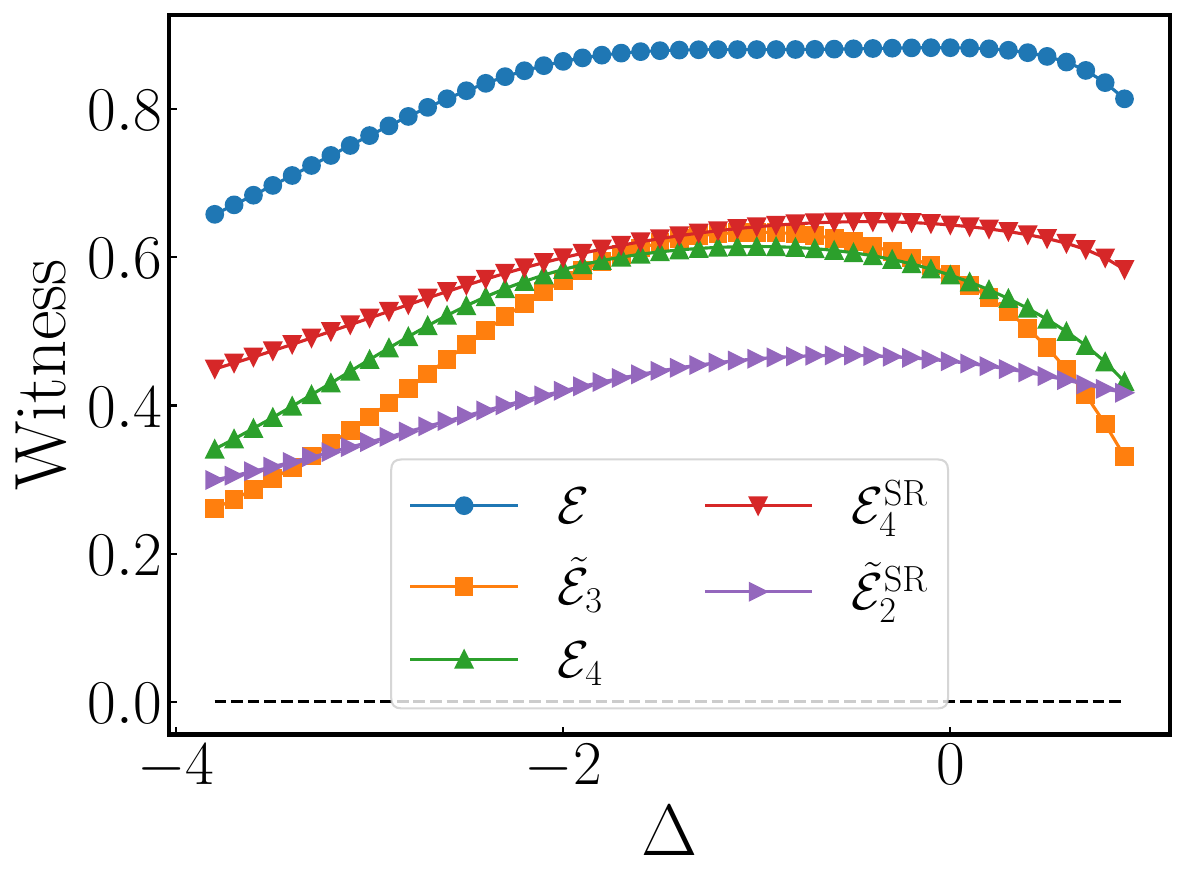}
\subfigimg[width=0.49\linewidth]{b}{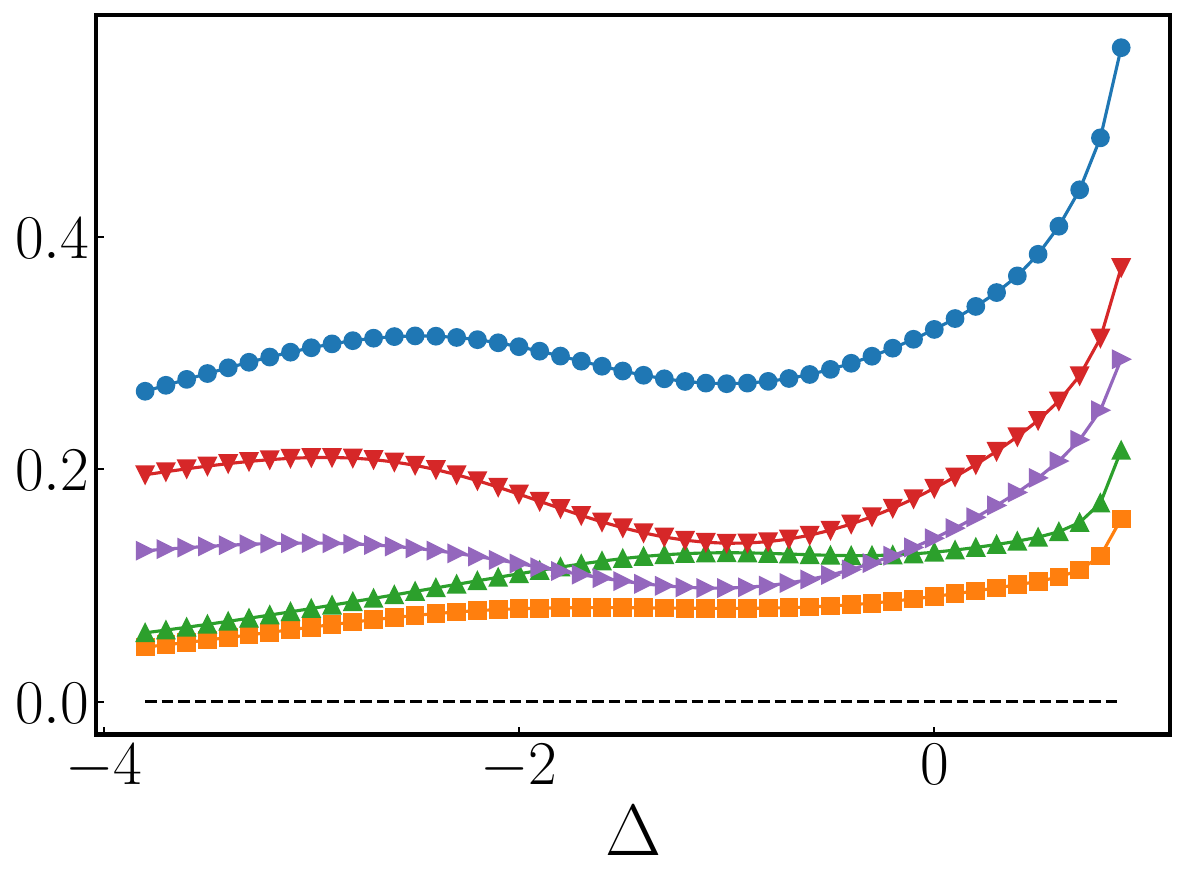}
\caption{Entanglement witnesses for the groundstate of the XXZ model~\eqref{eq:xxz} of $n=12$ sites and subsystem size $n_A=n_B=5$ for (a) connected and (b) disconnected partitions as a function of the anisotropy $\Delta$. }
\label{fig:neg_L12_xxz}
\end{figure}

\section{Discussion} \label{sec:discussion}
In this work, we introduced a generalization of entanglement Renyi entropies to mixed states, which we call the $p_\alpha$-negativity. For $\alpha=1$, it equals the PT negativity, while for $\alpha>1$ it serves as a quantitative witness for mixed state entanglement. In other words, they offer useful quantitative information about the degree of entanglement by rigorously bounding entanglement monotones~\cite{eisert2007quantitative}. Further, we showed that symmetries can be exploited to enhance the entanglement witnessing capabilities, improving over previous symmetry-resolved approaches. We also introduced similar witnesses based on the CCNR criterion~\cite{chen2003ccnr}, dubbed the $r_\alpha$-negativity. The CCNR negativity bounds the robustness of entanglement, thus making the witnesses quantitative. Notably, both $\mathcal{C}_4$ and $\mathcal{E}_4$ can be efficiently measured on quantum computers. In particular, $\mathcal{C}_4$ can be measured in experiment via SWAP tests, using only a single layer of CNOT and Hadamard gates acting on $4$ copies of $\rho$, while $\mathcal{E}_4$ requires a linear number of Toffoli gates. $\mathcal{C}_4$ requires only Clifford gates and constant-depth circuit to be measured via Bell measurements, which is an essential advantage for benchmarking early fault-tolerant quantum computers. This protocol can also be applied to efficiently measure the operator space entanglement entropy.

By leveraging our witnesses, we obtain several powerful results regarding entanglement properties across diverse scenarios.
We provide an algorithm to test whether a state has $O(\log n)$ or $\omega(\log n)$ entanglement, which is efficient as long as the state is weakly mixed, i.e. the 2-R\'enyi entropy scales at most as $S_2=O(\log n)$. This bound  is tight, as testing is inherently inefficient for $S_2=\omega(\log n)$~\cite{bansal2024pseudorandomdensitymatrices}. This test generalizes the efficient testing algorithm that was previously only known for pure states~\cite{bouland2022quantum}, to the mixed states setting. Similarly, we show that any pseudorandom density matrices with $S_2=O(\log n)$ must have PT negativity $\E=\omega(\log n)$. Our results also have direct implications on quantum cryptography: We show that the pseudoentanglement gap, i.e. the capability of low-entangled states to mimic high-entangled states, is fundamentally constrained for mixed states with bounded entropy $S_2=O(\log n)$, with the gap of $f(n)=\Theta(n)$ vs $g(n)=\omega(\log n)$, generalizing the previously known constraint for pure states. The optimal gap is only achieved for highly mixed states with $S_2=\omega(\log n)$, where one finds $f(n)=\Theta(n)$ vs $g(n)=0$~\cite{bansal2024pseudorandomdensitymatrices}. Thus, we resolve the question of the pseudoentanglement gap as a function of entropy, with the complete summary given in Tab.~\ref{tab:pseudoresource}.
While usually entropy in quantum states is associated negatively with noise and destroying quantumness, we find the opposite is true in quantum cryptography: Here, entropy is an important resource needed to hide information about quantum resources such as entanglement from eavesdroppers~\cite{haug2025pseudorandom}.
To completely hide entanglement, one requires $S_2=\omega(\log n)$, which we find is tight. These findings complement our recent work~\cite{haug2025efficientwitnessingtestingmagic}, where we found analogous results for the mixed-state nonstabilizerness.

We can also robustly certify the circuit depth needed to prepare a given quantum state. In particular, we can efficiently determine the minimum depth of the circuit even when the state is subjected to noise, assuming the final state is weakly mixed. Previously, such certification was only possible for pure states, i.e. in noiseless quantum circuits~\cite{hangleiter2023bell}. In addition, we provide an efficient algorithm to test the Schmidt rank for pure states and the operator Schmidt rank for mixed states. 

Our witness provides a significant simplification in determining the entanglement phase diagram of Haar random states: it can be entirely characterized by computing only the $p_4$-negativity. This is done by exploiting the fact that the $p_4$-negativity provides a lower bound for the PT negativity. By explicitly computing the Haar average of the $p_4$-negativity, we find that it saturates the upper bound of the PT negativity, which implies that the latter must also saturate the maximal value. Remarkably, this approach stands in contrast to the previous more intricate methods, which required calculating all PT moments and then analytically continuing to the PT negativity~\cite{shapourian2021entanglement}. As a direct consequence, we demonstrate that any $4$-design exhibits the same entanglement phase diagram as Haar random states. We note that entanglement transitions in Haar random states have been observed in experiment by full-state tomography~\cite{liu2023observation}, and our results provide a scalable route to observe the transitions in larger devices.

We further show that the Haar average of the $p_\alpha$-negativity is also close to its maximal value for any $\alpha\geq0$. This parallels the behavior of entanglement R\'enyi entropies in Haar random states, which are similarly close to their maximal values for any R\'enyi index $\alpha$. This further demonstrates the similarity between the $p_\alpha$-negativities and the entanglement R\'enyi entropies, highlighting the significant potential of the former to characterize mixed-state entanglement in many-body systems.

Surprisingly, we find that stabilizer states also share the same phase diagram, despite not forming exact $4$-designs~\cite{zhu2016clifford}. 
We also investigate the impact of noise on the entanglement of pure Haar random states. Our findings show that the $p_4$-negativity remains effective in capturing entanglement properties in this case. Notably, we observe that the volume-law entanglement survives for sufficiently large subsystems, even when subjected to exponentially strong noise.

Our work also makes progress on the computation of entanglement in extensive mixed many-body systems, which has been a challenge. Mixed states can be represented as MPOs or as purifications in form of MPS, both of bond dimension $\chi$.
For MPS of bond dimension $\chi$, we provide an efficient algorithm to compute PT moment $p_\alpha$  for integer $\alpha>1$ in $O(n\alpha \chi^{2\alpha+1})$, while for mixed states written as MPO it can be computed in $O(n\alpha\chi^{\alpha+1})$. Similarly, $\mathcal{C}$ can be computed in $O(n\chi^3)$ for MPO, while for MPS the cost is $O(n\chi^5)$ for computing the realignment moment $r_4$.
We show that our entanglement witnesses can be used to characterize the entanglement of ground states for the TFIM and XXZ model when subject to noise or when subsystems have been traced out.
Notably, we find that $\mathcal{E}_4$, which is efficient for MPS, characterizes entanglement nearly as well as the PT negativity $\mathcal{E}$, which is difficult to compute. 
Further, we find that our symmetry-resolved witnesses can detect entanglement significantly better than standard witnesses, both for discrete and continuous symmetries,  promising to be useful in symmetric many-body systems. 
Taken together, these results demonstrate the potential of $p_\alpha$-negativity and its symmetry resolution to characterize mixed-state entanglement in large many-body systems, where the PT negativity is no longer tractable. 

In summary, our work presents several significant results concerning mixed-state entanglement that were previously understood only in the context of pure states. In pure states, these insights often relied on entanglement R\'enyi entropies. We have been able to extend these findings to mixed states through the introduction of the $p_\alpha$-negativity, showcasing its role as a powerful generalization of entanglement R\'enyi entropies for mixed systems. While $p_\alpha$-negativity is not an entanglement monotone, its function as a quantitative entanglement witness enables it to effectively characterize mixed-state entanglement properties, as evidenced by our findings. Given the widespread utility of entanglement R\'enyi entropies in understanding pure-state entanglement, we anticipate that our witnesses will find further applications in deepening our understanding of mixed-state entanglement.

Our work opens several interesting avenues for future research. A key direction involves exploring the applications of these quantitative entanglement witnesses in the characterization of entanglement in many-body quantum systems. We posit that the $p_4$-negativity can serve as a crucial tool in the study of entanglement phases and properties in general mixed states. 
In particular, it would be interesting to extend the minimal membrane picture of entanglement spreading~\cite{nahum2017quantum} to non-complementary subsystems. Further, investigations into tripartite entanglement at measurement-induced criticality~\cite{sang2021entanglement}, and in noisy monitored circuits~\cite{weinstein2022measurement,liu2023universal,liu2024entanglement,qian2025protect} are also warranted. It would also be interesting to study whether $p_\alpha$-negativity can characterize finite-temperature phase transitions~\cite{lu2020structure} and detect topological order at finite temperature~\cite{lu2020detecting,hart2018entanglement}. Future research could also explore the construction of entanglement witnesses for fermionic system, where the PT negativity is defined based on a different notion of partial transpose~\cite{shapourian2017partial,shapourian2019entanglement}. Another interesting direction is to extend our ideas for the quantitative detection of multipartite entanglement~\cite{liu2022detecting}. 
Finally, we posit the question whether it is possible to construct proper mixed-state entanglement monotones that are also efficiently measurable. 

\begin{acknowledgments}
P.S.T. acknowledges funding by the Deutsche Forschungsgemeinschaft (DFG, German Research Foundation) under Germany’s Excellence Strategy – EXC-2111 – 390814868.
\end{acknowledgments}



\section*{Author contributions}
P.S.T. conceived the initial idea and carried out most of the analytical and numerical calculations, with the guidance of T.H. T.H. developed the analysis of pseudorandom density matrices and pseudoentanglement, and the measurements of the witnesses. Both authors discussed the results and contributed to the writing of the manuscript.


\bibliographystyle{plainnat}
\bibliography{biblio}

\clearpage 

\onecolumngrid

\let\addcontentsline\oldaddcontentsline

\appendix

\setcounter{figure}{0}

\renewcommand{\thesection}{\Alph{section}}
\renewcommand{\thesubsection}{\arabic{subsection}}
\renewcommand*{\theHsection}{\thesection}

\clearpage
\begin{center}

\textbf{\large \SMLong{}}
\end{center}

\makeatletter

\renewcommand{\thefigure}{S\arabic{figure}}

We provide proofs and additional details supporting the claims in the main text.

\makeatletter
\@starttoc{toc}

\makeatother

\section{Efficient algorithm to measure PT negativity moment}\label{sec:efficientWitness}
Here, we detail the algorithm to measure $\mathcal{E}_4(\rho)=\frac{1}{2}[-\ln p_4 +3\ln p_2]$. In particular, we give an efficient algorithm to estimate $p_4$, following Refs.~\cite{carteret2005noiseless,gray2018machinelearning}. The circuit to measure is sketched in Fig.~\ref{fig:SWAP}.

\begin{fact}[Efficient PT moment]\label{thm:entanglement_sup}
For a given (mixed) $n$-qubit state $\rho$, there exists an efficient algorithm to measure $p_4(\rho)=\operatorname{tr}(\rho^\Gamma)^4$, where $\rho^{\Gamma} = (I \otimes \Gamma) \rho$ is the partial transpose,  to additive precision $\epsilon$ with failure probability $\delta$ using $O(\epsilon^{-2}\log(2/\delta))$ copies of $\rho$, $O(n)$ circuit depth, $O(n)$ Toffoli gates, $1$ auxiliary qubit,  and $O(\epsilon^{-2}\log(2/\delta))$ classical post-processing time. 
\end{fact}
\begin{proof}
We measure the PT moment $p_4$ with respect to complementary bipartition $A$, $B$. Then, one can write~\cite{gray2018machinelearning}
\begin{equation}
    p_4=\text{tr}(\rho^\Gamma)^4=\text{tr}(S_A^{3,4}S_A^{2,3}S_A^{1,2}\otimes S_B^{1,2}S_B^{2,3}S_B^{3,4}\rho^{\otimes 4})
\end{equation}
where 
\begin{equation}
    S_A^{i,j}=\prod_{k\in A}({I+\sigma^x_{i,k}\sigma^x_{j,k}+\sigma^y_{i,k}\sigma^y_{j,k}+\sigma^z_{i,k}\sigma^z_{j,k}})
\end{equation} 
swaps copies $i$ and $k$ on bipartition $A$ where $\sigma^\alpha_{j,k}$ is the $\alpha\in\{x,y,z\}$ Pauli operator acting on copy $j$ and qubit $k$. Note that $S_A^{i,j}$ is simply a product of individual SWAP operations on each qubit pair of the bipartition. 
Now, $p_4$ can be measured by performing the SWAP test with operator~\cite{gray2018machinelearning}
\begin{equation}\label{eq:P4}
    P_4^{T_B}=S_A^{3,4}S_A^{2,3}S_A^{1,2}\otimes S_B^{1,2}S_B^{2,3}S_B^{3,4}\,.
\end{equation} 
Here, one implements a controlled version of $P_4^{T_B}$, where the control acts on an auxiliary qubit with Hadamard gates before and after the control. As $P_4^{T_B}$ can be decomposed as a product of $3n$ two qubit SWAPs, the controlled version of $P_4^{T_B}$ can be implemented by $3n$ controlled SWAPs, which are applied consecutively on different qubit pairs of the copies. 
Then, the state of the ancilla is measured, with $p_4=\langle \sigma^z\rangle_\text{aux}$ being the expectation value of $\sigma^z$ of the ancilla. 

Now, we can bound the estimation of $\langle \sigma^z\rangle_\text{aux}$ in terms of additive error $\epsilon$, failure probability $\delta$, number of measurements $L$ and range of measurement outcomes $\Delta\lambda$ via Hoeffding's inequality
\begin{equation}\label{eq:Hoeffding2}
    \delta\equiv\text{Pr}(\vert \hat{\sigma}_z - \langle \sigma^z\rangle\vert\ge \epsilon)\le 2\exp\left(-\frac{2\epsilon^2 L}{\Delta\lambda^2}\right)\,.
\end{equation}
Now, we have $\Delta\lambda=2$ as $\sigma^z$ has eigenvalues $\pm1$. Then, by inverting, we get
\begin{equation}
    L \leq 2\epsilon^{-2}\log(2/\delta)\,.
\end{equation}
As each measurement requires $4$ copies of $\rho$, we need in total $O(\epsilon^{-2}\log(2/\delta))$ copies. The post-processing only involves counting outcomes, and thus we have $O(\alpha \epsilon^{-2}\log(2/\delta))$ post-processing time. Finally, we have to implement $3n$ controlled SWAPs, each acting with control on the ancilla and a pair of qubits, which can be implemented on $O(n)$ circuit depth. SWAP gates require $3$ CNOTs, and controlled SWAP gates are implemented using $3$ Toffoli gates. Thus, in total we require $9n$ Toffoli gates and $2$ Hadamard gates.
\end{proof}

We note that the algorithm requires coherent control over $4$ copies at the same time, but it does not require tomographically complete measurements or reconstruction of the density matrix.
Note that the other component of $\mathcal{E}_4$, namely $p_2=\text{tr}(\rho^2)$, can be efficiently measured via SWAP tests or Bell measurements with even lower complexity, requiring only coherent control over $2$ copies~\cite{ekert2002direct,garcia2013swap}.

\begin{figure}[htbp]
	\centering	
	\subfigimg[width=0.4\textwidth]{}{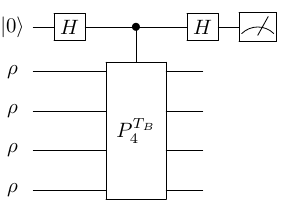}
    \caption{SWAP test to measure expectation value of PT moment $P_4^{\text{T}_B}$ of~\eqref{eq:P4} for computing $4$th moment of PT negativity $\mathcal{E}_4(\rho)$ for state $\rho$. Uses $4$ copies of $\rho$ and a single ancilla. Control-$P_4^{\text{T}_B}$  can be implemented as a sequence of control-SWAP operations between two qubits, using $O(n)$ depth using $O(n)$ Toffoli gates in total.
	}
	\label{fig:SWAP}
\end{figure}

\section{Efficient algorithm to measure CCNR negativity moment}\label{sec:CCNRmeas}
We now study the measurement complexity of the CCNR negativity moment
\begin{equation}\label{eq:r4_sup}
    \mathcal{C}_4(\rho) =\frac{1}{2}[-\ln(r_4)-3 S_2(\rho)]\,,
\end{equation}
with 
\begin{equation}
r_4=\text{tr}(S_A^{1,2}S_A^{3,4}\otimes S_B^{1,3}S_B^{2,4}\rho^{\otimes 4})\,.
\end{equation}
Notably, the different SWAPs $S_A^{i,j}$ act on different qubits. Thus, in contrast to the PT moments, one can measure $r_4$ using destructive Bell measurements instead of SWAP test~\cite{garcia2013swap}.
We show the measurement setup using Bell measurements in Fig.~\ref{fig:Bell}.
\begin{figure}[htbp]
	\centering	
	\subfigimg[width=0.3\textwidth]{}{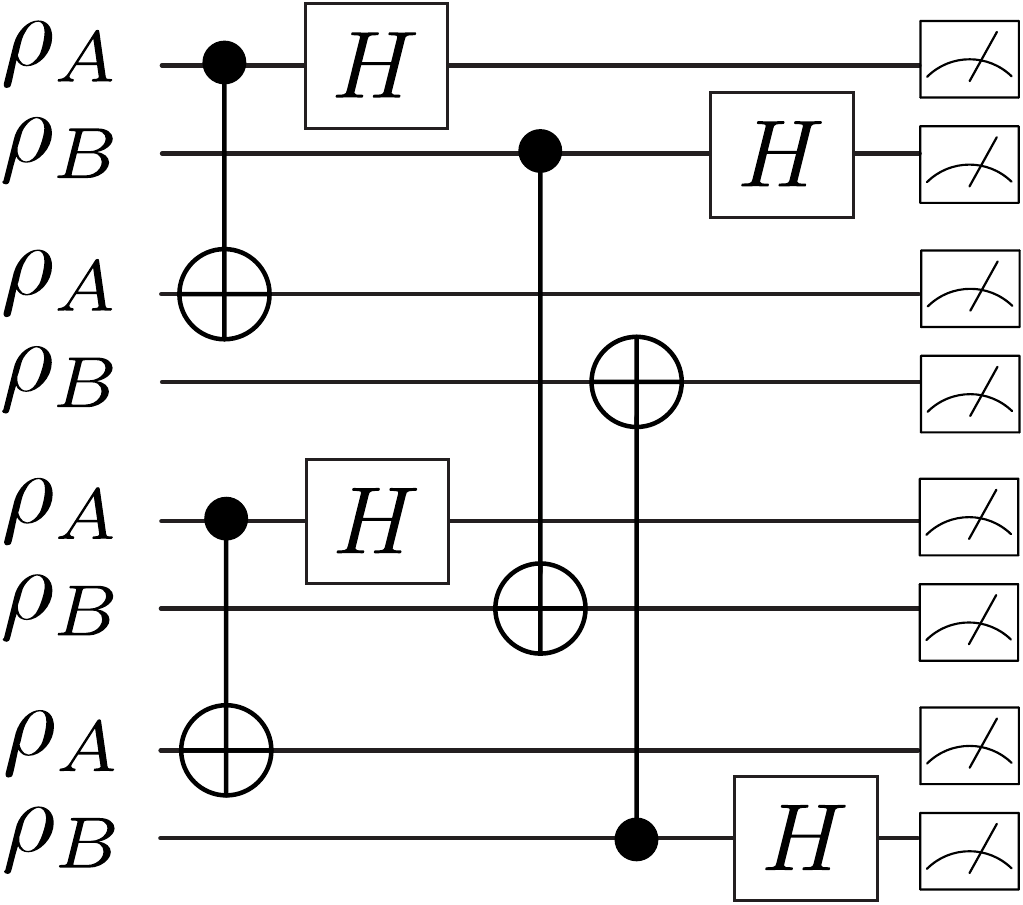}
    \caption{Bell measurements to  measure expectation value of CCNR negativity moment $r_4$ for state $\rho$. Uses $4$ copies of $\rho$ and $O(1)$ depth. 
	}
	\label{fig:Bell}
\end{figure}

\begin{proposition}[Efficient CCNR negativity moment]\label{thm:CCNR_sup}
For a given (mixed) $n$-qubit state $\rho$, there exists an efficient algorithm to measure $r_4(\rho)$ to additive precision $\epsilon$ with failure probability $\delta$ using $O(\epsilon^{-2}\log(2/\delta))$ copies of $\rho$, $O(1)$ circuit depth, $O(1)$ CNOT and Hadamard gates,  and $O(N\epsilon^{-2}\log(2/\delta))$ classical post-processing time. 
\end{proposition}
\begin{proof}
This follows directly from the fact that $r_4=\text{tr}(S_A^{1,2}S_A^{3,4}\otimes S_B^{1,3}S_B^{2,4}\rho^{\otimes 4})$ can be implemented by a tensor product of SWAPs $S_A^{1,2}$, each acting on a different set of qubits. It is known that the Bell measurements correspond to measurements in the eigenbasis of the SWAP operator. Thus, from Bell measurements as shown in Fig.~\ref{fig:Bell} and classical post-processing one can compute $r_4$. 
In particular, for every qubit pair where one applies the CNOT and Hadamard gate, assign $-1$ to outcome $\ket{11}$, and $+1$ else. Then, multiply all mapped $\pm1$ together, which is the estimator of the expectation value. As each outcome has range $\pm1$ and thus $\Delta\lambda=2$, the dependence on $\epsilon$ and $\delta$ follows from Hoeffding's inequality similar to Eq.~\eqref{eq:Hoeffding2}.
\end{proof}

\section{Upper bounds for the PT negativity and the CCNR negativity} \label{sec:upper_bounds_neg}
Here, we prove upper bounds for the PT negativity and the CCNR negativity. First, we show an inequality that holds for the $p_\alpha$-negativity:
\begin{equation} \label{eq:bound_e0}
    \E_\alpha(\rho) \leq \frac{1}{2}(n \ln 2 - S_2(\rho)),
\end{equation}
for any $\alpha\geq 0$. This can be directly shown using the monotonicity of the $p_\alpha$-negativity $\E_0(\rho)\geq\E_\alpha(\rho)$. An upper bound for the PT negativity is thus obtained by setting $\alpha=1$. 

Next, we have
\begin{equation} \label{eq:bound_dim}
    \E(\rho) \leq m \ln 2,
\end{equation}
where $m=\min(n_A,n_B)$. To see this, we will show that 
\begin{equation} \label{eq:bound_s1/2}
    \E(\rho) \leq \min(S_{1/2}(\rho_A), S_{1/2}(\rho_B)),
\end{equation}
which immediately leads to Eq.~\eqref{eq:bound_dim} since $S_{1/2}(\rho_{A(B)})\leq n_{A(B)} \ln 2$. We define
\begin{equation} \label{eq:purification}
    \ket{\sqrt{\rho}} = \sum_i (\rho^{1/2} \ket{i}) \otimes \ket{i},
\end{equation}
which is a purification of $\rho$ in $\mathcal{H}_A\otimes\mathcal{H}_B\otimes\mathcal{H}_{A'}\otimes\mathcal{H}_{B'}\simeq\mathbb{C}^{2^n} \otimes \mathbb{C}^{2^n} $. It is also known as the canonical purification~\cite{dutta2021canonical}; however, the particular purification does not matter in this context. The PT negativity of the subsystem $A$ in $\ket{\sqrt{\rho}}$ is given by $S_{1/2}(\rho_A)$ (since $\ket{\sqrt{\rho}}$ is a pure state) , while the original PT negativity $\E(\rho)$ is obtained by tracing out $A'B'$ in $\ket{\sqrt{\rho}}$. Due to the monotonicity of PT negativity with respect to partial tracing~\cite{plenio2005negativity}, we obtain $\E(\rho) \leq S_{1/2}(\rho_A)$.  By exchanging $A$ and $B$, we also obtain $\E(\rho) \leq S_{1/2}(\rho_B)$, thus proving Eq.~\eqref{eq:bound_s1/2}.

While the upper bound in Eq.~\eqref{eq:bound_dim} does not apply for $\E_\alpha$ with $\alpha<1$, we can give an upper bound for $p_0$ as
\begin{equation} \label{eq:bound_p0}
    \ln p_0 \leq 2m\ln 2 + \ln r,
\end{equation}
where $r=\text{rank}(\rho)$. For $r=1$, i.e. $\rho$ is a pure state, this can easily be shown using the Schmidt decomposition. To show the inequality for general $r$, note that $\rho$ can be purified using an ancilla with physical dimension $r$. We denote a purification of $\rho$ as $\ket{\psi_\rho}\in \mathbb{C}^{2^n} \otimes \mathbb{C}^{r}$. We can represent $\ket{\psi_\rho}$ as an MPS, where the ancilla, that we now denote as subsystem $C$, is placed in the middle of the chain, separating the subsystems $A$ and $B$, which are placed to the left and right of $C$, respectively. In this case, the PT spectrum is stored in the PT of the transfer matrix of $C$, denoted by $E_C^\Gamma$ (see Appendix~\ref{sec:mps_pt_left_right}). One can see that the maximal size of this matrix is $(2^{S_0(\rho_A)} \times 2^{S_0(\rho_B)} ) \times ( 2^{S_0(\rho_A)} \times  2^{S_0(\rho_B)})$, so that the maximal rank of $\rho^\Gamma$ is
\begin{equation} \label{eq:bound_p0_2}
    \ln p_0 \leq S_0(\rho_A) + S_0(\rho_B).
\end{equation}
Finally, Eq.~\eqref{eq:bound_p0} follows by noticing that $S_0(\rho_B)\leq S_0(\rho_A)+\ln r$ and vice versa.

It is worth noting that Eq.~\eqref{eq:bound_p0_2} also implies
\begin{equation}
    \E_\alpha(\rho) \leq \frac{1}{2}(S_0(\rho_A) + S_0(\rho_B) - S_2(\rho))
\end{equation}
by monotonicity of $p_\alpha$-negativity. The term on the right hand side has a form that resembles the mutual information $I_\alpha(\rho_{AB})=S_\alpha(\rho_A)+S_\alpha(\rho_B)-S_\alpha(\rho_{AB})$, although with different R\'enyi indices for $S_\alpha(\rho_{A(B)})$ and $S_\alpha(\rho_{AB})$. Note that the mutual information is not a measure of mixed-state entanglement, as it is sensitive to classical correlations. If the entanglement spectrum is flat (such that all R\'enyi entropies are identical), this gives a direct relation between the PT negativity and the mutual information as $\E(\rho)\leq \frac{1}{2}I(\rho)$. The latter inequality has been shown, e.g., for stabilizer states~\cite{sang2021entanglement}.


For the $r_\alpha$-negativity, we have
\begin{equation} \label{eq:bound_e0_sup}
    \mathcal{C}_\alpha(\rho) \leq m \ln 2 - \frac{1}{2}S_2(\rho),
\end{equation}
for any $\alpha\geq 0$, and we recall $m=\min(n_A,n_B)$. This can again be shown using the monotonicity of the $r_\alpha$-negativity $\mathcal{C}_0(\rho)\geq\mathcal{C}_\alpha(\rho)$. An upper bound for the CCNR negativity is obtained by setting $\alpha=1$. 

\section{Upper bound on $D_\text{F}$} \label{sec:proof_df}
Here, we prove the bound on the \gEnt{} $D_\text{F}$ as
\begin{equation} \label{eq:bound_df_sup}
    S_\alpha(\rho_{A(B)}) \geq D_\text{F}(\rho)\,,
\end{equation}
for any $\alpha\geq0$.

First, by the Uhlman's theorem~\cite{nielsen2011quantum}, we can write
\begin{equation} 
    D_\text{F}(\rho)=\min_{\sigma\in \text{SEP}} \min_{\ket{\psi_\sigma}} -\ln \vert \braket{\sqrt{\rho}}{\psi_\sigma} \vert^2
\end{equation}
where $\ket{\sqrt{\rho}}$ is the canonical purification of $\rho$ defined in Eq.~\eqref{eq:purification} and $\ket{\psi_\sigma}$ is a purification of $\sigma$ in $\mathcal{H}_A\otimes\mathcal{H}_B\otimes\mathcal{H}_{A'}\otimes\mathcal{H}_{B'}\simeq\mathbb{C}^{2^n} \otimes \mathbb{C}^{2^n} $. Consider the Schmidt decomposition for the state $\ket{\sqrt{\rho}}$ in the subsystem $A$:
\begin{eqnarray}
    \ket{\sqrt{\rho}} = \sum_i \sqrt{\lambda_i} \ket{\phi_i}_A \otimes \ket{\phi_i}_{A'BB'},
\end{eqnarray}
where $\lambda_i$ are the eigenvalues of $\rho_A=\Tr_B(\rho)$, where $\lambda_1$ is the largest eigenvalue. The state $ \Tr_{A'B'}(\ket{\phi_1}_A \otimes \ket{\phi_1}_{A'BB'} \bra{\phi_1}_A \otimes \bra{\phi_1}_{A'BB'}) = \ket{\phi_1}_A \bra{\phi_1}_A \otimes \Tr_{A'B'}(\ket{\phi_1}_{A'BB'} \bra{\phi_1}_{A'BB'})$ is a separable state in $\mathcal{H}_A\otimes\mathcal{H}_B$. Thus, we have
\begin{equation} 
    D_\text{F}(\rho)\leq -\ln \vert \braket{\sqrt{\rho}}{\ket{\phi_1}_A \otimes \ket{\phi_1}_{A'BB'}} \vert^2 = -\ln \lambda_1.
\end{equation}
Note that $-\ln\lambda_1 = \lim_{\alpha\to\infty}S_{\alpha}(\rho_A)$. Thus, by the hierarchy of R\'enyi entropy, we have
\begin{equation}
S_\alpha(\rho_A) \geq D_\text{F}(\rho),
\end{equation}
for any $\alpha\geq0$. By exchanging $A$ and $B$, we obtain a similar inequality with $S_\alpha(\rho_B)$. 

\section{Relationship between PT negativity, CCNR negativity, and robustness of entanglement} \label{sec:neg_rob}
Here, we prove the inequality between the PT negativity and robustness of entanglement:
\begin{equation} \label{eq:neg_rob_sm}
    \E(\rho) \leq \ln(2\mathcal{R}(\rho)+1).
\end{equation}
 Note that this has previously been proven in Ref.~\cite{vidal2002computable}, and here we give an alternative proof. Let $s$ be the optimal value such that $\rho=(1+s)\rho_+ -s\rho_-$ for $\rho_+,\rho_- \in \text{SEP}$, so that $\mathcal{R}(\rho)=s$. We have $\rho^\Gamma=(1+s)\rho^\Gamma_+ -s\rho^\Gamma_-$, with $\rho^\Gamma_+,\rho^\Gamma_- \in \text{SEP}$, yielding
\begin{equation}
    \lVert \rho^\Gamma \rVert_1 \leq (1+s)\lVert \rho^\Gamma_+ \rVert_1 + s \lVert \rho^\Gamma_- \rVert_1 = 2s+1.
\end{equation}
In the first inequality, we used the triangle inequality on the trace norm $\lVert \cdot \rVert_1$, and the second equation is due to $\lVert \rho_s \rVert_1=1$ for $\rho_s \in \text{SEP}$. Taking the logarithm, we obtain Eq.~\eqref{eq:neg_rob_sm}.

Similarly, for the CCNR negativity, we have 
\begin{equation}
    \lVert R_\rho\rVert_1 \leq (1+s)\lVert R_{\rho_+} \rVert_1 + s \lVert R_{\rho_-} \rVert_1 \leq 2s+1,
\end{equation}
where we used that $\lVert R_{\rho} \rVert_1\leq1$ for separable states.

\section{Relationship between PT negativity, CCNR negativity, and the Schmidt number} \label{sec:neg_schmidt}
Here, we prove the inequality
\begin{equation} \label{eq:neg_schmidt_sm}
    \E(\rho) \leq \ln s(\rho).
\end{equation}.
Let $\rho=\sum_i p_i \ket{\psi_i}\bra{\psi_i}$ be the optimal pure-state decomposition such that $s(\rho)=\max_i s(\ket{\psi_i})$. We have
\begin{equation} \label{eq:b}
\begin{split}
    \lVert \rho^\Gamma \rVert_1 &\leq \sum_i p_i \lVert (\ket{\psi_i}\bra{\psi_i})^\Gamma \rVert_1\\
    &= \sum_i p_i e^{S_{1/2}(\ket{\psi_i})} \\
    &\leq \sum_i p_i s(\ket{\psi_i}) \\
    &\leq s(\rho).
\end{split}
\end{equation}
Here, we used the triangle inequality in the first line. In the second line, we used that $\E(\rho)=S_{1/2}(\ket{\psi})$ for pure states $\rho=\ket{\psi}\bra{\psi}$.  In the third line, we used the monotonicity of entanglement R\'enyi entropy.

It is easy to see that the argument above immediately extends to the CCNR negativity, thus we also have
\begin{equation} \label{eq:ccnr_schmidt}
    \mathcal{C}(\rho)\leq  \ln s(\rho)\,.
\end{equation}
This was previously shown in Ref.~\cite{johnston2013duality} through a different technique.

\section{Depth-dependent bound on the PT negativity and the CCNR negativity} \label{sec:depth_bound}
Here, we prove the depth-dependent bound:
\begin{equation} \label{eq:depth_bound_sm}
    \mathcal{E}(\rho)\leq  d\vert \partial A\vert\ln 2\,.
\end{equation}
To this end, we recall that the Schmidt number $s(\rho)$ is an entanglement monotone for mixed states~\cite{terhal2000schmidt}. Since $S_0(\ket{\psi})\leq d\vert \partial A\vert \ln{2}$ in noiseless quantum circuits~\cite{hangleiter2023bell}, it follows that we have the bound
\begin{equation} \label{eq:a}
    s(\rho) \leq 2^{d\vert \partial A\vert},
\end{equation}
in the case of noisy circuits.
Now, combining Eq.~\eqref{eq:a} and Eq.~\eqref{eq:neg_schmidt_sm}, we obtain Eq.~\eqref{eq:depth_bound_sm}.
Finally, we note that by Eq.~\eqref{eq:ccnr_schmidt}, we can also show
\begin{equation} \label{eq:depth_bound_sup_ccnr}
    \mathcal{C}(\rho)\leq  d\vert \partial A\vert\ln 2\,.
\end{equation}

\section{Certification of circuit depth}\label{sec:depthcert}
Here, we give bounds on the number of measurements needed to certify the circuit depth $d$ for preparing mixed states:
\begin{proposition}[Efficient certification of circuit depth]
    Given $n$-qubit mixed state $\rho$, there exists an efficient quantum algorithm to certify the minimal circuit depth $d$ needed to generate the state using a fixed architecture where $\vert\partial A\vert$ noisy two-qubit gates act across an equal-sized bipartition. In particular, for $c>0$ and $L=n^{2c+1}$ SWAP tests to measure the estimators $\hat{p}_4(\rho)$ and $\hat{p}_2(\rho)$, we have
    \begin{equation}
        d\geq \frac{1}{2\vert\partial A\vert\ln2}[-\ln(\hat{p}_4(\rho)+n^{-c})+3\ln(\hat{p}_2(\rho)-n^{-c})]
    \end{equation}
    with exponentially small failure probability.
\end{proposition}
\begin{proof}
    We have shown that the depth is bounded $d \geq \frac{1}{\vert \partial A\vert\ln 2}\mathcal{E}_4(\rho)$ in Sec.~\ref{sec:depth_bound}. Now, it remains to bound the estimation error for $\mathcal{E}_4(\rho)=\frac{1}{2}[-\ln(p_4(\rho))+3\ln(p_2(\rho))]$. Both $p_4$ and $p_2$ can be measured efficiently using SWAP tests, where one has additive precision $\epsilon$ when using $O(\log(\epsilon^{-2})$ samples.
    Now, we bound the number of SWAP tests $L$ needed to estimate $\mathcal{E}_4(\rho)$ and thus $d$.

    Now, let us assume we have estimator $\hat{p}_4$ after $L$ SWAP tests. Then, from Hoeffding's inequality, we have
    \begin{equation}
    \delta\equiv\text{Pr}(\vert \hat{p}_4 - p_4\vert\ge \epsilon)\le 2\exp\left(-\frac{\epsilon^2 L}{2}\right)\,,
\end{equation}
which implies
    \begin{equation}
    \frac{\delta}{2}\equiv\text{Pr}( \hat{p}_4  \le  p_4-\epsilon)\le \exp\left(-\frac{\epsilon^2 L}{2}\right)\,,
\end{equation}
where $\delta$ is the failure probability, i.e.  estimator $\hat{p}_4$ deviating by more than $\epsilon>0$ from the mean $p_4$.
Now, we assume we stay within $\epsilon$, i.e. $\hat{p}_4\geq p_4-\epsilon$, which occurs with probability $1-\delta/2$. Then we have 
\begin{equation}
    -\ln(\hat{p}_4+\epsilon)\leq -\ln(p_4)
\end{equation}
where this bound is fulfilled with probability $1-\delta/2$.
Similarly,  assuming $\hat{p}_2>\epsilon$, we can bound $p_2$ via
\begin{equation}
\ln(\hat{p}_2-\epsilon)\leq \ln(p_2)\,,
\end{equation}
 with the same probability $1-\delta/2$.
Putting both bounds together and inserting into $d \geq \frac{1}{\vert \partial A\vert}\mathcal{E}_4(\rho)$, we get
\begin{equation}
    d\geq \frac{1}{2\vert\partial A\vert \ln 2}[-\ln(\hat{p}_4+\epsilon)+3\ln(\hat{p}_2-\epsilon)]\,,
\end{equation}
where this bound is fulfilled with failure probability $\delta-\delta^2/4$, where $\delta=2\exp(-\frac{\epsilon^2 L}{2})$. 
Now, we choose $L=n^{2c+1}$, $\epsilon=n^{-c}$ and $c>0$, which concludes our proof.
\end{proof}

For $S_2=O(\log n)$, we can efficiently measure non-trivial bounds on $d$. When we have $\mathcal{E}_4=\omega(\log n)$, the algorithm will correctly determine $d=\omega(\log n)$. 

Note that when we have $S_2=\omega(\log n)$ and a polynomial number of measurements $L=\text{poly}(n)$, we only get a trivial bound on the circuit depth, i.e. $d\geq -\omega(\log n)$.

\section{Entanglement of pseudorandom density matrices and pseudoentanglement}\label{sec:prdm}
In this section, we show that pseudorandom density matrices (PRDMs) with $S_2=O(\log n)$ must have $\E(\rho)=\omega(\log n)$.

PRDMs are efficiently preparable states that are indistinguishable for any efficient observer from truly random mixed states~\cite{bansal2024pseudorandomdensitymatrices}. 
They generalize pseudorandom states (PRSs), which are computationally indistinguishable from Haar random states~\cite{ji2018pseudorandom}. 
PRDMs are indistinguishable from random mixed states. Formally, these random mixed states are sampled from the generalized Hilbert-Schmidt ensemble (GHSE)~\cite{braunstein1996geometry,hall1998random,Zyczkowski_2001,bansal2024pseudorandomdensitymatrices}
\begin{equation}\label{eq:GHSE}
    \eta_{n,m}=\{\operatorname{tr}_m(\ket{\psi}\bra{\psi})\}_{\ket{\psi}\in\mu_{n+m}}\,,
\end{equation}
which are states constructed by taking $n+m$ qubit random states from the Haar measure $\mu_{n+m}$ and tracing out $m$ qubits.

PRDMs are now defined as follows:
\begin{definition}[Pseudo-random density matrix (PRDM)~\cite{bansal2024pseudorandomdensitymatrices}]\label{def:PRDM}
    Let $\kappa=\operatorname{poly}(n)$ be the security parameter with keyspace $\mathcal{K}=\{0,1\}^{\kappa}$. A family of $n$-qubit density matrices $\{\rho_{k,m}\}_{k \in \mathcal{K}}$ are pseudorandom density matrices (PRDMs) with mixedness parameter $m$ if:
    \begin{enumerate}
        \item {Efficiently preparable}: There exists an efficient quantum algorithm $\mathcal{G}$ such that $\mathcal{G}(1^{\kappa}, k,m) = \rho_{k,m}$.
        \item {Computational indistinguishability}: $t=\mathrm{poly}(n)$ copies of $\rho_{k,m}$ are computationally indistinguishable (for any quantum polynomial time adversary $\mathcal{A}$) from the GHSE $\eta_{n,m}$
        \begin{equation}
            \Big{|}\Pr_{k \leftarrow \mathcal{K}}[\mathcal{A}(\rho_{k,m}^{\otimes t}) = 1] - \Pr_{\rho \leftarrow \eta_{n,m}}[\mathcal{A}(\rho^{\otimes t}) = 1]\Big{|} = \operatorname{negl}(n).
        \end{equation}
    \end{enumerate}
\end{definition}
For $m=0$, one recovers PRS~\cite{ji2018pseudorandom}, while for $m=\omega(\log n)$, PRDMs are computationally indistinguishable from the maximally mixed state~\cite{bansal2024pseudorandomdensitymatrices,haug2025pseudorandom}.
However, the case $m=O(\log n)$ has not been understood well so far.

\begin{proposition}[Lower bound on entanglement of PRDM]
    Any ensemble of PRDMs  with entropy $S_2=O(\log n)$ must have PT negativity $\E(\rho)=\omega(\log n)$ with high probability.
\end{proposition}
\begin{proof}
    We prove by contradiction: Assume there exists PRDMs with $\E(\rho)=O(\log n)$. Since $\E_4(\rho)\leq\E(\rho)$, we also have that $\E_4(\rho)=O(\log n)$. If $S_2=O(\log n)$, then $-\ln p_4=O(\log n)$. On the other hand, GHSE have $-\ln p_4=\Theta(n)$ (see Sec.~\ref{sec:pneg_haar}). Therefore, the algorithm to measure $p_4$ in Fact.~\ref{thm:entanglement_sup} can distinguish the two ensembles with polynomial number of copies, thus contradicting the definition of PRDM. Therefore, PRDMs must have $\E(\rho)=\omega(\log n)$ whenever $S_2=O(\log n)$.
\end{proof}

Finally, let us formally define pseudoentanglement. Note that our definition explicitly allows mixed states, in contrast to Ref.~\cite{bouland2022quantum} which used entanglement entropy as measure of entanglement and thus was only well defined for pure states. We extend to mixed states by using for entanglement $\mathcal{E}(\rho)$ 
and \gEnt{}:
\begin{definition}[Pseudoentanglement]
A \emph{pseudoentangled state ensemble} with gap $f(n)$ vs $g(n)$ consists of two ensembles of $n$-qubit states $\rho_k$ and $\sigma_k$, indexed by a key $k\in\{0,1\}^{\mathrm{poly}(n)}$ with the following properties:
\begin{enumerate}
\item \emph{Efficient Preparation}: Given $k$, $\rho_k$ (or $\sigma_k$, respectively) is efficiently preparable by a uniform, poly-sized quantum circuit.

\item \emph{Pseudoentanglement}: With probability $\geq 1 - 1/\mathrm{poly}(n)$ over the choice of $k$, the PT negativity $\mathcal{E}(\rho)$ 
and \gEnt{} $D_\mathrm{F}(\rho)$
for $\rho_k$ (or $\sigma_k$, respectively) is $\Theta(f(n))$ (or $\Theta(g(n))$, respectively).

\item \emph{Indistinguishability}: For any polynomial $p(n)$, no poly-time quantum algorithm can distinguish between the ensembles of $\mathrm{poly}(n)$ copies 
with more than negligible probability. That is, for any poly-time quantum algorithm $A$, we have that
\[\left| \Pr_{k} [A(\rho_k^{\otimes \mathrm{poly}(n)}) = 1] - \Pr_{k} [A(\sigma_k^{\otimes \mathrm{poly}(n)}) = 1] \right| = \operatorname{negl}(n)\,.\]
\end{enumerate}
\end{definition}

\section{Schmidt rank testing}\label{sec:schmidtrank}
Here, we present an efficient algorithm to test the Schmidt rank of a given pure state. Consider the Schmidt decomposition of a pure $n$-qubit state $\ket{\psi}$ with respect to bipartition $A$, $B$
\begin{equation} \label{eq:schmidt_decomp}
    \ket{\psi} = \sum_{i=1}^R \sqrt{\lambda_i} \ket{\phi_i}_A \otimes \ket{\phi_i}_{B},
\end{equation}
where $ \ket{\phi_i}_A$ and  $\ket{\phi_i}_B$ are normalized states on subsystem $A$ and $B$, $\lambda_i\geq0$ are the Schmidt coefficients with $\lambda_1 \geq \cdots \geq \lambda_R$, $\sum_i \lambda_i=1$ and Schmidt rank $R\leq 2^n$. The maximum overlap of $\ket{\psi}$ with the set of states $M_r=\{\ket{\eta} : \ket{\eta} = \sum_{i=1}^r \sqrt{\lambda_i} \ket{\phi_i}_A \otimes \ket{\phi_i}_{B}\}$ of at most Schmidt rank $r$ is given by the Eckart-Young theorem:~\cite{eckart1936approximation}
\begin{equation}
\mathcal{F}_r(\ket{\psi})=\max_{\ket{\eta}\in M_r}\vert \braket{\eta}{\psi}\vert^2=\sum_{i=1}^r \lambda_i\,.
\end{equation}
Next, we prove the following inequality
\begin{equation} \label{eq:bounds_Fr}
    \text{tr}(\rho_A^2) \leq  \mathcal{F}_r(\ket{\psi}) \leq \sqrt{r \text{tr}(\rho_A^2) }.
\end{equation}
The lower bound follows from $\text{tr}(\rho_A^2) \leq  \mathcal{F}_1(\ket{\psi}) \leq \mathcal{F}_r(\ket{\psi})$. For the upper bound, we use the Jensen's inequality $(\sum_{i=1}^r \lambda_i)^2/r \leq  \sum_{i=1}^r \lambda_i^2 \leq \sum_{i=1}^R \lambda_i^2$. With this bound, we can now provide an efficient algorithm to test the Schmidt rank by measuring the purity $\text{tr}(\rho_A^2)$.

\begin{theorem}[Efficient testing of Schmidt rank] \label{thm:schmidt_rank_test_sup}
Let $\ket{\psi}$ be an $n$-qubit state where it is promised that 
\begin{align*}
\mathrm{either}\quad (a)& \,\,\mathcal{F}_r(\ket{\psi})\geq \epsilon_1 \,,\\
\mathrm{or}\quad (b)& \,\, \mathcal{F}_r(\ket{\psi})\leq \epsilon_2\,,
\end{align*} 
\revA{across a given bipartition}, where it is assumed that $\epsilon_1^2/r> \epsilon_2+2\epsilon$.
Assuming $\epsilon=1/\mathrm{poly}(n)$, there exists an efficient quantum algorithm to distinguish case ($a$) and ($b$) using $O(r^2)$ two-copy measurements of $\ket{\psi}$ with high probability. 
\end{theorem}

\begin{proof}
    We set $\epsilon<\frac{1}{2}(\frac{\epsilon_1^2}{r}-\epsilon_2)$ and $\epsilon_t=\frac{1}{2}(\frac{\epsilon_1^2}{r}+\epsilon_2)$. By Hoeffding's inequality, the purity $p_2=\text{tr}(\rho_A^2)$ can be measured efficiently using SWAP test or Bell measurement with $L\leq \frac{2}{\epsilon^2}\ln(\frac{2}{\delta})$ number of samples, such that $|p_2-\hat{p}_2|<\epsilon$ with failure probability $\delta$. We then make the following decision: if $\hat{p}_2\geq\epsilon_t$ then we output (a), otherwise (b). 

    If we are in case (a), then the upper bound in Eq.~\eqref{eq:bounds_Fr} implies $\sqrt{r p_2} \geq \mathcal{F}_r(\ket{\psi}) \geq \epsilon_1$, so that $p_2 \geq \epsilon_1^2/r$. If $\hat{p}_2<\epsilon_t$ then $p_2-\hat{p}_2 > \frac{\epsilon_1^2}{r}-\epsilon_t > \epsilon$, which occurs with probability at most $\delta/2$. If we are in case (b), then the lower bound in Eq.~\eqref{eq:bounds_Fr} implies $p_2 \leq \mathcal{F}_r(\ket{\psi}) \leq \epsilon_2$. If $\hat{p}_2\geq\epsilon_t$ then $\hat{p}_2 -p_2 \geq \epsilon_t -\epsilon_2 > \epsilon$, which occurs with probability at most $\delta/2$. Therefore, the algorithm efficiently distinguishes case (a) and (b) with probability $1-\delta/2$. Since $\epsilon<\frac{\epsilon_1^2}{r}$, the algorithm requires $O(r^2 \epsilon_1^{-4})$ number of copies.
    

\end{proof}

A few remarks are in order. First, since the Schmidt rank of a pure state is equivalent to the rank of the reduced density matrix, our algorithm can also be used to test the rank of a mixed state. Second, we note that our algorithm imposes non-trivial condition on $\epsilon_1$ and $\epsilon_2$, thus it only works in certain parameter regimes. Nevertheless, it requires only two-copy measurements of the state, which is an essential advantage for practical applications. Previously, an algorithm to test the rank in the one-sided error setting (i.e. $\epsilon_1=1$) was proposed in Ref.~\cite{ODonnell2015} using $O(r^2)$ copies. This test used a technique known as weak Schur sampling~\cite{childs2007weakschur}, which requires entangled measurements on many copies of the state. Subsequently, Ref.~\cite{lovitz2024nearlytightbounds} proposed a test with $(r+1)$-copy measurements using $\Theta((r+1)!)$ copies, again in the
setting of one-sided error.

\section{Operator Schmidt rank testing}\label{sec:operator_schmidt_rank}
 Here, we present an efficient algorithm to test the operator Schmidt rank of a given mixed state. Consider the operator Schmidt decomposition of a mixed $n$-qubit state $\rho$ with respect to bipartition $A$, $B$
\begin{equation}
    \frac{\rho}{\sqrt{\text{tr}(\rho_A^2)}} = \sum_i \sqrt{\lambda_i^O} O_{A,i} \otimes O_{B,i},
\end{equation}
where $\{O_{A,i}\}$ and $\{O_{B,i}\}$ are a set of orthonormal operators in the subsystem $A$ and $B$, $\lambda_i^O\geq0$ are the operator Schmidt coefficients with $\lambda_1^O \geq \cdots \geq \lambda_R^O$, $\sum_i \lambda_i^O=1$ and operator Schmidt rank $R\leq 4^n$. The operator Schmidt decomposition can be obtained by vectorizing the normalized operator $\ket{\rho}=\frac{\rho}{\sqrt{\text{tr}(\rho_A^2)}}$ and performing the Schmidt decomposition in Eq.~\eqref{eq:schmidt_decomp} on $\ket{\rho}$. The maximum pure-state overlap of $\ket{\rho}$ with the set of operators $M_r^O=\{\ket{O} : \frac{O}{\sqrt{\text{tr}(O_A^2)}} = \sum_i \sqrt{\lambda_i^O} O_{A,i} \otimes O_{B,i}\}$ whose operator Schmidt rank is at most $r$ is given by
\begin{equation}
\mathcal{F}_r^O(\rho)=\max_{\ket{O}\in M_r^O}\vert \braket{O}{\rho}\vert^2= \max_{\ket{O}\in M_r^O} \frac{\text{tr}(\rho O)^2}{\text{tr}(\rho_A^2)\text{tr}(O_A^2)} \,.
\end{equation}
Note that $\sqrt{\mathcal{F}_r^O(\rho)}$ is equivalent to the mixed-state fidelity introduced in Ref.~\cite{wang2008alternative}, which satisfies Jozsa's axioms for mixed-state fidelities~\cite{jozsa1994fidelity} up to a normalization factor. Following the pure-state case, we have
\begin{equation}
\mathcal{F}_r^O(\rho)= \sum_{i=1}^r \lambda_i^O\,.
\end{equation}
Next, we recall that the purity of the state $\ket{\rho}$ is given by $r_4 / r_2^2$. From Eq. \eqref{eq:bounds_Fr}, we thus have the inequality
\begin{equation} \label{eq:bounds_Fr_O}
    \frac{r_4}{r_2^2} \leq  \mathcal{F}_r^O(\rho) \leq \sqrt{\frac{rr_4}{r_2^2} }.
\end{equation}
With this bound, we can now provide an efficient algorithm to test the operator Schmidt rank by measuring $r_4$ and $r_2$.

\begin{theorem}[Efficient testing of operator Schmidt rank]
Let $\rho$ be an $n$-qubit state where it is promised that 
\begin{align*}
\mathrm{either}\quad (a)& \,\,\mathcal{F}_r^O(\rho)\geq \epsilon_1 \,,\\
\mathrm{or}\quad (b)& \,\, \mathcal{F}_r^O(\rho)\leq \epsilon_2\,,
\end{align*} 
\revA{across a given bipartition}, where it is assumed that $\operatorname{tr}(\rho_A^2)=1/\mathrm{poly}(n)$ and  $(\epsilon_1^2/r- \epsilon_2)\operatorname{tr}(\rho_A^2)^2>1/\mathrm{poly}(n)$. Then, there exists an efficient quantum algorithm to distinguish case ($a$) and ($b$) using $O(r^2)$ copies of $\ket{\psi}$ with high probability. 
\end{theorem}

\begin{proof}
    The algorithm works as follows: We measure $r_2^2=\text{tr}(\rho_A^2)^2$ using Bell measurements and $r_4$ using the algorithm in Appendix~\ref{sec:CCNRmeas} with $L=O(\epsilon^{-2}\log(1/\delta))$ number of copies, such that $|r_2^2-\hat{r}_2|<\epsilon$ and $|r_4-\hat{r}_4|<\epsilon$, each with failure probability $\delta$. The parameter $\epsilon$ will be fixed later. We then make the following decision: if $\hat{r}_4/\hat{r}_2\geq\epsilon_t$, where $\epsilon_t$ is a parameter to be fixed later, then we output (a), otherwise (b). 

    If we find $\hat{r}_4/\hat{r}_2\geq\epsilon_t$, we need that $\mathcal{F}_r^O(\rho)> \epsilon_2$. We have
    \begin{equation}
        \mathcal{F}_r^O(\rho) \geq \frac{r_4}{r_2^2}  > \frac{\hat{r}_4 -\epsilon}{\hat{r}_2+\epsilon} \geq \frac{\epsilon_t-\epsilon/\hat{r}_2}{1+\epsilon/\hat{r}_2} ,
    \end{equation}
    with failure probability $\delta^2/4$. We now impose the condition
    \begin{equation} \label{eq:condition_et_1}
        \frac{\epsilon_t-\epsilon/\hat{r}_2}{1+\epsilon/\hat{r}_2}  >\epsilon_2.
    \end{equation}
    Next, if we find $\hat{r}_4/\hat{r}_2<\epsilon_t$, we need that $\mathcal{F}_r^O(\rho)<\epsilon_1$. We have
    \begin{equation}
        \mathcal{F}_r^O(\rho) \leq \sqrt{\frac{rr_4}{r_2^2}}   < \sqrt{\frac{r(\hat{r}_4 +\epsilon)}{\hat{r}_2-\epsilon}} < \sqrt{\frac{r(\epsilon_t+\epsilon/\hat{r}_2)}{1-\epsilon/\hat{r}_2}},
    \end{equation}
    with failure probability $\delta^2/4$. Here, we must have that $\hat{r}_2-\epsilon>0$. We further impose the condition
    \begin{equation}
    \label{eq:condition_et_2}\sqrt{\frac{r(\epsilon_t+\epsilon/\hat{r}_2)}{1-\epsilon/\hat{r}_2}}  <\epsilon_1.
    \end{equation} 
    Combining the two inequalities in Eq. \eqref{eq:condition_et_1} and Eq. \eqref{eq:condition_et_2}, we obtain 
    \begin{equation}
        \frac{\epsilon}{\hat{r}_2}+\epsilon_2 (1+\frac{\epsilon}{\hat{r}_2}) < \epsilon_t < \frac{\epsilon_1^2}{r} (1-\frac{\epsilon}{\hat{r}_2}) - \frac{\epsilon}{\hat{r}_2}.
    \end{equation}
    Therefore, assuming $\epsilon_1^2/r> \epsilon_2$, we can choose
    \begin{equation}
        \epsilon < \frac{\hat{r}_2}{2+\epsilon_1^2/r+\epsilon_2}\left(\frac{\epsilon_1^2}{r}- \epsilon_2\right)
    \end{equation}
    and 
    \begin{equation}
        \epsilon_t = \frac{1}{2}\left(\frac{\epsilon_1^2}{r}+ \epsilon_2\right) - \frac{\epsilon}{2\hat{r}_2}\left(\frac{\epsilon_1^2}{r}- \epsilon_2\right) .
    \end{equation}
    Therefore, the algorithm requires $O(r^2 r_2^{-4}\epsilon_1^{-4})$ number of copies, which is polynomial in $n$ assuming $r_2=1/\mathrm{poly}(n)$.
\end{proof}

\section{$p_\alpha$-negativity in Haar random states} \label{sec:pneg_haar}
Here, we compute the Haar average of $p_\alpha$-negativity as a function of the partition sizes. The computation for $\E_4$ has been given in the main text, and here we discuss $\E_\alpha$ for $\alpha\geq0$. 

We recall an upper bound for $\mathbb{E}_{\text{H}}[\E(\rho)]$, which we report here for convenience
\begin{align}\label{eq:upp_bound_haar_2}
    \mathbb{E}_{\text{H}}[\E(\rho)] \leq &\min\{n_A \ln{2},n_B\ln{2},
    \frac{1}{2}(n_{AB}\ln{2}-\min(n_{AB}\ln{2},n_C\ln{2}))\}.
\end{align} 

We will next compute ${\mathbb E}_{\text{H}}[\E_\alpha(\rho)]$ for even integer $\alpha$, using the results of Ref.~\cite{shapourian2021entanglement} for Haar averages of the PT moments. In the PPT phase, we have ${\mathbb E}_{\text{H}}[p_\alpha]\simeq L_{AB}^{1-\alpha}$ in the thermodynamic limit~\cite{shapourian2021entanglement}, which implies ${\mathbb E}_{\text{H}}[-\ln{p_\alpha}]\gtrsim(\alpha-1)n_{AB}\ln{2}$. We also have $-\ln{p_2}=S_2(\rho_{AB})\leq n_{AB}\ln{2}$. It follows that ${\mathbb E}_{\text{H}}[\E_\alpha(\rho)]={\mathbb E}_{\text{H}}[-\frac{1}{\alpha-2}\ln{p_\alpha}+\frac{\alpha-1}{\alpha-2}\ln{p_2}]\gtrsim 0$. Combining ${\mathbb E}_{\text{H}}[\E_\alpha(\rho)]\leq{\mathbb E}_{\text{H}}[\E(\rho)]$ and the third upper bound in Eq.~\eqref{eq:upp_bound_haar_2}, we obtain 
\begin{equation} \label{eq:pneg_haar_ppt}
{\mathbb E}_{\text{H}}[\E]\simeq{\mathbb E}_{\text{H}}[\E_\alpha]\simeq0.    
\end{equation}

In the ES phase, where $n_C<n_{AB}$ and both $n_A<N/2$ and $n_B<N/2$, one gets in the thermodynamic limit ${\mathbb E}_{\text{H}}[p_\alpha]\simeq
\displaystyle{\frac{C_k L_{AB}}{(L_{AB} L_C)^{\alpha/2}}}$ for even integer $\alpha$, where $C_k=\binom{2k}{k}/(k+1)$ is the $k$th Catalan number. Thus, we have ${\mathbb E}_{\text{H}}[-\ln{p_\alpha}]\gtrsim (\frac{\alpha}{2}-1)n_{AB}\ln{2}+\frac{\alpha}{2}n_C\ln{2}$. Since $-\ln{p_2}\leq n_C\ln{2}$, we have ${\mathbb E}_{\text{H}}[\E(\rho)]\geq{\mathbb E}_{\text{H}}[\E_\alpha(\rho)]\gtrsim \frac{1}{2}(n_{AB}-n_C)\ln{2}$. Combining with the third upper bound in Eq.~\eqref{eq:upp_bound_haar_2}, we obtain 
\begin{equation} \label{eq:pneg_haar_es}
    {\mathbb E}_{\text{H}}[\E]\simeq{\mathbb E}_{\text{H}}[\E_\alpha]\simeq  \frac{1}{2}(n_{AB}-n_C)\ln{2}.
\end{equation}

Finally, in the ME phase, when $n_{AB}>n_C$ and $n_A>N/2$, we obtain in the thermodynamic limit
${\mathbb E}_{\text{H}}[p_\alpha]\simeq
L_C^{1-\alpha}L_{B}^{2-\alpha}$, for even integer $\alpha$. Thus, we have ${\mathbb E}_{\text{H}}[-\ln{p_\alpha}]\gtrsim (\alpha-2)n_B\ln{2}+(\alpha-2)n_C\ln{2}$. We again obtain ${\mathbb E}_{\text{H}}[\E(\rho)]\geq{\mathbb E}_{\text{H}}[\E_4(\rho)]\gtrsim n_B\ln{2}$. The case in which $n_{AB}>n_C$ and $n_B>N/2$ is obtained by replacing $L_{B}$ with $L_{A}$ in the latter formula. Combining with the first and second upper bound in Eq.~\eqref{eq:upp_bound_haar_2}, we obtain 
\begin{equation} \label{eq:pneg_haar_me}
    {\mathbb E}_{\text{H}}[\E]\simeq{\mathbb E}_{\text{H}}[\E_\alpha]\simeq  \min(n_A,n_B)\ln{2}.  
\end{equation}

Finally, using the monotonicity of the $p_\alpha$-negativity, the results in Eqs.~\eqref{eq:pneg_haar_ppt},~\eqref{eq:pneg_haar_es}, and ~\eqref{eq:pneg_haar_me}, extend to any $\alpha\geq1$. This also implies that the absolute values of the PT spectrum of Haar random states is (nearly) flat, thus the results extend to any $\alpha\geq0$. We confirm this numerically in Fig.~\ref{fig:neg_haar}, where $\E_\alpha$ for different $\alpha$ are close to the predicted behavior already at small system sizes. Note that, in the ME phase we find numerically that $p_0=2^{2m+n_C}$ for any realization, where $m=\min(n_A,n_B)$. While the PT is not full-rank, it saturates the bound in Eq.~\eqref{eq:bound_p0}. Instead, in the ES phase and the PPT phase, we find $p_0=2^{n_{AB}}$, so that the PT is full-rank.

\begin{figure}
\centering
\subfigimg[width=0.39\linewidth]{a}{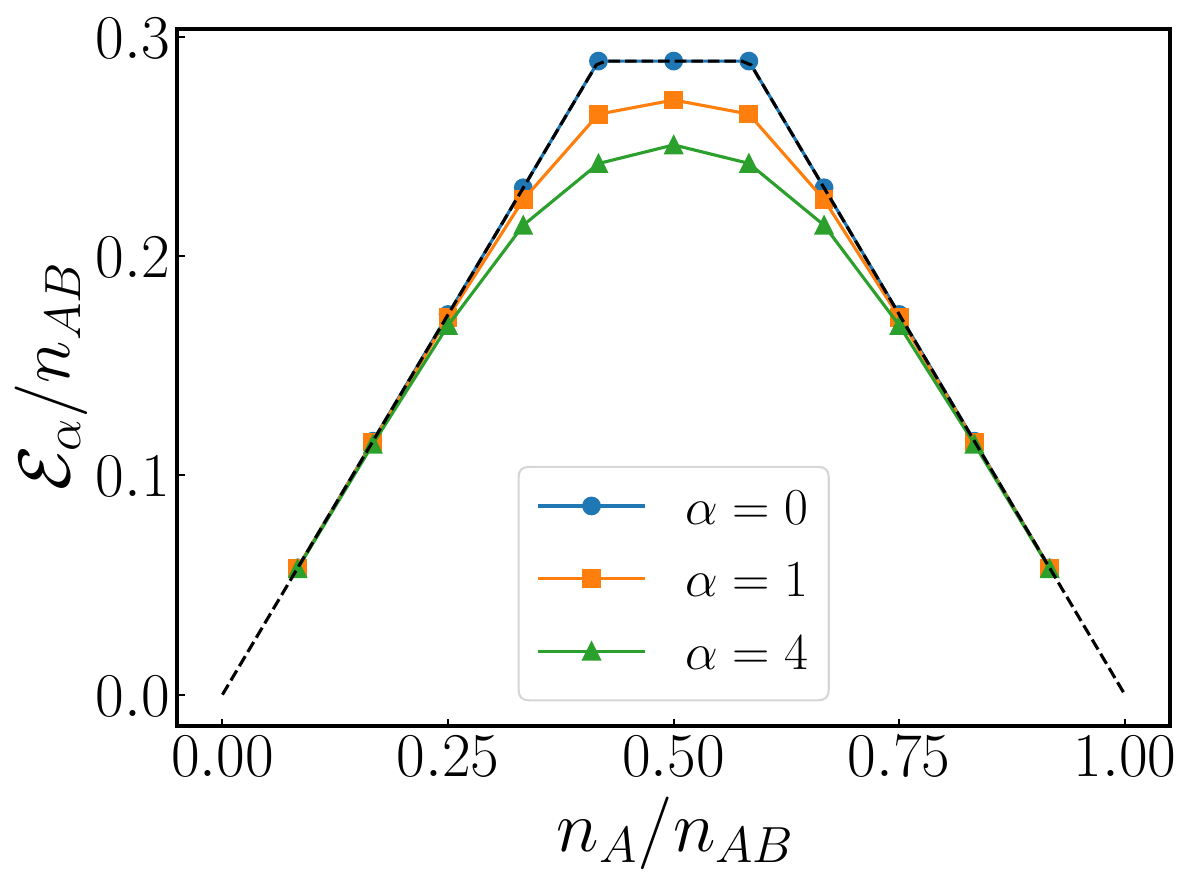}
\subfigimg[width=0.39\linewidth]{b}{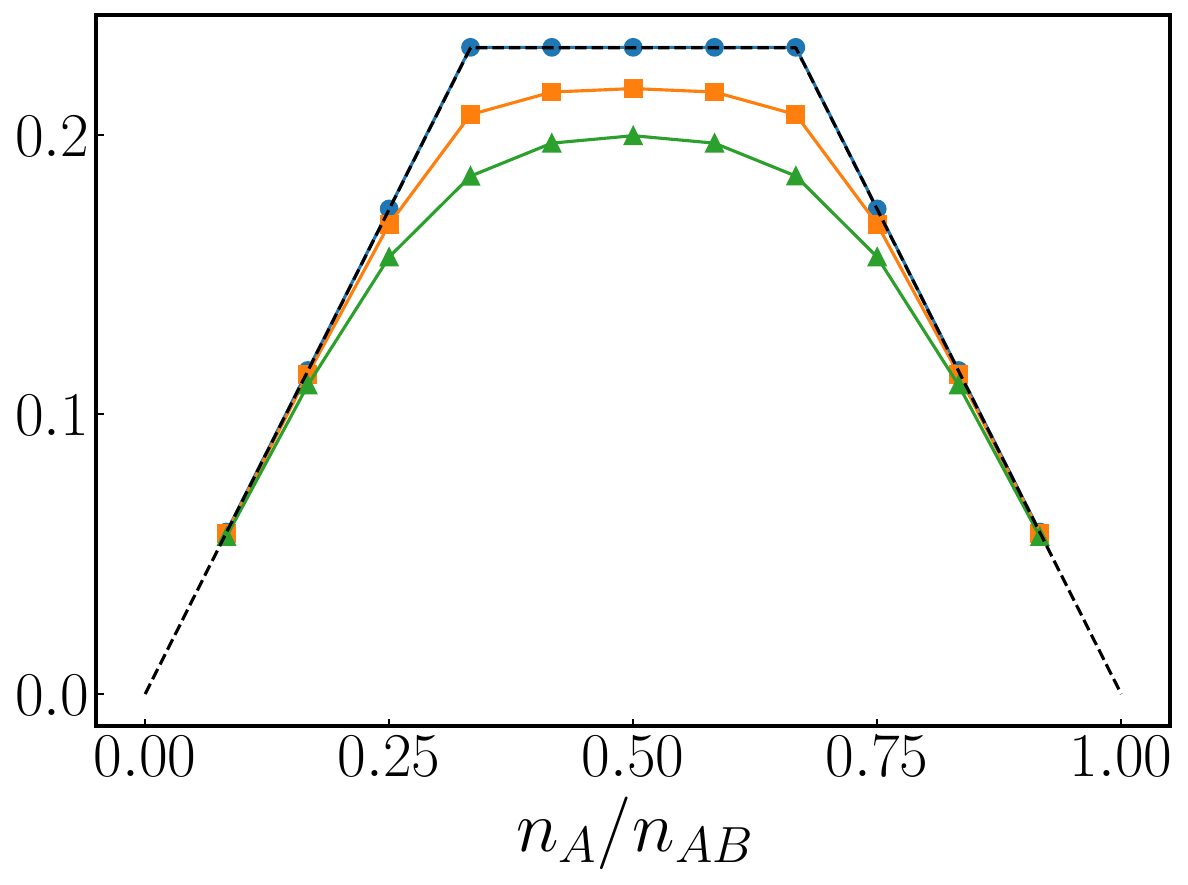}
\caption{$p_\alpha$-negativity for Haar random states. The system size is (a) $n_{AB}=12$ sites and $n_C=2$ and (b) $n_{AB}=12$ sites and $n_C=4$. Each data point is averaged over 100 realizations. The dashed lines show the analytical prediction in the thermodynamic limit.}
\label{fig:neg_haar}
\end{figure}

\section{PT and realignment moments in translation-invariant MPS} \label{sec:ti_mps}

Let us consider a translationally invariant (TI) MPS of $n$ qubits of bond dimension $\chi$:
\begin{equation} \label{eq:ti_mps}
\ket{\psi}=\sum_{s_1,s_2,\cdots,s_n} \text{tr}\left(A^{s_1} A^{s_2} \cdots A^{s_n} \right)|s_1,s_2,\cdots s_n \rangle
\end{equation}
with $A^{s}$ being a $\chi \times \chi$ matrix. We consider the entanglement within $\rho=\text{tr}_C(\ket{\psi}\bra{\psi})$, where $C$ is a contiguous subsystem of $\ket{\psi}$. In particular, we study entanglement for $\rho$ with bipartition $A=\{1,2,\cdots,n_A\}, B=\{n_A+1,n_A+2,\cdots,n_A+n_B\}$, and we assume that all subsystems are extensive.

The transfer matrix of the MPS is defined as
\begin{equation} \label{eq:tm_mps}
    E = \sum_s A^s \otimes \overline{A}^s.
\end{equation}
We denote the eigenvalues of $E$ as $\lambda_i$, with $\lambda_1\geq\lambda_2\geq\cdots\lambda_{r}$. The corresponding left and right eigenvectors, $L_i$ and $R_i$, satisfy $\sum_s (A^s)^\dagger L_i A^s=\lambda_1 L$ and $\sum_s A^s R_i (A^s)^\dagger=\lambda_i R$, respectively. Here, we impose $\text{tr}(L_iR_i)=1$. The matrices $L_i$ and $R_i$ are Hermitian with non-negative eigenvalues. We assume that the dominant eigenvalue $\lambda_1$ is unique. 

The PT moment can be written as 
\begin{equation}
    p_\alpha = \frac{\sum_{ijk} \lambda_i^{\alpha n_A} \lambda_j^{\alpha n_B} \lambda_k^{\alpha n_C} \text{tr}((L_kR_i)^\alpha) \text{tr}((L_jR_k)^\alpha)  \text{tr}((L_iR_j)^{\alpha/2})^2}{\sum_i \lambda_i^{\alpha n}},
\end{equation}
for even $\alpha$. In the thermodynamic limit $n\to\infty$, we find 
\begin{equation}
    p_\alpha = \text{tr}((L_1R_1)^\alpha)^2  \text{tr}((L_1R_1)^{\alpha/2})^2 + O(\lvert\frac{\lambda_2}{\lambda_1}\rvert^{\alpha m}),
\end{equation}
where $m=\min(n_A,n_B,n_C)$. 
 This yields the PT negativity
\begin{equation}
    \E(\rho) = \lim_{\alpha\to1/2}\ln p_{2\alpha} \simeq 2\ln \text{tr}((L_1R_1)^{1/2}),
\end{equation}
where the subleading term decays exponentially with $n$.
Here, $\lim_{\alpha\to1/2}\text{tr}((L_1R_1)^{2\alpha})=1$ since the eigenvalues of $L_1R_1$ are non-negative. Note that one can also find that $S_{1/2}(\rho_A)\simeq 4\ln \text{tr}((L_1R_1)^{1/2})$ for $n\to\infty$, where $\rho_A=\text{tr}_B(\rho)$. This implies that $\E(\rho)\simeq\frac{1}{2}S_{1/2}(\rho_A)$. Note that a similar observation was made in quantum circuits at short times~\cite{bertini2022entanglement} 

For the realignment moment, we find for even $\alpha$
\begin{equation}
    r_\alpha = \frac{\sum_{ijk} \lambda_i^{\alpha n_A} \lambda_j^{\alpha n_B} \lambda_k^{\alpha n_C} \text{tr}((L_kR_i)^2)^{\alpha/2} \text{tr}((L_jR_k)^2)^{\alpha/2}  \text{tr}((L_iR_j)^{\alpha/2})^2}{\sum_i \lambda_i^{\alpha n}},
\end{equation}
which in the limit $n\to\infty$ gives
\begin{equation}
    r_\alpha = \text{tr}((L_1R_1)^2)^\alpha \text{tr}((L_1R_1)^{\alpha/2})^2 + O(\lvert\frac{\lambda_2}{\lambda_1}\rvert^{\alpha m}).
\end{equation}
Thus, the CCNR negativity is given by
\begin{equation}
    \mathcal{C}(\rho) = \lim_{\alpha\to1/2}\ln r_{2\alpha} \simeq 2\ln \text{tr}((L_1R_1)^{1/2}) + \ln \text{tr}((L_1R_1)^{2}) = \frac{1}{2}(S_{1/2}(\rho_A) -S_2(\rho_A)).
\end{equation}
We see that $\mathcal{C}(\rho)$ is smaller than $\E(\rho)$, and thus provides a weaker entanglement witness for any TI MPS. This is similarly the case for stabilizer states and Haar random states~\cite{aubrun2012realigning}. Nevertheless, $\mathcal{C}(\rho)$ is always successful in detecting entanglement in the generic case when the spectrum of $L_1R_1$ is not flat. 

\section{Matrix product states algorithm for the PT negativity and the CCNR negativity} \label{sec:mps_pt_left_right}

Given MPS $\ket{\psi}$, we consider the entanglement within $\rho=\text{tr}_C(\ket{\psi}\bra{\psi})$, where $C$ is a contiguous subsystem in the middle of the chain, $C=\{n_A+1,n_A+2,\cdots,n_A+n_C\}$, and we take the subsystem $A=\{1,2,\cdots,n_A\}$ and $B$ is the complement of $A$ in $\rho$, i.e. the subsystem connected to the right edge. Both the PT negativity and the CCNR can be efficiently computed in this scenario. The method for the PT negativity was given in Ref.~\cite{ruggiero2016entanglement}.

We first place the orthogonality center at a site $i\in C$. Then, we define 
\begin{equation} \label{eq:E_C}
    E_C = \prod_{i\in C} E_i
\end{equation}
where $E_i$ is the transfer matrix at site $i$ as defined in Eq.~\eqref{eq:tm_mps}. The graphical representation is shown in Fig.~\ref{fig:transfer_matrix_c}. $E_C$ can be obtained with computational cost $O(n_C \chi^5)$. We now define the partially transposed matrix $E_C^\Gamma$ which is obtained from $E_C$ by permutation of indices, $[E_C^\Gamma]_{(\alpha,\beta'),(\alpha',\beta)}=[E_C]_{(\alpha,\alpha'),(\beta,\beta')}$. The spectrum of $\rho^\Gamma$ is obtained by diagonalizing $E_C^\Gamma$ with cost $O(\chi^6)$~\cite{ruggiero2016entanglement}. Similarly, the singular values of the realigned matrix $R_\rho$ are obtained by performing SVD on $E_C$ (note that the matrix $E_C$ already represents the realigned matrix). Therefore, the total cost to compute the PT negativity and the CCNR negativity in this setup is $O(\chi^6+n_C\chi^5)$.

\begin{figure}
\centering
\includegraphics[width=.49\linewidth]{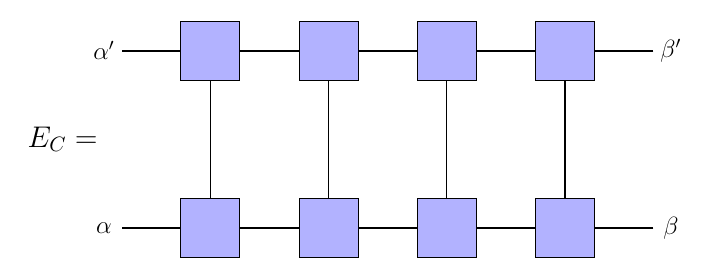}
\caption{MPS contraction to construct $E_C$ in Eq.~\eqref{eq:E_C}.}
\label{fig:transfer_matrix_c}
\end{figure}

\section{Effects of depolarizing noise on entanglement} \label{sec:depo_noise}
In this section, we consider the effects of global
depolarizing noise 
\begin{equation} \label{eq:depo_gen}
    \rho'
 = (1-p)\ket{\psi}\bra{\psi} + p I/2^n,   
\end{equation}
on the bipartite entanglement of a general pure state $\ket{\psi}$. For a given subsystem $A$, the entanglement in the pure state $(p=0)$ is characterized by the Schmidt coefficients $\lambda_i$ for $i=1,\dots,r$, where $\lambda_i\geq\lambda_j$ for $i<j$, and the Schmidt decomposition of the state is given by
\begin{equation}
    \ket{\psi} = \sum_{i=1}^{r} \sqrt{\lambda_i} \ket{i}_A \ket{i}_B.
\end{equation}
Applying the partial transpose on a pure state $\rho=\ket{\psi}\bra{\psi}$, one can show that the spectrum of $\rho^\Gamma$ is $s_{ij}\sqrt{\lambda_i \lambda_j}$ for $i,j=1,\dots,r$ and $s_{ij} = 1$ if $i\leq j$ and $s_{ij}=-1$ otherwise. From Eq.~\eqref{eq:depo_gen}, the spectrum of $(\rho')^\Gamma$ is simply rescaled and shifted to $(1-p)s_{ij}\sqrt{\lambda_i \lambda_j}+p/2^n$. Thus, the state becomes PPT when
\begin{equation}
    p \geq  1 - \frac{1}{1+2^n\sqrt{\lambda_1 \lambda_2}}.
\end{equation}
Namely, the robustness of a state with respect to depolarizing noise is determined solely by the two largest Schmidt coefficients.

 As an illustration, let us now consider stabilizer states. These states are known to have a flat entanglement spectrum, i.e., $\lambda_i=\text{const}$ for any $\lambda_i\neq 0$ ~\cite{tirrito2023quantifying}. As a consequence, the entanglement R\'enyi entropy is determined by integer $k$, with $S_\alpha=k\ln{2}$ and $\lambda_i=2^{-k}$, for $i=1,\dots,2^k$. Thus, for $k>0$, we have that the state becomes PPT when
 \begin{equation}
     p \geq 1 - \frac{1}{1+2^{n-k}}.
 \end{equation}
We see that exponentially strong noise is required to destroy entanglement, regardless of the initial value of the entanglement in the noise-free state. Surprisingly, states with less entanglement are more robust when subjected to noise. Despite of this, one can show that 
\begin{equation}
    \E_4 \simeq S_2 + \ln(1-p),
\end{equation}
and thus the entanglement is detected when $p<1-e^{-S_2}$. For low-entangled pure states, $\E_4$ already fails to detect entanglement at a low value of $p$. This is intuitively expected, since states with low entanglement are harder to distinguish from separable states. Note that a similar result holds for the $p_3$-PPT condition, where the entanglement is detected at $p<1-e^{-2S_3}$.

\section{Witnessing entanglement in Werner States} \label{sec:werner}

Werner states are bipartite quantum states in a Hilbert space $\mathcal{H}_{AB}=\mathcal{H}_A \otimes \mathcal{H}_B$ with dimensions $d_A=d_B \equiv d$, defined as 
\begin{equation}
	\rho_W= \alpha \binom{d+1}{2}^{-1} \Pi_+ + (1-\alpha) \binom{d}{2}^{-1}\Pi_-
	\label{eq:werner}
\end{equation}
with $\alpha \in [0,1]$ and $\Pi_\pm=\frac{1}{2}\left( \mathbb{I}\pm \Pi_{12} \right)$ projectors onto symmetric $\mathcal{H}_+$ and anti-symmetric $\mathcal{H}_-$ subspaces of $\mathcal{H}=\mathcal{H}_+ \oplus \mathcal{H}_-$, respectively. Here, $\Pi_{12}=\sum_{i,j=1}^d\ket{i}\bra{j} \otimes \ket{j}\bra{i}$ is the swap operator. Using that $\Pi^{T_A}_{\pm}= 1/2(\Delta_1 \pm (d\pm 1) \Delta_0)$ with $\Delta_0=\ket{\phi_+}\bra{\phi_+}$ being a projector onto the maximally entangled state and $\Delta_1=\mathbb{I}-\Delta_0$, we have
\begin{equation}
	\rho^{T_A}_W= \frac{2\alpha-1 }{d}  \Delta_0 + \frac{1 +d -2\alpha}{d } \frac{\Delta_1}{d^2-1 }
\end{equation}
with eigenvalues  $\lambda_0=(2\alpha-1)/{d}$ with multiplicity 1 and  $\lambda_1=({1 +d -2\alpha})/{d (d^2-1) }$ with multiplicity $d^2-1$.

We see that, for any $d$, $\lambda_0<0$ for  $0\leq \alpha < 1/2 $. Thus, according to the PPT condition, $\rho_W$ is entangled for $0\leq \alpha < 1/2$. It has been shown that the entanglement can be detected using the $p_3$-PPT condition in the whole interval $0\leq \alpha < 1/2 $, for any local dimension $d$~\cite{elben2020mixed}. We can also show, by explicit computation, that the result holds true for the $p_4$-negativity, i.e. 
\begin{equation}
	\E_4(\rho_W) > 0 \quad \text{for} \quad
	0\leq \alpha < \frac{1}{2}  
\end{equation}
for any local dimension $d$. Moreover, $\E_4(\rho_W) = 0$ for $\alpha=1/2$. From the derivation of the monotonicity of the $p_\alpha$-negativity, we can now understand that this is a consequence of the fact that the spectrum of $\rho^{T_A}_W$ is almost flat, and it becomes exactly flat at $\alpha=1/2$.

Finally, we note that Werner states have non-positive PT-moments in an interval $\left[0,\alpha^*\right)$ for $d>3$~\cite{elben2020mixed}.
This highlights that the $p_3$-negativity, and a related quantity $\Tilde{R}_3=-\log_2 (p_3/\tr[\rho^3])$, is not always well-defined. 

\end{document}